\documentclass[a4paper,allowtoday,twocolumn,11pt,reprint,superscriptaddress,nofootinbib,floatfix,accepted=2025-09-18]{quantumarticle}
\pdfoutput=1

\usepackage[unicode=true]{hyperref}
\usepackage{graphicx}
\usepackage{svg}
\usepackage{float}
\usepackage{tabularx,booktabs}
\usepackage{tcolorbox}
\usepackage[inline]{enumitem}
\usepackage{fancybox}
\usepackage{nicefrac}
\usepackage{circledsteps}
\usepackage[numbers,sort&compress]{natbib}

\usepackage[utf8]{inputenc}
\usepackage[T1]{fontenc}

\usepackage{orcidlink}
\usepackage{makecell}

\usepackage{amsmath}
\usepackage{amsthm}
\usepackage{bm}
\newtheorem{theorem}{Theorem}
\newtheorem{lemma}{Lemma}
\newtheorem*{proposition*}{Proposition}
\theoremstyle{remark}

\newtheorem{corollary}{Corollary}
\newtheorem{remark}{Remark}

\usepackage{mathrsfs}
\usepackage{physics}
\usepackage{mathtools}
\usepackage{mathbbol}
\DeclareSymbolFontAlphabet{\mathbbol}{bbold}
\usepackage{amssymb}
\DeclareSymbolFontAlphabet{\mathbb}{AMSb}
\usepackage{dsfont}

\usepackage{soul}
\usepackage{comment}

\newcommand{\SLD}{\mathrm{C}_{\mathrm{SLD}}}
\newcommand{\RLD}{\mathrm{C}_{\mathrm{RLD}}}
\newcommand{\CH}{\mathrm{C}_{\mathrm{HCRB}}}
\newcommand{\CNH}{\mathrm{C}_{\mathrm{NHCRB}}}

\newcommand{\CGM}{\mathrm{C}_{\mathrm{GMCRB}}}
\newcommand{\MI}{\mathrm{C}_{\mathrm{MI}}}

\newcommand{\rat}{\mathcal{R}^{\mathrm{NH}}}
\newcommand{\ratGM}{\mathcal{R}^{\mathrm{GM}}}
\newcommand{\ratMI}{\mathcal{R}^{\mathrm{MI}}}
\newcommand{\B}{\mathcal{B}}

\newcommand{\bigtrace}{\mathbb{T}\mathrm{r}}
\newcommand{\ctrace}{\mathsf{Tr}}
\newcommand{\qtrace}{\Tr}
\newcommand{\nmax}{n_\mathrm{max}}

\def\ANU{Centre for Quantum Computation and Communication Technology,
Department of Quantum Science and Technology, Australian National University, Canberra, ACT 2601, Australia}
\def\ASTAR{Quantum Innovation Centre (Q.InC), Agency for Science Technology and Research (A*STAR), 2 Fusionopolis Way, Innovis 08-03, Singapore 138634, Singapore}
\def\UEC{Graduate School of Informatics and Engineering, The University of Electro-Communications, 1-5-1 Chofugaoka, Chofu-shi, Tokyo 182-8585, Japan}
\begin{document}

\title{Holevo Cram\'{e}r-Rao bound: How close can we get without entangling measurements?}

\author{Aritra Das \texorpdfstring{\orcidlink{0000-0001-7840-5292}}{}}
\email{Aritra.Das@anu.edu.au}
\affiliation{\ANU}
\author{Lorcán~O.~Conlon \texorpdfstring{\orcidlink{0000-0002-0921-5003}}{}}
\affiliation{\ASTAR}
\author{Jun Suzuki \texorpdfstring{\orcidlink{0000-0003-1975-6003}}{}}
\affiliation{\UEC}
\author{Simon K. Yung \texorpdfstring{\orcidlink{0009-0007-1509-3958}}{}}
\affiliation{\ANU}
\affiliation{\ASTAR}
\author{Ping K. Lam  \texorpdfstring{\orcidlink{0000-0002-4421-601X}}{}}
\affiliation{\ASTAR}
\affiliation{\ANU}
\author{Syed M. Assad \texorpdfstring{\orcidlink{0000-0002-5416-7098}}{}}
\affiliation{\ASTAR}
\affiliation{\ANU}

%\date{November 3, 2024}
\date{September 25, 2025}

\begin{abstract}
In multi-parameter quantum metrology,
the resource of entanglement
can lead to an increase in efficiency
of the estimation process.
Entanglement can be used
in the state preparation stage,
or the measurement stage, or both,
to harness this advantage---here
we focus on the role of entangling measurements.
Specifically,
entangling or collective measurements
over multiple identical copies of a probe state
are known to be superior to measuring each probe individually,
but the extent of this improvement is an open problem.
It is also known that such entangling measurements,
though resource-intensive,
are required to attain the ultimate limits
in multi-parameter quantum metrology
and quantum information processing tasks.
In this work we investigate
the maximum precision improvement that
collective quantum measurements
can offer over individual measurements,
calling this the `collective quantum enhancement'.
We show that, whereas the maximum enhancement
can, in principle, be a factor of~\texorpdfstring{$n$}{n} for estimating~\texorpdfstring{$n$}{n} parameters,
this bound is not tight for large~\texorpdfstring{$n$}{n}.
Instead, our results prove an enhancement
linear in dimension of the qudit probe is possible
using collective measurements
and lead us to conjecture that
this is the maximum collective quantum enhancement
in any local estimation scenario.
\end{abstract}

\maketitle

\section{Introduction}

Over half-a-century of advances
in quantum metrology~\cite{Giovannetti2006,Giovannetti2011,Szczykulska2016}
has vastly improved our ability
to measure, sense, image, and estimate with enhanced precision~\cite{Caves1981,Giovannetti2001,Dorner2009,Kacprowicz2010,Yonezawa2012,Tsang2016}.
Of significant interest is
the multi-parameter estimation scenario~\cite{Paris2004,Hayashi2005,Paris2009,Szczykulska2016,ABGG20},
where two hall-mark quantum effects manifest
themselves, playing opposing roles.
On the one hand,
incompatibility
between the unknown parameters
of a quantum system~\cite{Yuen1973,Helstrom1974,
Hayashi1999,Sidhu2021}
hinders their simultaneous estimation
from a single copy of an unknown state~\cite{Helstrom1969,Szczykulska2016,OptiMeasRev2cite1,OptiMeasRev2cite2}.
On the other hand,
given multiple identical copies of the state,
an entangling measurement on all the copies,
called a collective (or joint) measurement~\cite{Lorcan21},
can extract more information about the parameters
than any measurement where
the copies are measured individually~\cite{Massar1995,Gill2000,Mansouri2022}.
As individual and separable measurements~\cite{Zhou2020}
can be recovered as special cases of
collective measurements,
it is clear that the latter can only lead
to precision enhancements in estimation tasks~\cite{Conlon2023,Mansouri2022},
but the extent of this improvement is a major open problem~\cite{Gill2000,Ballester2004}.
In this work, we study \emph{the maximum enhancement
collective measurements stand to offer
over individual measurements},
specifically in the context of parameter estimation and state tomography.

Despite their advantages,
collective measurements
are challenging to implement in any real
estimation scenario
and experimental demonstrations are few and far between~\cite{Hou2018,Yuan2020,Mansouri2022,Conlon2023,Conlon2023c}.
Resultantly, the ratio between the
optimal precisions attainable via
collective versus individual measurement
serves as a useful quantifier
of both the quantum advantage offered
by collective measurements,
and the utility
of performing complicated entangling measurements~\cite{PPDG25}
and expending vast amounts of resources.
If this ratio is small,
then there is not much advantage
to be gained from entangling measurements.
But even if the ratio is large,
our ability to perform the requisite measurements
might be limited, meaning that the
collective performance is just an overly optimistic goal
that is far from being achievable.

In local estimation theory,
where unknown parameters
are assumed to be close
to known true values,
collective measurements
on identical copies of a separable state
do
not offer any advantage
for estimating a single parameter~\cite{Giovannetti2006}
or multiple parameters of a pure state~\cite{Matsumoto2002}.
Beyond this, except for some simple cases,
not much is known about the optimal individual
or collective measurement strategies
or their performance relative to each other~\cite{OptiMeasRev2cite1,OptiMeasRev2cite2}.
One reason for this is that the analytic evaluation
of the optimal performance of either class of
measurements is notoriously difficult.
In fact, instead of finding the optimal measurements,
it is easier (and more common)
to evaluate bounds on their precision.
The most widely-used precision bounds
for local estimation
are quantum generalisations of
the classical Cram\'er-Rao bound
(CRB),
called quantum CRBs (QCRBs)~\cite{Helstrom1967,Helstrom1968,Yuen1973,Holevo2011,Nagaoka2005b}.
These include the quantum Fisher information
(QFI)-based CRBs~\cite{Helstrom1967,Helstrom1968,Yuen1973},
the Holevo CRB (HCRB)~\cite{Holevo2011},
the Nagaoka-Hayashi CRB (NHCRB)~\cite{Nagaoka2005a,Nagaoka2005b},
the Gill-Massar CRB (GMCRB)~\cite{Gill2000},
and the most informative CRB (MICRB)~\cite{Nagaoka2005b,Hayashi1997,HayashiOuyang2023}.

In general, QCRBs
are not always attainable,
especially in the multi-parameter setting~\cite{Szczykulska2016},
and the exploration of criteria
for their attainability~\cite{HayashiOuyang2023,ZhuThesis}
is an active albeit challenging area
of research~\cite{Conlon2025}. That said,
in the collective measurement setting,
the HCRB is known to be
attainable in the asymptotic limit
by performing collective measurements
on a large number of identical copies of
the unknown
state~\cite{Kahn2009,Yamagata2013,Yang2019,Massar1995}.
Additionally, the HCRB
can be computed efficiently
through a semi-definite program~\cite{Albarelli2019},
making it amenable to both numerical
and analytical techniques~\cite{ZhuThesis}.
In contrast,
the attainable bound for the individual measurement setting,
given by the MICRB~\cite{Nagaoka2005b},
requires a conic program~\cite{Hayashi1997,HayashiOuyang2023} that
is challenging to compute
even numerically. Further,
analytical solutions to the MICRB,
as reformulated in Ref.~\cite{HayashiOuyang2023},
are only known for either two-level systems
or single-parameter problems~\cite{HayashiOuyang2023}.
For exploring multi-parameter estimation
from arbitrary finite-dimensional mixed quantum states,
which is the problem we address in this work,
these limitations render the MICRB
a computationally-intractable choice for gauging individual precision.

As a result, other
individual-measurement precision bounds
such as the NHCRB~\cite{Nagaoka2005a,Nagaoka2005b,Lorcan21}
and the GMCRB~\cite{Gill2000,ZhuThesis},
which are more tractable computationally
but less tight,
are frequently used as
substitutes for the MICRB~\cite{Lorcan21,ZhuThesis}.
In particular,
the NHCRB
is a well-suited candidate to quantify
individual measurement precision
in multi-parameter qudit estimation.
This is because of two reasons.
First, the NHCRB is more
analytically-tractable than the MICRB,
because it can be efficiently computed
through a semi-definite program
that scales reasonably with system dimension
and number of parameters~\cite{Lorcan21};
this has, for instance, led to analytical
lower and upper bounds to the NHCRB~\cite{Suzuki2023}.
Second, the NHCRB is provably
attainable in many cases of
interest~\cite{Jun2024Qest,HayashiOuyang2023,Lorcan21},
including two-level systems~\cite{Nagaoka2005a,Nagaoka2005b},
and despite this attainability
not extending to higher dimensions,
the NHCRB is still very close to the tight MICRB
(with a gap of less than~5\% reported
for random qudit models
up to dimension~17,
see Fig.~6 of Ref.~\cite{HayashiOuyang2023}),
and serves as a good approximation to it.

Armed with these bounds,
we study how far off
the collective-optimal precision
can be from the individual-optimal one
by looking at their ratio.
Specifically, by investigating
the maximum ratio between the NHCRB
and the HCRB,
we identify situations where
collective quantum measurements
are the most advantageous.
Because the NHCRB is close
to the MICRB,
the ratio we study should be close
to the attainable maximum collective enhancement.
The majority of our results
concern the NHCRB-to-HCRB ratio,
which we represent using the symbol~$\rat$
and refer to as \textit{the ratio}.
To distinguish from this
the attainable collective enhancement,
given by the ratio between the MICRB
and the HCRB, we refer to the latter
ratio as
\textit{the true ratio}
and use the symbol~$\ratMI$.

Our first result
is preliminary and
shows
that the ratio of precisions~$\rat$
is at most equal to
the number of unknown parameters
being estimated,~$n$. This
result can be intuitively
expected and
agrees with empirical data
for estimating a few parameters
(blue line for~$n=1,2,3$ in Fig.~\ref{fig:summarymainresult})
but deviates with increasing~$n$.
To find a tighter bound
for large numbers of parameters,
we focus on state tomography, where
the number of parameters is maximal,~$n=\nmax$.
For the qubit tomography case,
we extrapolate existing results
to find a decreasing trend
of the ratio~$\ratMI$ with purity~\cite{BacuiLi}.
Motivated by this,
we propose a model of estimating
the coefficients of the
generalised Gell-Mann matrices
(GMMs)~\cite{GellMann1962},
which extend the Pauli matrices
to higher dimensions,
in mixed~$d$-dimensional qudit states~\cite{Watanabe2011}.
This ``linear GMM model''
is symmetric enough to admit
analytical results in the
full-parameter case~$(n=\nmax)$
for both~$\rat$ and~$\ratMI$.
Further, this model is
equivalent to the problem of tomography
in arbitrary orthonormal basis (ONB)
and is therefore representative of
a large class of full parameter models.
Then, via semi-definite programming
arguments, we extend our results
to tomography in non-orthonormal bases
and to the~$n<\nmax$ case.

A summary of our main
analytical and numerical results comparing
the collective-
and individual-optimal precisions
for local estimation from smooth models
on~$d$-dimensional qudits now
follows:
\begin{itemize}
    \item for any model comprising~$n$ parameters,
    the ratio~$\rat$ is at most~$n$ (blue line in Fig.~\ref{fig:summarymainresult}),
    (proved in Sec.~\ref{subsec:ration}),
    \item for ONB tomography of the maximally-mixed state,
    the ratio~$\rat$ is exactly~$d+1$ (green line in Fig.~\ref{fig:summarymainresult}),
    (proved in Sec.~\ref{subsec:OdRatioMaxMixState}),
    and the true ratio~$\ratMI$ is exactly~$d+1$ (proved in Sec.~\ref{sec:trueratGMMmodel}),
    \item
    for ONB tomography of any state,
    the maximum ratio~$\rat$ is upper-bounded by~$d+2$
    (proved in Sec.~\ref{subsec:ArbitStates})
    and the maximum true ratio~$\ratMI$ is upper-bounded by~$d+2$
    (proved in Sec.~\ref{sec:trueratGMMmodel}),
    \item for tomography of the maximally-mixed state
    in any non-orthogonal basis, the ratio~$\rat$ is upper-bounded
    by~$d+1$ (proved in Appendix~\ref{sec:arbitweight}),
    \item for tomography of any state
    in any non-orthogonal basis,
    the ratio~$\rat$ is upper-bounded by~$d+2$
    (numerical result
    in Appendix~\ref{sec:arbitweight}),
    \item for estimating any number of GMM-coefficients
    of the maximally-mixed state,
    the ratio~$\rat$ is at most~$d+1$ (proved in Sec.~\ref{sec:EstimatingFewGGMMs}),
    \item for ONB tomography of any state,
    the maximum ratio~$\rat$ at fixed (known) purity
    decreases with purity and is at most~$d+1$
    (numerical result in Sec.~\ref{sec:numericalresults}),
    \item for any model comprising~$n$ parameters,
    the ratio~$\rat$ is upper-bounded by~$\min(n, d+1)$ (red line in Fig.~\ref{fig:summarymainresult}), (conjecture).
\end{itemize}

Notably, the~$d+2$ bounds
in the third and fifth points above
are loose and based on numerical evidence,
we expect the attainable bound therein to be~$d+1$.
The rest of our paper is structured as follows.
In Sec.~\ref{sec:Background}
we introduce and define precision bounds
for the individual and the collective measurement
scenarios and
review relevant background on them.
In Sec.~\ref{sec:npjqiresults},
we present our results,
formalising the key quantities~$\rat$ and~$\ratMI$
in Sec.~\ref{subsec:prelims},
and presenting analytical and numerical results
in Secs.~\ref{subsec:analresults} and~\ref{sec:numericalresults},
respectively.
We follow up with a discussion
of our results in Sec.~\ref{sec:Disc}.
Finally, our methodology is presented in
Sec.~\ref{sec:Methods},
whilst deferring mathematical proofs
to Appendices~\ref{sec:appB}--\ref{sec:appMICRB}.

\section{Background: Precision Bounds for Parameter Estimation}
\label{sec:Background}

In this section, we present a brief recap on
quantum parameter estimation
and introduce
precision bounds
for the collective and the individual
measurement scenarios.
The general recipe
to estimate~$n$ parameters~$\theta\coloneqq\{\theta_j\}_{j\in [n]} \in \Theta \subseteq \mathbb{R}^n $
(where we define~$[n]\coloneqq \{1, \dots, n\}$
and denote by~$\Theta$ the set of all possible parameter values)
of a quantum state~$\rho_\theta$
belonging to a~$d$-dimensional Hilbert
space~$\mathcal{H}_d$
involves two steps.
First, one performs quantum measurements,
generally positive operator-valued measures (POVMs)~$\{\Pi_l\}_{l\in [m]}$
with~$m$ outcomes, on~$\rho_\theta$.
Second, a classical estimator operator~$\hat{\theta}_{jl}$
is constructed that assigns an estimated value
to~$\theta_j$ for each measurement outcome~$l\in[m]$,
which occurs with probability~$p_l \coloneqq \qtrace (\rho_\theta \Pi_l)$.
Here~$\qtrace$ (in serif font) denotes
tracing over the quantum system.

\begin{figure}[hbtp]
    \centering
    \includegraphics[width=0.85\columnwidth]{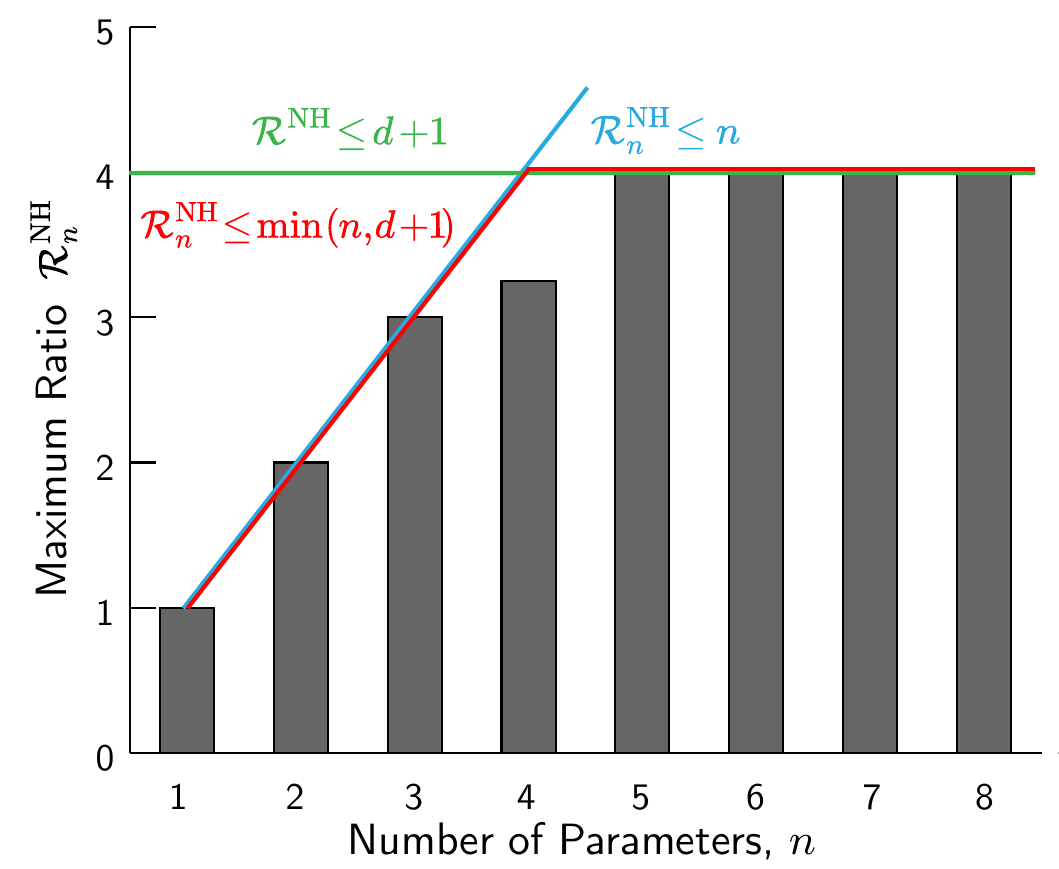}
    \caption{Summary of our main results on
    the maximum ratio quantities~$\rat_n$~\&~$\rat$.
    Our~$\rat_n \leq n$ (blue),~$\rat\leq d+1$ (green)
    and~$\rat_n\leq\min(n,d+1)$ (red) bounds are plotted
    against numerically and analytically found
    maximum collective enhancement values (bar chart)
    %results for the maximum collective enhancement
    for~$n$-parameter estimation
    from qutrits~($d=3$).}
    \label{fig:summarymainresult}
\end{figure}

The performance of the estimator
is quantified via its
mean squared error (MSE) matrix
\begin{equation}
    (V_\theta)_{jk} \coloneqq \sum_{l\in[m]} (\hat{\theta}_{jl} - \theta_j)(\hat{\theta}_{kl} - \theta_k) p_l \, ,
\end{equation}
the trace of which gives the
total average squared deviation~$\ctrace (V_\theta)  = \sum_{l,j} (\hat{\theta}_{jl} - \theta_j)^2 p_l$.
Here~$\ctrace$ (in sans serif font)
denotes tracing over the classical
or parameter indices.
In this work, we focus on
the local setting,
where the parameters of interest
are close to
their true values,~$\theta^* \coloneqq \{ \theta_j^*\}_{j\in[n]}$,
i.e.,~$\theta\approx\theta^*$.
For locally-unbiased (LUB) estimators,
which have zero bias at the true parameter
values,~$V_\theta$
is equivalent to the covariance matrix
of parameter estimates
and~$\ctrace (V_\theta)$ is simply the
sum of variances of each parameter.

Precision bounds lower-bound
the uncertainties
in estimating multiple (possibly) incompatible parameters.
In this work, we focus on precision bounds on~$\ctrace (V_\theta)$;
the classical CRB yields a lower bound to this via
\begin{equation}
\label{eq:CFICRB}
    V_\theta \succcurlyeq J^{-1} \implies \ctrace (V_\theta)\geq \ctrace (J^{-1}) \, ,
\end{equation}
where~$A\succcurlyeq B$ denotes positive semi-definiteness
of~$A-B$, and~$J\equiv J(\rho_\theta, \{\Pi_l\}_{l\in[m]})$ is the classical Fisher information (CFI) matrix.
The CFI
(defined later in Eq.~\eqref{eq:QFImatrix11})
is best understood as a measure
on the parameter space~$\Theta \subseteq \mathbb{R}^n$
of the local sensitivity of measurements~$\{\Pi_l\}_{l\in[m]}$
towards each~$\theta_j$ when measuring state~$\rho_\theta$.
Minimising~$\ctrace (J^{-1})$
in Eq.~\eqref{eq:CFICRB}
over all possible measurements~$\{\Pi_l\}_{l\in[m]}$
yields the tight
MICRB~\cite{Nagaoka2005b,Hayashi1997,HayashiOuyang2023},
\begin{equation}
\label{eq:CMIdef}
    \MI \coloneqq \min_{\{\Pi_l\}_{l\in[m]}} \ctrace (J^{-1}) \, ,
\end{equation}
stipulating the ultimate precision attainable
via individual measurements.
A recent reformulation of the MICRB~\cite{HayashiOuyang2023}
recasts Eq.~\eqref{eq:CMIdef} into a conic program
(see Eq.~\eqref{eq:CMIdef2} for definition)
and presents SDP lower bounds to it.
However, neither closed-form expressions
nor efficient numerical techniques
are known for evaluating the conic program
(the SDP approximation's complexity
scales as~$d^{12}$ for qudit tomography)
and analytic solutions are known
only for qubit models,
single-parameter problems,
and Gaussian models~\cite{HayashiOuyang2023}.

A different precision bound on~$\ctrace(V_\theta)$
for the separable-measurement case,
developed by Nagaoka~\cite{Nagaoka2005a}
and Hayashi~\cite{Hayashi1999},
is the NHCRB,
\begin{equation}
\label{eq:HNCRBdefn}
    \begin{aligned}
        \CNH \coloneqq &\min_{\mathbb{L}, \mathbb{X}}  \Big \{ \mathbb{T}\mathrm{r} [ \mathbb{S}_\theta \mathbb{L} ] \; \vert \;
        \mathbb{L} \succcurlyeq \mathbb{X} \mathbb{X}^\top \, , \\
        &\mathbb{L}_{jk} = \mathbb{L}_{kj} \; \mathrm{Hermitian} \,  \Big \} -\ctrace \left( \theta \theta^\top\right) .
    \end{aligned}
\end{equation}
Here~$\mathbb{X}\coloneqq \{X_1, \dots, X_n\}^\top$
are the Hermitian LUB operators
that satisfy (abbreviating~$\frac{\partial}{\partial \theta_j}$ as~$\partial_j$)
\begin{equation}
\label{eq:origunbiaseq}
     \qtrace(\rho_\theta X_j) = \theta_j \quad \& \quad \qtrace(\partial_j\rho_\theta X_k) = \delta_{jk} \, ,
\end{equation}
and~$\mathbb{S}_\theta = \mathds{1}_n\otimes \rho_\theta$,
$\mathbb{T}\mathrm{r}$ denotes trace over both
classical and quantum
subsystems,~$(\,\cdot\,)^\top$ denotes transpose with respect to
the classical (parameter) index,~$\mathds{1}_n$
denotes the $n$-dimensional identity matrix,
and blackboard fonts represent classical-quantum matrices.
The minimisation in Eq.~\eqref{eq:HNCRBdefn} is
a conic optimisation over the extended Hilbert space~$\mathbb{H} := \mathbb{C}^n\otimes \mathcal{H}_d$,
combining parameter space~$\mathbb{C}^n$
and Hilbert space~$\mathcal{H}_d$ of the qudit.
The other optimisation variable~$\mathbb{L}$
is a positive operator on~$\mathbb{H}$
that is also symmetric under
a partial transpose with respect to the first subspace~$\mathbb{C}^n$,
so as to reduce to valid covariance matrices on~$\mathbb{C}^n$
under~$\qtrace [ \mathbb{S}_\theta \mathbb{L}]$~\cite{Suzuki2023, HayashiOuyang2023}.

Although
the NHCRB is not always tight~($\CNH\leq \MI$),
it is an analytically-tractable
lower bound to~$\MI$
that is often provably tight~($\CNH=\MI$)~\cite{Jun2024Qest,Lorcan21},
and always efficiently-computable
(the SDP's complexity scales as~$d^{4.5}$
for qudit tomography)~\cite{Lorcan21}.
Further, Ref.~\cite{HayashiOuyang2023}
reformulated the tight bound and demonstrated
the NHCRB to be a good approximation to
the tight bound, with a gap of less than~5\%,
for large qudit dimensions.
The GMCRB, another
separable-measurement precision bound,
is defined as~$\CGM \coloneqq 
\left ( \ctrace[J_\mathrm{SLD}^{-1/2}]\right )^2/(d-1)$~\cite{Gill2000},
where~$J_\mathrm{SLD}$ is the
symmetric logarithmic derivative (SLD)
QFI
(see Appendix~\ref{sec:appSLDRLD} for definition).
The GMCRB
is generally not tight for~$d$-dimensional states,
and for~$d>2$ and~$n\leq d-1$, the GMCRB
is known to be weaker than the SLD CRB~\cite{Gill2000}
(see Appendix~\ref{sec:appSLDRLD} for definition).
Nonetheless, the GMCRB is analytically-tractable
and, for the problem of
tomography in orthonormal bases,
has been proven
to be at most a factor of~2 away from
the tight bound~\cite{ZhuThesis}.

The HCRB is a collective-measurement
precision bound on~$\ctrace(V_\theta)$,
defined as
\begin{align}
    \CH \coloneqq &\min_\mathbb{X} \left \{ \ctrace \left ( \mathbb{Z}_\theta [\mathbb{X}]\right ) + \left \Vert \Im \mathbb{Z}_\theta [ \mathbb{X} ]  \right \Vert_1 \right \} \nonumber \\
    & \quad \quad \quad \quad \quad \quad  -\ctrace \left( \theta \theta^\top\right)  \, ,  
 \nonumber \\
    \mathbb{Z}_\theta[\mathbb{X}]_{jk} \coloneqq &\qtrace ( \rho_\theta X_j X_k ) \, ,
    \label{eq:Holevodefn1}
\end{align}
where~$\Vert X \Vert_1 \coloneqq \Tr (\sqrt{X^\dagger X})$
denotes the trace norm.
An equivalent expression for~$\CH$,
written in a similar form as Eq.~\eqref{eq:HNCRBdefn},
is
\begin{equation}
\begin{gathered}
\label{eq:Holevodefn2}
    {\CH} {\coloneqq}  \min_{\mathbb{L}, \mathbb{X}}
    \Big{\{} \bigtrace [ \mathbb{S}_\theta \mathbb{L}] \, \vert \,
    {\qtrace} [ \mathbb{S}_\theta \mathbb{L} ] \; {\mathrm{{real}, {symmetric,}}} \,   \\
     \qtrace [ \mathbb{S}_\theta \mathbb{L} ] \succcurlyeq \qtrace [ \mathbb{S}_\theta \mathbb{X} \mathbb{X}^\top ] \Big{\}}  -\ctrace \left( \theta \theta^\top\right) \, .
\end{gathered}
\end{equation}
Note that the minimisations
in Eqs.~\eqref{eq:HNCRBdefn},~\eqref{eq:Holevodefn1}
and~\eqref{eq:Holevodefn2}
have no explicit closed-form solution
for general mixed states~$\rho_\theta$~\cite{Jun2016}
and are typically
evaluated numerically via
SDPs~\cite{Albarelli2019,Lorcan21}.

Besides $\CH \leq \CNH$,
the following ordering
between the various precision bounds
is known
\begin{equation}
\label{eq:orderingprecisionbounds}
   {\max} \left ( {\SLD}, {\RLD} \right ) \leq {\CH} \leq {\CNH}  \leq {\MI} \, .
\end{equation}
Here~$\SLD$ and~$\RLD$ are, respectively,
the SLD CRB
and the right-logarithmic derivative (RLD) CRB
(see Appendix~\ref{sec:appSLDRLD} for definitions).
By definition,~$\MI$
is the tightest precision bound for separable measurements,
and hence is greater than or equal to
all four other CRBs in Eq.~\eqref{eq:orderingprecisionbounds}.
We know that all three inequalities in Eq.~\eqref{eq:orderingprecisionbounds}
are saturated
for single-parameter estimation~\cite{Matsumoto2002}.
Moreover,~$\CNH = \CH$ for estimating any number of parameters
from pure states~\cite{Matsumoto2002}.
On the other hand, if the single-copy NHCRB and the HCRB are
unequal, this gap persists between the finite-copy NHCRB
and the HCRB, shrinking asymptotically
with the number of copies~\cite{LorcanGap}.
As for the GMCRB, its position
in Eq.~\eqref{eq:orderingprecisionbounds}
depends on model properties:
for~$d=2$, we have~$\SLD \leq \CGM$
but for~$d>2$ and~$n\leq d-1$, we have~$\CGM \leq \SLD$,
whereas for~$d>2$ and~$n>d-1$,
no ordering can be specified.

Beyond ordering,
some ratio relationships
between the precision bounds
are known.
The HCRB is known
to be at most twice the SLD CRB,
i.e.,~$\SLD \leq \CH \leq 2 \SLD$~\cite{Carollo2019,TAD20}
so that, up to a factor of~2,
the SLD CRB also
quantifies collective precision.
For the problem of tomography
in an orthonormal basis, the
relation~$\CGM \leq \MI \leq 2 \CGM$ was established
in Ref.~\cite{ZhuThesis}.
At the level of the Fisher information,
the Gill-Massar trace~$\ctrace[J^{-1}_\mathrm{SLD} J] \leq d-1$
reduces to a ratio of the QFI~$J_\mathrm{SLD}$
and the CFI~$J$ when they are proportional:~$J \propto J_\mathrm{SLD}$.
This is the case for Fisher symmetric measurements~\cite{FS16}
that attain a CFI such that~$J_\mathrm{SLD} = (d+1) J$,
and thus~$\MI/\SLD=d+1$,
for tomography of the maximally-mixed state~\cite{Zhu1}.
However, Fisher symmetric measurements exist
only for maximally-mixed states and pure states~\cite{FS16,Zhu1},
so the ratio for general mixed states remains undecided.
The recent reference~\cite{Candeloro2024}
analysed the
impact of dimensionality on precision
and parameter incompatibility~\cite{Compatibility16}
by defining the normalised gap~$\Delta = (\CH-\SLD)/\SLD$,
which satisfies~$0\leq \Delta \leq 1$.
By considering the estimation
of~$n=2$ and~$3$ parameters
encoded unitarily onto~$d$-dimensional states,
the authors showed that asymptotic incompatibility
can vanish altogether when~$d>n$,
making the gap~$\Delta=0$.
Surprisingly, a direct ratio relationship
between~$\MI$ and~$\CH$,
or even between~$\CNH$ and~$\CH$
is not known~\cite{BG21},
leading to a gap in our knowledge
of the potential quantum advantage
offered by collective measurements.

\section{Results}
\label{sec:npjqiresults}

%\kern-3em

\subsection{Collective Quantum Enhancement}
\label{subsec:prelims}

We first define a measure of collective
enhancement for
estimation tasks.
Whereas the truly attainable
ratio of collective and separable
variances is given by~$\MI/\CH$,
the analytic intractability of~$\MI$
renders this quantity beyond the reach
of currently available tools.
Instead, we consider the ratio
between the NHCRB and the HCRB,
which provides a lower bound to~$\MI/\CH$.
Specifically, in this work,
we shall analyse
the collective quantum enhancement
\begin{equation}
\label{eq:ratiodef}
    \rat \left [\{\rho_\theta \vert \, \theta \in \Theta \} \right ] \coloneqq \max_{\theta \in \Theta} \;  \frac{\CNH[\rho_\theta]}{\CH[\rho_\theta]} \, ,
\end{equation}
for a given quantum statistical
model~$\{\rho_\theta \vert \, \theta \in \Theta \}$
with~$n$ parameters for~$d$-dimensional qudit states,
where the maximum is over
all allowed parameter values for this model.
This quantity can be interpreted as
a measure of the maximum
quantum enhancement obtainable from using collective measurements
over separable measurements for this particular model.

A further maximisation over all
%full-rank
quantum models with the
same number of parameters,~$n$,
and for the same probe dimension,~$d$,
is possible,
\begin{equation}
\label{eq:defmaxration}
    \rat_{n} \coloneqq \max_{n\text{-parameter models}} \; \rat \left [\{\rho_\theta \vert \, \theta \in \Theta \} \right ] \, ,
\end{equation}
where we have suppressed the dependence
on~$d$ for conciseness.
A final maximisation over
all possible~$n$ for a given system dimension~$d$
($1\leq n \leq \nmax $) then leads to
\begin{equation}
\label{eq:defmaxratio}
    \rat\coloneqq \max_{1 \leq n \leq \nmax} \; \rat_n \, ,
\end{equation}
corresponding to
the ultimate collective quantum enhancement
in precision
for this system
dimension~\cite{Szczykulska2016,Conlon2023b}.
Here~$\nmax$ represents the maximum possible number
of independent parameters
and equals~$2(d-1)$ for pure qudits
and~$d^2-1$ for mixed qudits.
In this work we consider general mixed states
and, thus, define~$\nmax\coloneqq d^2-1$.

Physically,
Eqs.~\eqref{eq:defmaxration}
\&~\eqref{eq:defmaxratio}
define model-free quantities that capture
the maximum potential advantage
of collective measurements across various quantum models,
comparing their general utility
beyond particular estimation scenarios.
Practically, these two quantities,~$\rat_n$
and~$\rat$, offer insights into
the behaviour of
the maximum utility
of collective measurements
from two opposing extremes
of the number of parameters.
As depicted in Fig.~\ref{fig:summarymainresult},
analysis using~$\rat_n$
reveals a maximum utility linear in~$n$
for small~$n$ whereas
analysis using~$\rat$
reveals a plateauing utility,
fixed by~$d$ instead of~$n$,
for large~$n$.
In the following we
abbreviate the model
representation~$\{\rho_\theta \vert \, {\theta \in \Theta} \}$
to~$\{\rho_\theta\}$
while specifying the model explicitly.

The NHCRB is generally close
to the MICRB~\cite{HayashiOuyang2023},
but whether~$\rat_n$ (or~$\rat$) corresponds to
the true maximum collective enhancement
depends on the NHCRB's attainability
for the model maximising~$\MI/\CH$
in the setting of Eq.~\eqref{eq:defmaxration}
(Eq.~\eqref{eq:defmaxratio}).
The non-attainability
of the NHCRB in specific
cases~\cite{HayashiOuyang2023}
means that the quantities~$\rat[\{\rho_\theta\}]$,
$\rat_n$ and~$\rat$
could be smaller than
the corresponding true
maximum collective enhancements,
denoted~$\ratMI[\{\rho_\theta\}]$,~$\ratMI_n$
and~$\ratMI$. These are defined
similarly to Eqs.~\eqref{eq:ratiodef}---\eqref{eq:defmaxratio}
through~\mbox{$\ratMI[\{\rho_\theta\}] \coloneqq \max_{\theta \in \Theta} \MI[\rho_\theta]/\CH[\rho_\theta]$},
$\ratMI_n \coloneqq \max_{n\text{-param. models}} \ratMI[\{\rho_\theta\}]$
and
\begin{equation}
\label{eq:ratlessactualrat}
%\begin{gathered}
    %\rat \leq  
     \ratMI \coloneqq %\qquad \qquad \qquad \qquad \qquad \quad %\\
    \max_{1 \leq n \leq \nmax} \ratMI_n \, .
    %\quad \quad \quad 
    %\max_{\substack{n\text{-parameter models}\\1 \leq n \leq \nmax}} \; 
    %\max_{\theta \in \Theta} \;  
    %\frac{\MI[\rho_\theta]}{\CH[\rho_\theta]} 
%\end{gathered}
\end{equation}
The same construction of maximum ratio quantities
applies to other choices of separable
measurement bounds,
leading to~$\ratGM[\{\rho_\theta\}]$ and~$\ratGM_n$
for the GMCRB, studied later
and depicted in Fig.~\ref{fig:GMNHratioComp}.

The quantities~$\rat$ and~$\ratMI$ can only increase with~$d$.
This is because of two reasons:
(i) higher dimensions allow more incompatible parameters
to estimate, and
(ii) a lower-dimensional model
can always be embedded in a higher dimension
keeping the ratio invariant.
Note that we only consider full-rank models
in defining Eqs.~\eqref{eq:ratiodef},~\eqref{eq:defmaxration},
\&~\eqref{eq:defmaxratio},
i.e., we assume~$\rho_\theta$ to be
non-singular or non-degenerate
for all~$\theta \in \Theta$.\footnote{%
Singular or degenerate
states can be approximated
arbitrarily well by non-degenerate
ones via regularisation~\cite{Watanabe2011}.}
Our definitions imply
that for any~$n$-parameter model
over~$d$-dimensional states,~$\rat[\{\rho_\theta\}] \leq \rat_n \leq \rat$
and~$\ratMI[\{\rho_\theta\}] \leq \ratMI_n \leq \ratMI$.

Below, in Sec.~\ref{subsec:analresults},
we present our analytical results.
First, in Sec.~\ref{subsubsec:NHCRBratios},
we present our results on
the ratio quantities~$\rat_n$ and~$\rat$
and on~$\rat[\{\rho_\theta\}]$ for the linear GMM model.
The GMM model results apply
to tomography in arbitrary ONB, and
we conjecture that
this model at~$\theta=0$
maximises Eq.~\eqref{eq:defmaxratio},
thereby attaining~$\rat$;
we prove the attainability
of the NHCRB at this point.
Then, in Sec.~\ref{subsec:analresultsothrbnds},
we present our analytical result
using the MICRB to upper-bound~$\ratMI[\{\rho_\theta\}]$
for the problem of ONB tomography
and compare this to
an upper bound derived from~$\ratGM_n$.
Finally,
in Sec.~\ref{sec:numericalresults},
we present
our numerical results including
the extension of~$\rat[\{\rho_\theta\}]$ to tomography in non-orthogonal bases
and to tomography under fixed probe purity,
as well as random sampling experiments
addressing~$\rat_n$ and~$\rat$
for generic qudit estimation models.

\subsection{Analytical Results}
\label{subsec:analresults}

\subsubsection{Results for NHCRB ratio}
\label{subsubsec:NHCRBratios}

In Sec.~\ref{subsec:ration},
we prove
a model-independent and attainable upper bound
on the collective precision enhancement.
Specifically,
by using an upper bound to~$\CNH$~\cite{Suzuki2023,HayashiOuyang2023}
and using~$\SLD$ to lower-bound~$\CH$,
we prove that for estimating~$n$ independent parameters
of a quantum state, the collective enhancement
is at most a factor of~$n$,
i.e.,~$\rat_n\leq n$.
We state this as
Theorem~\ref{th:ratngenstate}
(see Sec.~\ref{subsec:ration} for proof):

\begingroup
\begin{theorem}
For estimating~$n$ parameters of any
qudit state~$\rho_\theta$,
the collective enhancement~${\CNH[\rho_\theta]}/{\CH[\rho_\theta]}\, {\leq}\,{n}$,
i.e.,~\mbox{$\rat_n \leq n$}.
\end{theorem}
\endgroup

\noindent
The factor-of-$n$ smaller
attainable precision
intuitively aligns with
the approach of dividing the
multi-parameter problem into~$n$
single-parameter problems,
where the SLD-optimal separable measurement for each
parameter~$\theta_j$ individually
is performed on a fraction~$\nicefrac{1}{n}$
of the number of available state copies.\footnote{However,
this does not constitute a locally-unbiased multi-parameter
estimation strategy, which must assign locally-unbiased estimates
to all the parameters in each trial.
When performing the SLD-optimal measurement for~$\theta_j$,
the procedure for assigning locally-unbiased estimates
for~$\theta_k$ ($k\neq j$) is discussed in Ref.~\cite{Fujiwara1999}.}
Due to technical difficulties
in guaranteeing the local-unbiasedness
of parameter estimates within this approach~\cite{Fujiwara1999},
we instead use the analytic upper bound
to~$\CNH$ developed in Ref.~\cite{Suzuki2023}
(summarised in Appendix~\ref{sec:summaryJunTR})
to prove Theorem~\ref{th:ratngenstate}.

The upper bound in Theorem~\ref{th:ratngenstate}
can be a tight relation for
models with small~$n$.
For example, a ratio~$\rat_{n=2}=2$ is attained
for~$d=3$ by the model of estimating
the coefficients of~$\lambda_1$ \&~$\lambda_2$
(see Appendix~\ref{sec:appB} for definitions)
in the maximally-mixed qutrit state,
and a ratio~$\rat_{n=3}=3$ is attained
for~$d=2$ by the model of
estimating the three Pauli
coefficients of the maximally-mixed
qubit state~\cite{BacuiLi}.
We also expect the ratio~$\rat_n$
to increase with~$n$, because
having more parameters to estimate
can lead to higher incompatibility---this
implies~$\rat_n\leq \rat_{\nmax}$.
On the other hand,
at the~${n = \nmax= d^2-1}$ limit,
the upper bound from
Theorem~\ref{th:ratngenstate}
would imply
an enhancement~$\rat_{\nmax}\leq d^2-1$---quadratic
in~$d$---is attainable;
this is not the case,
as we subsequently establish.

In Sec.~\ref{subsec:model},
we introduce the `linear GMM model':
an~$\nmax$-parameter family of
$d$-dimensional qudit states
given by
\mbox{$\rho_\theta = \mathds{1}_d/d
+\sum_{j=1}^{\nmax} \theta_j \lambda_j$}~\cite{Watanabe2011}.
Here, the~$\nmax$ parameters
of interest,~$\{\theta_j\}_{j\in[\nmax]}$,
are the coefficients
of the GMMs~$\Lambda_d\coloneqq\{\lambda_j\}_{j\in [\nmax]}$.
The GMMs~$\Lambda_d$ are traceless,
Hermitian generalisations
of the qubit Pauli operators
(see Appendix~\ref{sec:appB})
and the parameterisation
is valid for any qudit state.
Estimating~$\{\theta_j\}_{j\in[\nmax]}$
thus corresponds to tomography
in the GMM basis,
which is an ONB
given that~$\qtrace(\lambda_j  \lambda_k) = \delta_{jk}$.
In fact,
the unweighted GMM model at any~$\rho_\theta$
is equivalent to tomography of~$\rho_\theta$
in any other ONB,
as we show using Lemma~\ref{lemma:orthtransmat}~in
Sec.~\ref{subsec:model}.
This means that the QCRBs
as well as ratios between them are invariant
to the particular choice of ONB
and the following results derived for
the GMM basis hold for any ONB.

For the GMM model,
we prove that the maximum collective enhancement~$\rat[\{\rho_\theta\}]$
is a factor linear in dimension~$d$,
as opposed to quadratic.
Specifically,
we show that the collective enhancement
is exactly~$d+1$ for estimating
from the maximally-mixed state,
$\rho_\mathrm{m} = \mathds{1}_d / d$,
in Theorem~\ref{th:ratdmms}
(see Sec.~\ref{subsec:OdRatioMaxMixState} for proof):

\begingroup
\begin{theorem}
For ONB tomography of the
maximally-mixed qudit state~$\rho_\mathrm{m}$, the collective enhancement~$\CNH[\rho_\mathrm{m}]/\CH[\rho_\mathrm{m}] = d+1$.
\end{theorem}
\endgroup

\noindent
The factor of~$d+1$ here previously
appeared in the relation~$J_\mathrm{SLD} = (d+1) J$
satisfied by Fisher SIC measurements for~$\rho_\mathrm{m}$~\cite{Zhu1}.
This is because for the particular
model considered in Theorem~\ref{th:ratdmms},
the NHCRB is tight (see below)
and the SLD CRB equals the HCRB.
Therefore, the ratio~$\CNH/\CH$
reduces to the ratio~$\MI/\SLD$, which equals
the proportionality constant between the
QFI~$J_\mathrm{SLD}$ and the CFI~$J$.
However, the true merit of Theorem~\ref{th:ratdmms}
lies in our novel method for its proof,
which directly extends
to upper-bound the collective enhancement
in several settings beyond Theorem~\ref{th:ratdmms}.
Specifically, this lets us
prove that the maximum enhancement~$\rat[\{\rho_\theta\}]$
is less than~$d+2$ for estimating
from arbitrary states in
Theorem~\ref{th:ratdarbstate}
(see Sec.~\ref{subsec:ArbitStates} for proof):
\begingroup
\begin{theorem}
    For ONB tomography of arbitrary~$d$-dimensional qudit state~$\rho_\theta$, the maximum collective enhancement~$\rat[\{\rho_\theta\}] 
    %= \max_{\theta \in \Theta} \CNH[\rho_\theta]/\CH[\rho_\theta] 
    \leq d+2$. 
\end{theorem}
\endgroup
\noindent
Although we prove an upper bound
of~$d+2$,
based on numerical evidence
shown in Fig.~\ref{fig:estimateallarbitstate},
we expect the attainable bound to be~$d+1$
and expect~$\rat[\{\rho_\theta\}]$ to be maximised
at~$\theta=0$ over~$\Theta$.
Thus for the linear GMM model,
we propose~$\rat[\{\rho_\theta\}]\leq d+1$.
Further, for the maximally-mixed case,
we prove the SIC POVM
to be an optimal individual measurement.
This proves the NHCRB to be
attainable or tight in this case,
implying~$\MI/\CH=d+1$ as well
at this maximum point.

We then treat two extensions
of the linear GMM model.
The first is the weighted version,
where an arbitrary, full-rank,
parameter-independent weight
matrix~$W$ is included
in the cost function (Appendix~\ref{sec:arbitweight}).
Importantly,
weighted models are
equivalent to reparameterisations
of the unweighted model,
i.e., estimating parameters
that are not
coefficients
in any particular ONB~\cite{Lorcan21,Fujiwara1999,ABGG20}.
Our results in Appendix~\ref{sec:arbitweight}
prove that for the maximally-mixed state,~$\rho_\mathrm{m}$,
and for estimating any~$\nmax$
independent parameters,
the collective enhancement is at most~$d+1$.
We also numerically show (Fig.~\ref{fig:weightedtomo})
that the collective enhancement for
estimating from any state~$\rho_\theta\neq \rho_\mathrm{m}$
under a weight~$W$ is smaller than
the collective enhancement for
estimating from~$\rho_\mathrm{m}$
%the maximally-mixed state
under the same weight~$W$.
This suggests the maximum collective enhancement~$\rat_{\nmax}$
for any full-parameter problem
to be at most~$d+1$,
i.e.,
\begin{equation}
      \rat_{\nmax} = \max_{\substack{\nmax\text{-parameter} 
      \\ \mathrm{models}}}
      \; \rat \left [\{\rho_\theta \} \right ] \leq d+1 \, ,
\end{equation}
and this upper bound is attained by the model
studied in Sec.~\ref{subsec:OdRatioMaxMixState}.

In Sec.~\ref{sec:EstimatingFewGGMMs},
we treat the extension of the linear
GMM model to the~${n< \nmax}$ case,
assuming the remaining~${\nmax-n}$
GMM coefficients to be zero.
In this case, we prove that for
estimating from
the maximally-mixed state,
the maximum collective enhancement
is~$d+1$ in
Theorem~\ref{th:estfewGMMs}
(Sec.~\ref{sec:EstimatingFewGGMMs}):

\begingroup
\begin{theorem}
For estimating fewer-than-$\nmax$ coefficients
of GMMs of the maximally-mixed qudit state~$\rho_\mathrm{m}$,
the collective enhancement~$\CNH[\rho_\mathrm{m}]/\CH[\rho_\mathrm{m}]\leq d+1$.
\end{theorem}
\endgroup

\noindent
We do not analytically solve this model
for other states, or for states
with the remaining~$\nmax-n$
GMM coefficients non-zero.
However, the~$n$ bound in
Sec.~\ref{subsec:ration}
and numerical results in
Sec.~\ref{sec:numericalresults}
suggest that the maximum enhancement~$\rat_n$
is non-decreasing with increasing~$n$ at fixed~$d$.
Having analysed the~$n=\nmax$
case in depth, we expect
that for any~${n}\,{<}\,{\nmax}$ model,
the same bound
of~${d+1}$ should hold,
i.e.,~$\rat_n \leq \rat_{\nmax} \leq d+1$.

In summary,
we find that there
are two different upper bounds
on the maximum collective enhancement
in the low~$n$ and the high~$n$ regimes.
Taking the example of qutrits~($d=3$),
for~$n=1,2$ \&~3, models attaining~$\rat_n=n$
are found, whereas for~$n=5,6,7$ \&~8,
models attaining~$\rat_n=d+1$ are found.
Combining these two cases,
we propose~$\rat_n \leq \min(n,d+1)$
and~$\rat\leq d+1$.
Figure~\ref{fig:summarymainresult}
summarises these results
along with numerically-found
maximum ratios
and analytically-found
ratios from known models.

\subsubsection{Results for MICRB \& GMCRB ratios}
\label{subsec:analresultsothrbnds}

The results on~$\rat_n$ and~$\rat[\{\rho_\theta\}]$
in Sec.~\ref{subsubsec:NHCRBratios}
lower-bound the corresponding quantities~$\ratMI_n$
and~$\ratMI[\{\rho_\theta\}]$.
In particular,
for the linear GMM model
that addresses ONB tomography,
we have
proved~$\ratMI[\{\rho_\theta\}] \geq d+1$.
Based on analytical and numerical results,
we conjecture that~$\rat \leq d+1$;
accordingly, Theorem~\ref{th:ratdmms}
implies~$\ratMI \geq d+1$.
In Sec.~\ref{sec:trueratGMMmodel},
%Appendix~\ref{sec:appMICRB},
we use the NHCRB solution
behind
Theorem~\ref{th:ratdmms}
to solve the MICRB for the linear
GMM model at~$\theta=0$,
obtaining a ratio of~$d+1$.
We then extend this solution
to upper-bound the true ratio
for arbitrary~$\theta$
by~$d+2$
in Theorem~\ref{th:ratdMIarbstate}
(see Appendix~\ref{sec:appMICRB} for proof):

\begingroup
\begin{theorem}
        For ONB tomography of arbitrary~$d$-dimensional qudit state~$\rho_\theta$, the maximum true collective enhancement~$\ratMI[\{\rho_\theta\}]
        %= \max_{\theta \in \Theta} \MI[\rho_\theta]/\CH[\rho_\theta] 
        \leq d+2$.
\end{theorem}
\addtocounter{theorem}{-5}
\endgroup
\noindent
Theorem~\ref{th:ratdMIarbstate} therefore
constrains the true maximum enhancement
for ONB tomography as
\begin{equation}
\label{eq:ratlessactualrat2}
    d+1 \leq \ratMI[\{\rho_\theta\}] \leq d+2 \, .
\end{equation}
%    \max_{\substack{1\leq n\leq \nmax\\
%    n\text{-parameter models}}} \max_{\theta \in \Theta} \frac{\MI}{\CH} \geq d+1 \, .

A different way to quantify
collective precision enhancement
is to consider the ratio~$\CGM/\SLD$.
As noted previously, the SLD CRB
can capture collective performance
up to a factor, and the GMCRB,
though not tight for qudit problems with small~$n$,
can be directly computed given the SLD QFI~$J_\mathrm{SLD}$.
This significantly simplifies
the evaluation of the ratio~$\CGM/\SLD$,
which then upper-bounds~$\CGM/\CH$,
compared to the ratios considered in Sec.~\ref{subsubsec:NHCRBratios}.

Let us denote the maximum GMCRB-to-HCRB ratio
for a fixed model as~$\ratGM[\{\rho_\theta\}]$,
for any~$n$-parameter~$d$-dimensional model as~$\ratGM_n$,
and for any~$d$-dimensional model as~$\ratGM$
(similar to Eqs.~\eqref{eq:ratiodef},~\eqref{eq:defmaxration}
and~\eqref{eq:defmaxratio} for the NHCRB).
A straightforward application
of the Cauchy-Schwartz inequality
to the eigenvalues~$\{\nu_j\}_{j\in[n]}$ of~$J_\mathrm{SLD}^{-1/2}$
(the positive square-root of the inverse of the SLDQFI matrix)
leads to
\begin{equation}
\label{eq:GMSLDratio}
    \frac{\left ( \ctrace[J_\mathrm{SLD}^{-1/2}]\right )^2}{\ctrace[J_\mathrm{SLD}^{-1}]}  = \frac{(\sum_j \nu_j)^2}{\sum_j \nu_j^2} \leq n \, ,
\end{equation}
with equality attained if~$J_\mathrm{SLD}$ is a scalar matrix.
The ratio in Eq.~\eqref{eq:GMSLDratio}
proves the upper bound
\begin{equation}
\label{eq:ratGMupbound}
\begin{split}
    \ratGM_n &\coloneqq \max_{\substack{\theta\in\Theta\\n\text{-parameter}\\ \text{models}\, \{\rho_\theta\}}} \frac{\CGM[\rho_\theta]}{\CH[\rho_\theta]} \\
    &\leq \max_{\substack{\theta\in\Theta\\n\text{-parameter}\\ \text{models}\, \{\rho_\theta\}}} \frac{\CGM[\rho_\theta]}{\SLD[\rho_\theta]} \leq \frac{n}{d-1}
\end{split}
\end{equation}
for any~$n$-parameter~$d$-dimensional model,
reducing to~$\ratGM_{\nmax} \leq \nmax/(d-1) = d+1$
at the maximum number of parameters.
Moreover, Eq.~\eqref{eq:ratGMupbound}
is a tight bound:
a ratio~$\ratGM_n = n/(d-1)$ is attained
by the model of estimating any~$n$
GMM coefficients of the maximally-mixed state.
This is because~$J_\mathrm{SLD} = d \, \mathds{1}_n$
here, so that~$\CGM = n^2/(d(d-1))$, dividing which
by~$\SLD=\CH=n/d$
(see Methods Sec.~\ref{sec:EstimatingFewGGMMs})
produces the ratio~$n/(d-1)$,
thereby proving~$\ratGM_n = n/(d-1)$.

In Fig.~\ref{fig:GMNHratioComp},
we compare the
maximum collective enhancement
predicted by~$\ratGM_n$
%and its upper bound
and by~$\rat_n$
%and its upper bound.
for generic~$n$ parameter
models
($1\leq n \leq \nmax$)
on qutrits.
The comparison with the NHCRB,
which is a tighter bound,
reveals
that whereas Eq.~\eqref{eq:ratGMupbound}
upper-bounding~$\ratGM_n$ is tight,~$\ratGM_n$
itself generally underestimates collective enhancement.
This is because the GMCRB
can be far from the attainable
individual precision
for qudit problems,
especially for small~$n$,
but also for full-parameter models
(see Fig.~\ref{fig:GMNHfewpara} in Appendix~\ref{sec:appNHvsGM}).
%largest found ratios for~$\CGM/\CH$
%(red points)
%and for~$\CNH/\CH$
%(blue points)
%across 1300
%randomly-sampled qutrit models for
%each~$n$ from 1 to~$\nmax=8$.
%The NHCRB ratios, satisfying~$\rat_n\leq n$,
%are larger than the GMCRB ratios,
%which satisfy~$\ratGM_n \leq n/(d-1)$
%(gray bars).
Nonetheless,
if it holds that the GMCRB is
separably attainable up to a scaling factor~$g$
(that is independent of~$d$ and~$n$,
e.g., $g$ = 2 for ONB tomography~\cite{ZhuThesis}),
implying~$\MI \leq g \, \CGM$,
then we can reason that
\begin{equation*}
    \frac{\MI}{\CH} \leq g \frac{\CGM}{\SLD} \leq g \frac{n}{d-1} \, ,
\end{equation*}
so that in general,~$\rat_n \leq \ratMI_n \leq g \, \ratGM_n = g \,  n/(d-1)$.
For ONB tomography, which
includes the GMM model as a special case,
%as a special case,
this proves that
\begin{equation}
\label{eq:ratlessactualrat3}
\begin{aligned}
    d+1 \leq \, &\rat[\{\rho_\theta\}] \leq \ratMI[\{\rho_\theta\}] \\
    &\leq 2 \, \ratGM[\{\rho_\theta\}] = 2(d+1) \, .
\end{aligned}
\end{equation}
However, no such scaling factor~$g$ is currently known
for general models
beyond the setting of tomography.

\begin{figure}
    \centering
    \includegraphics[width=0.9\columnwidth]{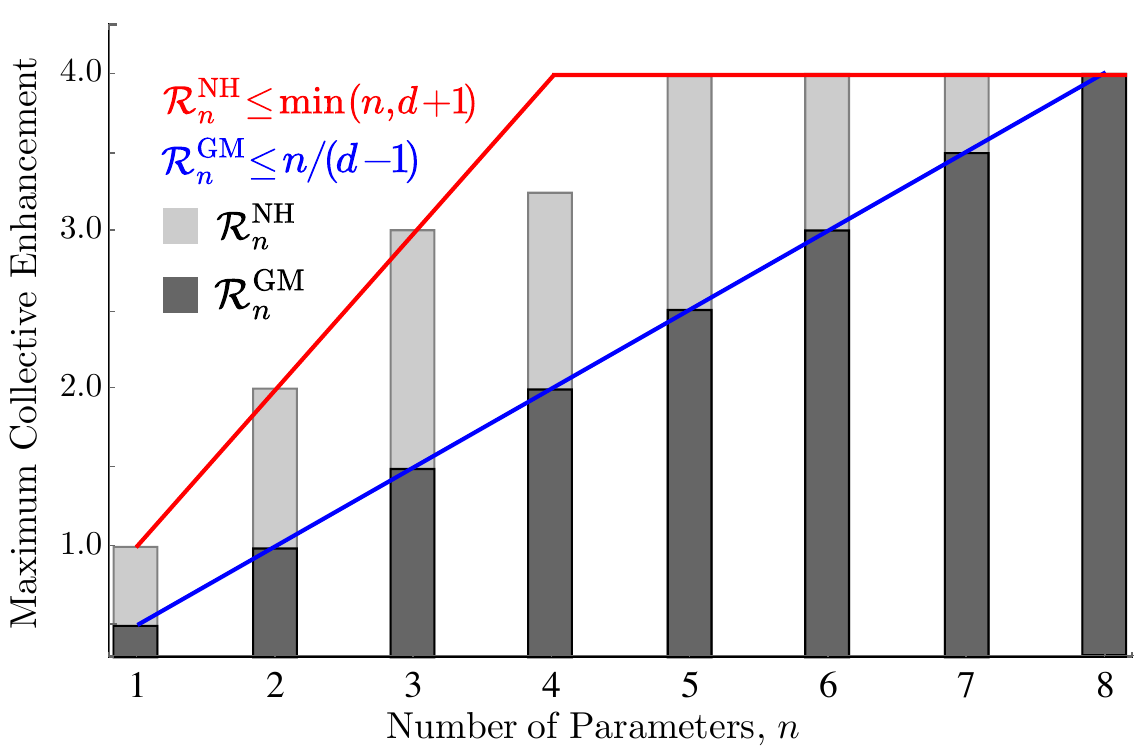}
    \caption{Comparison of the
    collective enhancements specified
    by the GMCRB ($\ratGM_n$)
    and by the NHCRB ($\rat_n$)
    for~$n$-parameter models
    ($1\leq n\leq \nmax$).
    The maximum GMCRB ratios
    (dark gray bar chart)
    satisfy~$\ratGM_n = n/(d-1)$
    (blue line).
    The maximum NHCRB ratios
    (light gray bar chart)
    satisfy~$\rat_n \leq \min(n,d+1)$
    (red line).
    The ratios are plotted against
    the number of parameters,~$n$,
    and include numerically found maximum
    ratios from random-sampling experiments
    (1300 samples for each~$n$ for~$d=3$)
    as well as analytically found ratios.
    The NHCRB ratio~$\rat_n$
    generally predicts a larger collective
    enhancement
    than the GMCRB ratio~$\ratGM_n$,
    except at
    the maximum number of parameters~($n=\nmax$).}
    \label{fig:GMNHratioComp}
\end{figure}

\subsection{Numerical Results}
\label{sec:numericalresults}

In this section,
we report our numerical results
on the effect
of probe-purity on
the collective enhancement
and its maximum,
and on the
dependence of the
maximum collective enhancement~$\rat_n$
on~$d$ and~$n$
for randomly-sampled~$n$-parameter,~$d$-dimensional models.
For both of these analyses,
we rely on random sampling
to study the dependence of the
maximum enhancement on state or model properties
because directly maximising the ratio~$\CNH/\CH$
over the state space~$\mathcal{H}_d$
or the model space
is computationally demanding for large~$d$.
For each randomly-generated problem instance, we
solve the~NHCRB and~HCRB SDPs numerically
and compute their ratio~\cite{Lorcan21}.
Further details of the random-sampling
procedure used for subsequent results
are presented
in Appendix~\ref{sec:appGridPlot}.

\begin{figure*}[hbtp]
    \centering
    \includegraphics[width=\textwidth]{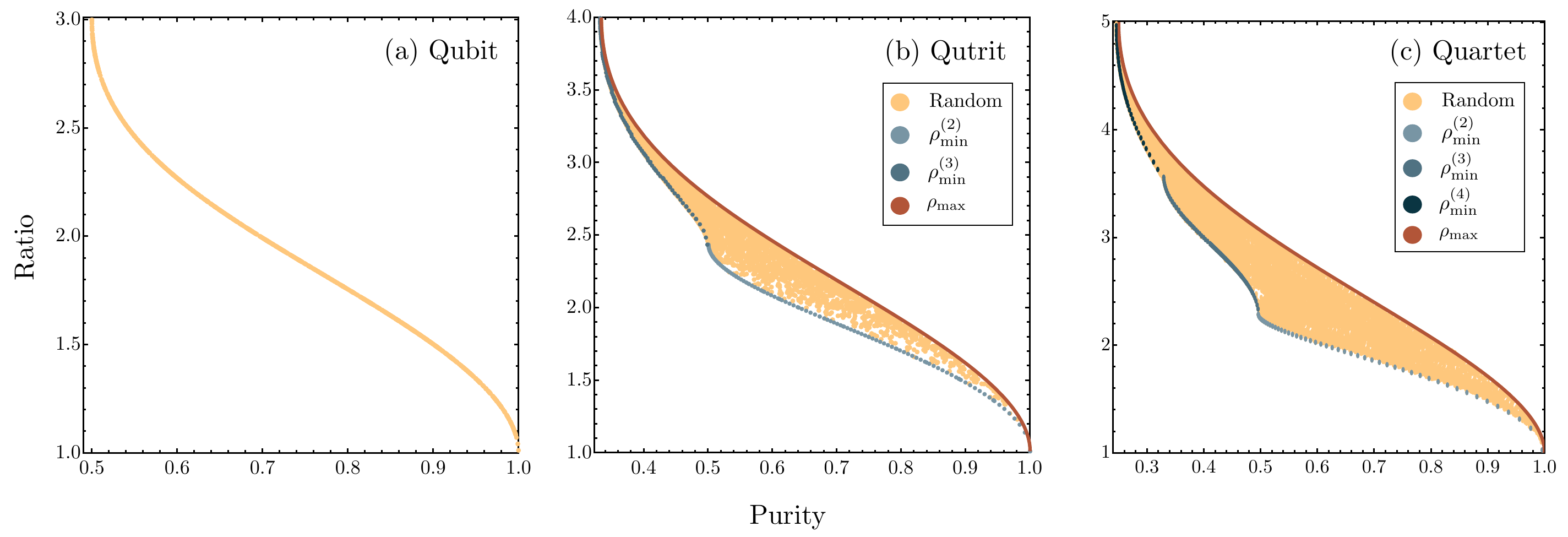}
    \caption{Ratio between the NHCRB and the HCRB versus purity for estimating all~$d^2-1$ GMMs from arbitrary states.
    For qubits (a), we find a one-to-one dependence between ratio and purity (10,000 samples). However,
    for qutrits (b) and quartets (c), there is a region of allowed ratios at any given purity
    (15,000 and 25,000 samples, respectively).
    The ratio at any fixed purity is maximised by the state~$\rho_\mathrm{max}$,
    which is a depolarised pure state,
    and minimised by the states
    $\rho_\mathrm{min}^{(2)}, \rho_\mathrm{min}^{(3)}$
    (and $\rho_\mathrm{min}^{(4)}$ in (c)),
    which are rank-deficient classical states}.
    \label{fig:estimateallarbitstate}
\end{figure*}

In Fig.~\ref{fig:estimateallarbitstate},
we plot our random-sampling results for
the ratio versus purity
in the linear GMM model for $d=2,3$ and~$4$.
For each~$d$, the overall maximum ratio
observed is~$d+1$.
Interestingly, whereas for qubits
the ratio is uniquely determined by purity,
the higher dimensionality of the qudit state space
allows for a range of ratios at any given purity.
We find that the ratio
at a given purity
is maximised by full-rank
depolarised pure states,
$p \ketbra{\phi} + (1-p) \mathds{1}_d/d$ for any pure state~$\ket{\phi}$
and~$p\in[0,1]$.
To simplify computation,
we choose the more specific family
$\rho_\mathrm{max}(p)=p \ketbra{+}_d + (1-p) \mathds{1}_d/d$,
where~$\ket{+}_d = (\ket0+\dots+\ket{d-1})/\sqrt{d}$
and calculate the HCRB to be
\begin{equation}
    \CH[\rho_\mathrm{max}(p)]= \frac{d^2-1}{d} + p(d-1)-\frac{d-1}{d} p^2 \, .
    \label{eq:maxratiostateHCRB}
\end{equation}
For the NHCRB,
based on numerical evidence
for~$d=3$ to~$8$,
the analytic solutions at the boundary cases
($(d^2-1)(d+1)/d$ at $p=0$ and~$2(d-1)$ at~$p=1$),
and the analytic solution for~$d=2$
(Eq.~\eqref{eq:PatricksEqs}),
we find that
\begin{align}
    \CNH&[\rho_\mathrm{max}(p)] = \frac{d^2+1}{2} -\frac{d^2-4d+5}{2} p^2 \nonumber \\
    &+\frac{d^3+2d^2-3d-2}{2d}\sqrt{1-p^2} \, .
\label{eq:maxratiostateNHCRB}
\end{align}
Accordingly,
the maximum collective enhancement at a fixed purity~$\mathrm{P}^*$
is $\CNH[\rho_\mathrm{max}(p^*)]/\CH[\rho_\mathrm{max}(p^*)]$
with~$p^*=\sqrt{\frac{\mathrm{P}^* d-1}{d-1}}$.
Eqs.~\eqref{eq:maxratiostateHCRB} and~\eqref{eq:maxratiostateNHCRB}
reveal that the HCRB
grows at most linearly with $d$,
whereas the NHCRB grows at most quadratically,
so that the maximum enhancement at fixed purity
(dark red line in Fig.~\ref{fig:estimateallarbitstate})
grows at most linearly with dimension
and is at most~$d+1$.

In contrast to the maximum ratio,
the minimum-ratio states
(blue dots in Fig.~\ref{fig:estimateallarbitstate})
are rank-deficient states\footnote{Although
we have treated only full-rank states until now,
rank-deficient states can be approximated arbitrarily well
by full-rank ones~\cite{Watanabe2011}.}
of the form
$\rho_\mathrm{min}^{(2)}(p)=p\ketbra{0}+(1-p)\ketbra{1}$
for purity greater than~$1/2$,
$\rho_\mathrm{min}^{(3)}(p)=p\ketbra{0}+p\ketbra{1}+(1-2p)\ketbra{2}$ for purity between~$1/3$ and~$1/2$,
and so on,\footnote{%
We only provide analytic expressions
for~$\rho_\mathrm{min}^{(2)}$
and~$\rho_\mathrm{min}^{(3)}$.}
down to~$\rho_\mathrm{min}^{(d)}$ for purity between~$1/d$ and~$1/(d-1)$.
This change in form of the minimum-ratio state
reflects as the points of non-differentiability
in the minimum-ratio curve
in Fig.~\ref{fig:estimateallarbitstate}.

\begin{figure}[H]
    \centering
    \includegraphics[width=\columnwidth]{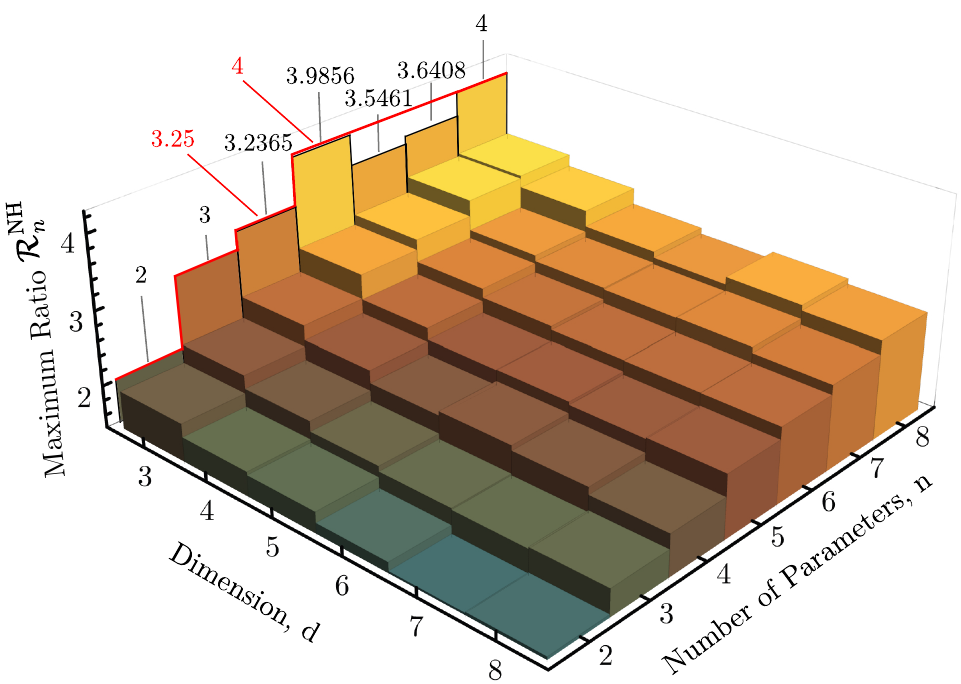}
    \caption{Maximum ratio~$\rat_n$
    between the NHCRB and the HCRB
    over 10,000 random models
    for each dimension,~$d$, from three to eight and
    for each number of parameters,~$n$, from two to eight.
    (See Fig.~\ref{fig:GridPlotDN} in Appendix~\ref{sec:appGridPlot}
    for the distribution of ratios for each~$d$ and~$n$.)
    The bar chart (with black callouts) on the back panel
    depicts the maximum ratio for estimating~$n$
    GMM coefficients from the maximally-mixed qutrit
    (Table~\ref{tab:Ratios} in Appendix~\ref{sec:AppEstimatingFewGGMMs}).
    The red line on the bar chart
    (with red callouts) represents
    the maximum ratio for each~$n$
    from known analytic models,
    applicable for all~$d\geq 3$.}
    \label{fig:estimatearbitstatearbitparas}
\end{figure}

Whereas our analytical results
in Secs.~\ref{subsec:OdRatioMaxMixState}--\ref{sec:trueratGMMmodel}
apply to the linear GMM model and its extensions,
we now consider
arbitrary smooth models
of full-rank qudit states,
i.e., we address the quantities~$\rat_n$ and~$\rat$
from Eqs.~\eqref{eq:defmaxration} and~\eqref{eq:defmaxratio}.
Specifically, we estimate~$n$
arbitrary independent parameters
from arbitrary full-rank~$d$-dimensional qudit states.
In this case,
the parameter derivatives~$\partial_j\rho_\theta$
are arbitrary
traceless Hermitian operators.
Our results for this model
are shown in Fig.~\ref{fig:estimatearbitstatearbitparas}
(and Fig.~\ref{fig:GridPlotDN} in Appendix~\ref{sec:appGridPlot}).
The 3D bar chart in Fig.~\ref{fig:estimatearbitstatearbitparas}
clearly depicts the increase in maximum ratio
with number of parameters,
in agreement with
Secs.~\ref{subsec:ration}
and~\ref{sec:EstimatingFewGGMMs}.
This suggests~$\rat_{n} \leq \rat_{n+1}$.
On the other hand,
in Fig.~\ref{fig:estimatearbitstatearbitparas},
the maximum ratio for a given number of parameters
seems to decrease with increasing
dimension---this is an anomaly
stemming from having an increasingly larger
sample space
of models but a fixed number
of samples ($10^4$) from them.
In fact, the maximum ratio cannot decrease
with increasing~$d$; any
model in~$d$ dimensions
can be extended to~$d+1$
dimensions by padding the state
and its derivatives with zeros.
This preserves both the individual
and the collective precision bounds,
and thus their ratio.

\section{Discussion}
\label{sec:Disc}

In this work,
we explored the
ratio between separable-
and collective-measurement precision bounds
in order to identify
the maximum collective quantum enhancement
in a range of settings.
Considering the ratio between
the NHCRB and the HCRB,
we established that for estimating any~$n$
independent parameters,
the maximum precision-enhancement
from collective measurements
can, in principle, be~$n$.
However, at the maximum value of~$n$,
we proved this maximum enhancement
to be only~$\mathrm{O}(d)$ or~$\mathrm{O}(\sqrt{\nmax})$.
Specifically, for the model of qudit tomography in
the Gell-Mann basis,
we proved the maximum
collective quantum enhancement to be $d+2$,
revealing the SIC POVM to be
an individual-optimal measurement
for the maximally-mixed case.
Based on the specific examples
provided and our numerical results,
we expect the attainable maximum enhancement
to be~$d+1$ instead.
We also established a maximum enhancement
of~$d+1$
for tomography in any other basis,
i.e., for estimating any other~$\nmax$
parameters,
as well as for estimating fewer than~$\nmax$
parameters
of the maximally-mixed state.
Finally, we numerically demonstrated
a maximum enhancement in~$\mathrm{O}(d)$
for states of a fixed known purity.
Our work thus suggests supplementing the known ratio
result~$\SLD \leq \CH \leq 2 \, \SLD$
with~$\CH \leq \CNH \leq (d+1) \, \CH$.
Throughout,
we have demonstrated our analytical findings
via numerics and figures for specific~$d$.

It is, however, important to note
that, though our methodology leveraging
the NHCRB as a separable-measurement precision bound
leads to several analytical results,
the NHCRB is a substitute here
for the most informative bound~$\MI$.
The NHCRB can overestimate the
best attainable separable precision
whereas~$\MI$ is tight by definition,
so the attainable maximum advantage
from collective measurements could be larger than~$d+1$,
as in Eq.~\eqref{eq:ratlessactualrat}.
However, numerical evidence for large qudit dimensions~\cite{HayashiOuyang2023}
suggests the gap~$(\MI-\CNH)/\CNH$ to be small,
so that the NHCRB-maximum ratio should be close
to the MICRB-maximum ratio~\cite{Jun2024Qest}.
Indeed, for qudit tomography in orthonormal bases,
we prove the attainable maximum collective enhancement
to be bounded between~$d+1$ and~$d+2$---the
same bound that applies for the NHCRB-maximum ratio.
We then compare this result to
a bound based on the GMCRB,
demonstrating that our approach
imposes a stronger constraint
on attainable enhancement.
Notably, both approaches predict
a maximum enhancement
linear in dimension, and not quadratic,
at the maximum number of parameters.
However,
beyond full-parameter models,
the NHCRB reflects
more accurately
a non-linear dependence of maximum collective enhancement
on number of parameters~$n$,
whereas the GMCRB predicts a linear relationship
(Fig.~\ref{fig:GMNHratioComp}).
These factors justify our choice of
the NHCRB as the individual measurement precision bound.

Our choice of the Gell-Mann basis
for tomography
was motivated by its symmetry and simplicity,
apart from being a generalisation of Pauli matrices.
The orthonormality of the basis
leads to (local) parameter orthogonality of the basis coefficients
for estimating from the maximally-mixed state~\cite{CR87},
making the 
classical and quantum (SLD)
Fisher information matrices
diagonal.
As further captured by the theory of
Fisher SIC measurements~\cite{Zhu1},
the two Fisher information matrices are proportional.
In general, any full-parameter model is D-invariant~\cite{Holevo2011,ABGG20,Jun2},
leading to~$\CH = \RLD$.
Additionally, for GMM tomography,
the SLD and RLD CRBs equal
the HCRB, which is a signature of
a locally-classical model~\cite{ABGG20,Jun2}.
In fact, the theory of quantum local
asymptotic normality~\cite{Kahn2009,Yang2019,DemkowiczDobrzaski2020,Fujiwara2023}
predicts that
in the asymptotic limit of number of copies,
this model becomes a completely classical Gaussian model~\cite{Jun2016}.
This theory therefore offers a physical
explanation for why collective enhancement
might be maximised in the orthonormal tomography setting:
the single-copy measurements
are subject to maximal parameter incompatibility
whereas all the parameters can be estimated
simultaneously in the asymptotic limit.

However, the significance of considering
the Gell-Mann basis cannot be overstated.
Generally, any parameter estimation problem
may be linearised about the true parameter values
as~$\rho_\theta\approx \rho_{\theta^*} + \sum (\theta_j-\theta_j^*)  \partial_j \rho_\theta^*$,
with the partial derivatives~$\partial_j \rho_\theta$
necessarily traceless and Hermitian,
meaning they are combinations of GMMs
(Sec.~VC in~\cite{Gill2000}).
This estimation problem can then
be linearly transformed
to the equivalent problem of
estimating some number of GMM coefficients~\cite{ABGG20},
precisely the model solved in Sec.~\ref{sec:EstimatingFewGGMMs}.
For example, the simple result of~$\CH[\rho_\mathrm{m}]=n/d$
from Sec.~\ref{subsec:OdRatioMaxMixState}
can be directly linearly transformed to
obtain a closed-form expression for
the HCRB for estimating any~$d^2-1$ independent parameters
from any full-rank qudit state.
Lastly,
although our main analysis
was specific to GMM tomography,
we proved
that the bounds
and their ratios
are invariant for tomography in any other orthonormal
basis for the same space.
These results lay a foundation for the future exploration
of the enhancement potential of entangling measurements
in multi-parameter quantum estimation.

Our approach in this work
was to study quantities~$\rat$\
and~$\ratMI$
to assess collective quantum enhancement.
The collective- and separable-optimal
precisions are identical for single-parameter problems,
but can grow increasingly farther apart
with increasing number of
parameters, disagreeing the most at the
maximum number of incompatible parameters.
Therefore, the maximum ratio of precisions
effectively
compares how parameter incompatibility
affects separable versus collective
measurement strategies.
Our results motivate defining
a finite incompatibility measure
through the gap~$\delta = (\CNH-\CH)/\CH$,
which complements
the normalised gap~$\Delta = (\CH-\SLD)/\SLD$
that
was connected to the asymptotic incompatibility
between parameters in Refs.~\cite{Carollo2019,Candeloro2024}.
In contrast to Ref.~\cite{Candeloro2024}, our results
show that even if the asymptotic incompatibility
vanishes~\cite{Carollo2019},
the individual and collective attainable precisions
can still disagree maximally
and the finite incompatibility~$\delta$
can be as large as the dimension~$d$,
a significant difference from the asymptotic case.

The results in Figs.~\ref{fig:summarymainresult},~\ref{fig:GMNHratioComp}
and~\ref{fig:estimatearbitstatearbitparas}
indicate that the maximum collective enhancement~$\rat_n$
increases with number of parameters,
and hence is largest for state tomography.
Moreover, for tomography, the maximum
enhancement
decreases with purity (Fig.~\ref{fig:estimateallarbitstate})
and is maximised by minimum-purity states.
Based on this,
we conjecture that maximum enhancement~$\rat$
is attained for orthonormal tomography of the maximally-mixed state.
This case was studied analytically
to find a ratio of~$d+1$.
Hence, we conjecture that~$\CNH \leq (d+1) \CH$
for all smooth full-rank models
in the local estimation setting, i.e.,~$\rat=d+1$.
Notably, both the~$n$ bound (Sec.~\ref{subsec:ration})
and the~$d+1$ bound
(Secs.~\ref{subsec:OdRatioMaxMixState}
\&~\ref{subsec:ArbitStates})
can be tight,
and for qubit tomography
(Eq.~\eqref{eq:PatricksEqs}),
they are tight and identical.
Resultantly, if our conjecture holds true,
we would also expect~$\rat_n \leq \min(n, d+1)$
(red line in Figs.~\ref{fig:summarymainresult} \&~\ref{fig:GMNHratioComp})
to hold for qudits.
For qutrits, Fig.~\ref{fig:summarymainresult}
shows this conjectured
upper bound to be attainable for all~$n$
except~$n=4$. The value of~$\rat_4$
for qutrits
remains an open problem
(we report a ratio of~3.25
in Fig.~\ref{fig:estimatearbitstatearbitparas}),
motivating further inquiry
into the utility of collective measurements
in multi-parameter quantum estimation.

In conclusion,
we find that for
local estimation problems
involving~$n$ parameters,
the optimal collective measurements
are at most~$n$ times more precise
than the optimal individual measurements.
Although this suggests
that a collective quantum enhancement of~$n$
is possible, and that
the utility of collective measurements
scales with the number of estimated parameters,
our further analysis indicates otherwise.
By taking the probe dimension~$d$
into account, we upper-bound
the collective enhancement
by $d+1$,
which is a tighter bound for large~$n$ ($n > d$).
Our investigation into the utility
of collective measurements
thus reveals a diminishing payoff
in the asymptotic limit.
Whereas collective measurements
on two copies, three copies, and so on,
are practically feasible
and outperform
the optimal individual measurements,
the optimal collective measurements
that saturate the HCRB require
entangling measurements on
asymptotically-large number of copies
but only enhance precision by a
factor at most linear in dimension,
underscoring their non-utility.

Our results apply to
multi-parameter quantum metrology
and quantum sensing, where
a judicious choice
between
measurement strategies
would be resource-wise beneficial.
Conversely, our work questions
the choice of the HCRB
when benchmarking the performance
of real-world quantum measurements,
and suggests the (finite-copy) NHCRB
as a more suitable alternative.
Investigating the advantage offered
by finite-copy collective measurements,
or extending to Bayesian settings
could offer valuable insights
into the potential of entangling measurements.

\section{Methods}
\label{sec:Methods}

In this section,
we present our methodology
for analytically deriving
the~$\rat_n\leq n$ and the~$\rat\leq d+1$
bounds on the maximum
collective enhancement.
%we establish some preliminary results
%based on recent work
%and introduce our model.
First, in Sec.~\ref{subsec:ration},
we establish a problem-independent upper bound of~$n$
on the collective enhancement~$\rat_n$.
Then, over Secs.~\ref{subsec:model}---\ref{sec:EstimatingFewGGMMs},
we establish the upper bound of~$d+1$
on~$\rat$.
We start by introducing the `linear GMM model'
and identifying some of its simplifying features
in Sec.~\ref{subsec:model}.
Then,
in Sec.~\ref{subsec:OdRatioMaxMixState},
we solve the full linear GMM model,
corresponding to GMM-basis tomography,
for maximally-mixed qudit states,
showing the enhancement here to be exactly~$d+1$.
Next, in Sec.~\ref{subsec:ArbitStates},
we extend our arguments
for the full linear GMM model
to arbitrary qudit states,
establishing a maximum collective enhancement of~$d+2$.
Finally, in Sec.~\ref{sec:EstimatingFewGGMMs},
we explore the related model of estimating
fewer than $\nmax$ parameters
of the maximally-mixed state, proving
that the maximum enhancement remains~$d+1$.
We also establish the maximum enhancement
to be~$d+1$ for the weighted linear GMM model
in Appendix~\ref{sec:arbitweight}
and for a different choice of
the individual-measurement precision bound
in Appendix~\ref{sec:appNHvsGM}.

\subsection{Ratio of~\texorpdfstring{$n$}{n}}
\label{subsec:ration}

We now establish a model-agnostic
(or problem-independent)
upper bound of~$n$ on
the collective enhancement~$\rat_n$.
\begin{theorem}
\label{th:ratngenstate}
For estimating~$n$ parameters of any
qudit state~$\rho_\theta$,
the collective enhancement~$\CNH[\rho_\theta] /\CH[\rho_\theta] \leq n$,
i.e.,
\begin{equation}
\label{eq:simpleupperboundonratio}
    \rat_n \leq n \, .
\end{equation}
\end{theorem}

\begin{proof}
Using Ref.~\cite{Suzuki2023}'s upper bound
(based on Ref.~\cite{HayashiOuyang2023})
to the NHCRB
(see summary in Appendix~\ref{sec:summaryJunTR}),
\begin{equation}
\label{eq:upNHBJunTR}
\begin{split}
    \CNH[\rho_\theta] \leq \min_\mathbb{X} \Bigg \{ &\ctrace ( \mathbb{Z}_\theta[\mathbb{X}]) \\
    &+   \sum_{j,k\in[n]} \left \Vert \rho_\theta [X_j, X_k] \right \Vert_1 \Bigg \} \, ,
\end{split}
\end{equation}
and~$ \left \Vert [\rho_\theta [X_j,X_k] \right \Vert_1 \leq 1/2 \qtrace [\rho_\theta(X_j^2+X_k^2)]$,
we get
\begin{equation}
    \CNH[\rho_\theta] \leq n \min_\mathbb{X}   \ctrace \left ( \mathbb{Z}_\theta[\mathbb{X}] \right ) = n \, \SLD[\rho_\theta] \, .
\end{equation}
On the other hand,
from Eq.~\eqref{eq:Holevodefn1},
we have~$\CH[\rho_\theta] \geq \min_\mathbb{X} \ctrace(\mathbb{Z}_\theta [ \mathbb{X}] )$, so that
$\CNH[\rho_\theta]/\CH[\rho_\theta] \leq n$
for all smooth models
and, thus,~$\rat_n \leq n$.
\end{proof}
Note that under the assumption
of independent parameters~$\theta$,
we have $n \leq \nmax$~\cite{Gill2000},
so that Eq.~\eqref{eq:simpleupperboundonratio} implies~$\rat \leq \nmax = d^2-1$.

\subsection{Model: Estimating GMMs from Qudits}
\label{subsec:model}

We now introduce our quantum statistical model,
which is an~$\nmax$-parameter family of
$d$-dimensional qudit states. This model,
which we call the `linear GMM model',
involves estimating the~$\nmax$
coefficients~$\{\theta_j\}_{j\in[\nmax]}$
of the GMMs~$\Lambda_d\coloneqq\{\lambda_j\}_{j\in [\nmax]}$
from the Bloch representation of a qudit state~\cite{Watanabe2011},
\begin{equation}
    \label{eq:quditblochvector}
    \rho_\theta = \mathds{1}_d/d
    + \sum_{j=1}^{\nmax} \theta_j \lambda_j \, .
\end{equation}
The GMMs~$\Lambda_d$ are traceless,
Hermitian generalisations
of the qubit Pauli operators
(see Appendix~\ref{sec:appB}),
and the decomposition in Eq.~\eqref{eq:quditblochvector}
is a one-to-one map between the Hilbert space~$\mathcal{H}_d$ of~$\rho_\theta$
and the parameter space~$\Theta \subset \mathbb{R}^{\nmax}$.
Estimating~$\theta$ is thus equivalent to
qudit state tomography.
Note that we adopt the convention of normalising the GMMs
such that~$\qtrace(\lambda_j \lambda_k)
= \delta_{jk}$.\footnote{Some authors~\cite{Bertlmann2008}
instead normalise as~$\qtrace(\lambda_j \lambda_k) = 2\delta_{jk}$
to be consistent with the~$d=2$ case for Pauli matrices.
Our convention rescales
the parameter values and bounds,
but leaves their ratios invariant.}

It is useful to summarise the~$d=2$ case results here~\cite{BacuiLi};
the HCRB and NHCRB are 
\begin{equation}
\label{eq:PatricksEqs}
\begin{split}
    \CH = \RLD = &= 3-r^2+2r \, , \\
    \CNH = \CGM &= 5-r^2+4\sqrt{1-r^2} \, 
\end{split}
\end{equation}
with~$r^2=\sum_j \theta_j^2=\qtrace(\rho^2)-\nicefrac{1}{2}$
the squared length of the Bloch vector.
In this case,
the NHCRB is attained by measuring
informationally-complete (IC) POVMs,
simplifying to
symmetric informationally-complete (SIC)
POVMs (see Eq.~\eqref{eq:defSICPOVM} for definition) at~$r=0$~\cite{BacuiLi}.
It is straightforward\footnote{%
$\CH$ increases with~$r$ whereas~$\CNH$ decreases.}
to see from Eq.~\eqref{eq:PatricksEqs} that
the ratio~$\CNH/\CH$ is maximised at~$r=0$,
corresponding to estimating parameters of
the maximally-mixed state.
Thus, for the qubit tomography model,
the maximum enhancement~$\rat[\{\rho_\theta\}]$
is three,
and this ratio is attained when
estimating the three Pauli matrix coefficients
of the maximally-mixed qubit state~\cite{BacuiLi}.

An important simplifying feature of
the linear GMM model
is that the LUB operators
$\mathbb{X}=\{X_1, \dots, X_n\}^\top$
are uniquely fixed
to be the GMMs themselves, i.e.,
\begin{equation}
\label{eq:optXsol}
    X_j = \lambda_j \, .
\end{equation}
That there is exactly one feasible
solution for the LUB operators
significantly simplifies the evaluation
of the bounds.\footnote{%
This is not generally true;
in most cases there are multiple
feasible LUB operators~$\mathbb{X}$, which need
to be optimised over to calculate
the bounds,
see for example Refs.~\cite{Lorcan21,Conlon2023} and Appendix~\ref{sec:AppEstimatingFewGGMMs}.}
To see this unique feature of our model,
consider that the true (unknown) state is
\begin{equation*}
    \rho_\theta^* =\mathds{1}_d/d + \sum_{j=1}^{\nmax} \theta_j^* \lambda_j \, .
\end{equation*}
The LUB constraints (Eq.~\eqref{eq:origunbiaseq})
at~$\theta^*$ are then
\begin{equation}
\begin{gathered}
    \qtrace(\rho_\theta X_k){\big\vert}_{\theta=\theta^*} = \theta^*_k \, , \\
       \qtrace(\partial_j \rho_\theta X_k){\big\vert}_{\theta=\theta^*} = \delta_{jk}\, .
    \label{eq:unbiasednessnew}
\end{gathered}
\end{equation}
Writing~$X_j =\sum_k c_{jk} \lambda_k$,
where~$c_{jk}$ are unknown real numbers
(to preserve Hermicity of~$X_j$),
reduces Eq.~\eqref{eq:unbiasednessnew}
to
\begin{equation*}
    \sum_j c_{kj} \theta_j^* = \theta_k^* \quad \& \quad c_{kj} = \delta_{jk},
\end{equation*}
which immediately implies~$X_j = \lambda_j$,
as claimed.

The simplification
from having a unique solution
for feasible LUB operators,~$\mathbb{X}$,
extends more generally to tomography
in any ONB ($n=\nmax$). The model here,
\begin{equation}
\label{eq:ONBmodel}
    \rho_{\theta'} = \mathds{1}_d + \sum_{j=1}^{\nmax} \theta_j' \, \mathcal{B}_j \, ,
\end{equation}
is called the ONB model, and corresponds to
tomography in the basis~$\{\mathcal{B}_j\}_{j\in[\nmax]}$,
which is orthonormal,~$\qtrace(\mathcal{B}_j \mathcal{B}_k) = \delta_{jk}$.
Further,~$\B_j$ are traceless
and Hermitian and,
therefore,
can be written as
a real linear combinations of the GMMs~$\lambda_k$
through~$\B_j = \sum_k \eta_{jk} \lambda_k$,
so that
\begin{equation}
	\begin{bmatrix}
		\B_1 \\
		\vdots\\
		\B_n
	\end{bmatrix} = \begin{bmatrix}
		\eta_{11} &  \dots & \eta_{1n} \\
		\vdots &  \ddots & \vdots \\
		\eta_{n1} & \dots & \eta_{nn}
	\end{bmatrix} \begin{bmatrix}
		\lambda_1 \\ \vdots \\ \lambda_{n}
	\end{bmatrix} \, ,
\end{equation}
or in short,
with vectorised notation~$\mathcal{B}_d$
for the left hand side,
\begin{equation}
	\label{eq:vecnotation}
	\mathcal{B}_d = (\eta \otimes \mathds{1}_d) \Lambda_d \, .
\end{equation}
Now imposing the condition
that~$\mathcal{B}_d$
is orthonormal, i.e.,
\begin{equation}
\label{eq:Bjorthonormal}
	\qtrace(\mathcal{B}_j \mathcal{B}_k) = \delta_{jk} \, ,
\end{equation}
implies that the transformation
matrix~$\eta$ is orthogonal,
so that the inner-products between
its different rows (or different columns) is zero.
This is stated and proved as
Lemma~\ref{lemma:orthtransmat} below.
\begin{lemma}
\label{lemma:orthtransmat}
For transforming from the GMM basis
to any other ONB for
the space of
$d\times d$ traceless Hermitian matrices
over reals,
the transformation matrix~$\eta$ is real
and orthogonal, meaning
\begin{equation}
    \eta \eta^\top = \eta^\top \eta = \mathds{1}_n \, .
\end{equation}
\end{lemma}

\begin{proof}
    The orthonormality condition
    from Eq.~\eqref{eq:Bjorthonormal}
    in the GMM basis simplifies to
    \begin{align*}
        \qtrace(\B_j \B_k) &= \sum_{a,b=1}^{\nmax} \eta_{ja} \eta_{kb} \qtrace (\lambda_a \lambda_b)\\
        &=  \sum_{a=1}^{\nmax} \eta_{j a} \eta_{k a} =  (\eta \eta^\top)_{j k} = \delta_{jk} \, ,
    \end{align*}
    so that~$\eta \eta^\top = \mathds{1}_n$.
    As both~$\B_d$ and~$\Lambda_d$ are bases for the space,
    the matrix~$\eta$ is full-rank and therefore invertible.
    This makes~$\eta^\top$ the inverse of~$\eta$, so that
    $\eta \eta^\top = \eta^\top \eta = \mathds{1}_n$,
    as claimed.
\end{proof}
Clearly, the GMM model is an ONB model,
and Lemma~\ref{lemma:orthtransmat}
says any two ONB models are related
by a real unitary (orthogonal) transformation
of the parameter derivatives.
Conversely, this connects the ONB parameters~$\theta'$
(in Eq.~\eqref{eq:ONBmodel})
to the GMM parameters~$\theta$
(in Eq.~\eqref{eq:quditblochvector})
through
\begin{equation*}
    \theta_k = \sum_{j=1}^{\nmax} \eta_{jk} \,  \theta_j'  \implies \theta = \eta^\top \theta' \implies \theta' = \eta \,  \theta \, .
\end{equation*}
This relation enables a much stronger connection
between different ONB models.
In particular, any ONB model
can now be considered
as a reparameterised GMM model
with~$\theta' = \eta \, \theta$~\cite{Fujiwara1999,ABGG20,Lorcan21}.
From Sec.~V of Ref.~\cite{Fujiwara1999},
we know that such a reparameterisation of any model
is equivalent to the weighted version of the original model
with weight matrix~$W = \eta^\top \eta$.
For the reparameterised GMM model,
the orthogonality of~$\eta$ (Lemma~\ref{lemma:orthtransmat})
makes this weight~$W = \eta^\top \eta = \mathds{1}_n$,
thereby proving that the GMM model
and the ONB model are completely equivalent.
This means that both individual
(NHCRB, MICRB, GMCRB)
and collective
(SLD CRB, HCRB) precision bounds
are invariant to the particular choice of ONB
for tomography.
Therefore, the ratios~$\rat[\{\rho_\theta\}]$ and~$\ratMI[\{\rho_\theta\}]$
obtained below for the GMM model
also hold for the model of tomography in any other ONB.

\subsection{Ratio of \texorpdfstring{$d+1$}{d+1}: Maximally-mixed State}
\label{subsec:OdRatioMaxMixState}

We now investigate the parameter estimation of~$\theta$
($n=\nmax$)
for the maximally-mixed qudit
state~$\rho_\theta^*=\mathds{1}_d/d \eqqcolon \rho_\mathrm{m}$
in~$d$ dimensions
(corresponding to~$\theta^*=0$).
For~$\rho_\mathrm{m}$, we calculate
the SLD and RLD CRBs, the HCRB, the NHCRB,
and the GMCRB.
We also find
the SIC-POVM in $d$ dimensions
to be an
optimal individual measurement
that attains the NHCRB,
thus establishing~$\MI = \CNH$ for this
case~\cite{Nagaoka2005a}.
Choosing~$\rho_\theta^*=\rho_\mathrm{m}$ simplifies
the evaluation of various CRBs as
this choice of~$\rho_\theta^*$
commutes with every linear operator.

From their definitions,
(see Eqs.~\eqref{eq:SLDops} and~\eqref{eq:RLDops} in Appendix~\ref{sec:appSLDRLD}),
we find both the SLD and the RLD operators
to be simply
\begin{equation}
    \mathrm{L}^{\mathrm{SLD}}_j =\mathrm{L}^{\mathrm{RLD}}_j = d \, \lambda_j \, .
\end{equation}
The two resulting QFI matrices are equal and diagonal
(see Appendix~\ref{sec:appSLDRLD}),
\begin{equation}
\label{eq:JSLDRLDQFI}
    J^{(\mathrm{SLD})} = J^{(\mathrm{RLD})} = \begin{bmatrix}
        d & 0 & \dots & 0 \\
        0 & d & \dots & \vdots \\
        \vdots & \vdots & \ddots & 0\\
        0 & \hdots & 0 & d
    \end{bmatrix}_{n\times n} \, ,
\end{equation}
which is a sign that our model is ``locally classical''~\cite{ABGG20,Jun2}.
The two QFIs then yield their respective CRBs,
\begin{equation}
\label{eq:SLDRLDmodelmaxmix}
    \SLD = \RLD = \frac{\nmax}{d} = \frac{d^2-1}{d} \, .
\end{equation}

As expected of a locally classical model,
the HCRB coincides with the SLD CRB and RLD CRB~\cite{ABGG20,Jun2}.
In fact, any full parameter model~($n=\nmax$)
with linearly-independent parameter derivatives
constitutes a ``D-invariant'' model,
for which~$\CH=\RLD$ is known to hold~\cite{Jun2,Holevo2011}.
Nonetheless, and more directly,
note that the minimisation over~$\mathbb{X}$
in the definition in Eq.~\eqref{eq:Holevodefn1}
is trivial due to the uniqueness
discussed in Sec.~\ref{subsec:model}.
Thus,~$\left (\mathbb{Z}_\theta[\mathbb{X}]\right)_{jk}=\frac{1}{d} \qtrace(\lambda_j \lambda_k) = \delta_{jk}/d$, which is exactly~${J^{(\mathrm{SLD})}}^{-1}$.
Correspondingly,
\begin{equation*}
\begin{split}
\label{eq:CHeq1}
    \CH = \ctrace(\mathbb{Z}_\theta[\mathbb{X}]) &= \ctrace ({J^{(\mathrm{SLD})}}^{-1}) \\
    &= \frac{d^2-1}{d} = \SLD  \, .
\end{split}
\end{equation*}
We write this result as Lemma~\ref{lemma:Holevo},
and defer the detailed proof to Appendix~\ref{subsec:ProofHCRB}.
Note that, more generally,
$\SLD = \min_{\mathbb{X}} \{\ctrace \left (\mathbb{Z}_\theta [ \mathbb{X}]\right )\}$.
Additionally, when~$\mathbb{X}$
is uniquely fixed,
\begin{equation}
\label{eq:CHCSLD}
     \CH \geq \ctrace \left ( \mathbb{Z}_\theta [ \mathbb{X}] \right ) = \bigtrace \left ( \mathbb{S}_\theta \mathbb{X} \mathbb{X}^\top \right ) = \SLD \, .
\end{equation}

\begin{lemma}
\label{lemma:Holevo}
The HCRB for estimating~$\theta\approx 0$ from~$\rho_\mathrm{m}$ is
\begin{equation}
    \CH[\rho_\mathrm{m}] =\frac{d^2-1}{d} .
    \label{eq:Holevo}
\end{equation}
\end{lemma}

The NHCRB is not as trivial to compute
because
despite~$\mathbb{X}$ being uniquely fixed,
there is still a minimisation over~$n d\times n d$ matrix~$\mathbb{L}$
in Eq.~\eqref{eq:HNCRBdefn}~\cite{Lorcan21}.
Moreover, directly proving the optimality
of a candidate~$\mathbb{L}$ is
difficult---for this purpose we turn to the
SDP formulation
of the NHCRB~\cite{Lorcan21}
(see Appendix~\ref{sec:ProofLemmaNHCRBmaxmix} for definition).
The SDP approach offers a simple way to prove optimality
via duality:
if we can find a primal-feasible solution
and a dual-feasible solution such that
the primal objective value
equals the dual objective value,
then the solutions are optimal.
In Appendix~\ref{sec:ProofLemmaNHCRBmaxmix},
we present a pair of such solutions
and prove their optimality using this approach.
The optimal argument~$\mathbb{L}^*$
we find to the SDP is
    \begin{equation}
    \label{eq:optLmat}
        \mathbb{L}^*_{jk} = \left ( \frac{d+1}{d+2} \right ) \big ( \{ \lambda_j, \lambda_k\} + \delta_{jk} \mathds{1}_d \big ) 
    \end{equation}
where~$j,k \in [n]$ and~$\{\, , \, \}$
denotes the anti-commutator.
Directly computing~$\mathbb{T}\mathrm{r} [ \mathbb{S}_\theta \mathbb{L}^* ]$
then leads to the following lemma.

\begin{lemma}
\label{theorem:NHCRB}
    The NHCRB for estimating~$\theta_j \approx 0$
    from~$\rho_\mathrm{m}$ is
    \begin{equation}
        \CNH[\rho_\mathrm{m}] = \frac{(d^2-1)(d+1)}{d}.
    \end{equation}
\end{lemma}

Our first main result
now follows straightforwardly
from Lemmas~\ref{lemma:Holevo} and~\ref{theorem:NHCRB}.
\begin{theorem}
\label{th:ratdmms}
For ONB tomography of the
maximally-mixed qudit state~$\rho_\mathrm{m}$, the collective enhancement~$\CNH[\rho_\mathrm{m}] /\CH[\rho_\mathrm{m}] = d+1$.
\end{theorem}

The HCRB is already known to be asymptotically attainable,
so we now prove the attainability or tightness of the NHCRB for our model.
Specifically, we show that the NHCRB in Lemma~\ref{theorem:NHCRB}
can be attained by measuring any rank-one
symmetric informationally-complete (SIC) POVM in~$d$ dimensions
(assuming one exists).
The SIC POVM is a set of~$d^2$
measurement operators~$\{\Pi_l\}_{l\in [d^2]}$
that form a POVM and are
completely symmetric between themselves under
the trace inner product,
\begin{equation}
\label{eq:defSICPOVM}
    \qtrace \left ( \Pi_{l_1} \Pi_{l_2} \right ) = \frac{1}{d^2(d+1)} \;  \quad \forall \; l_1 \neq l_2, \, l_1, l_2 \in [d^2] .
\end{equation}

To prove that measuring SIC POVMs
attains the NHCRB,
we show that the measured probabilities~$\qtrace(\rho_\mathrm{m} \Pi_l)$
directly yield a variance
equal to~$\CNH$ from Lemma~\ref{theorem:NHCRB},
establishing~$\MI=\CNH$ in this case.
The CFI matrix~$J_{jk}$ ($j,k\in [n]$),
which in the multi-parameter case is given by
\begin{equation}
    \label{eq:QFImatrix11}
    J_{jk} \left [\{\Pi_l\}_{l\in[m]}\right] = \sum_{l=1}^{m} \frac{ \qtrace  \left [ \partial_{j} \rho_\theta \Pi_l \right ] \qtrace  \left [ \partial_{k} \rho_\theta \Pi_l \right ]}{\qtrace \left [ \rho_\theta \Pi_l \right ] } \, ,
\end{equation}
simplifies to (see Lemma~\ref{th:CFI}
and proof in Appendix~\ref{subsec:ProofCFI})
\begin{equation}
\begin{split}
\label{eq:CFINagaoka}
    J_{jk} &=d^2 \sum_{l=1}^{d^2} \qtrace [ \lambda_j \Pi_l] \qtrace [ \lambda_k \Pi_l] \\
    &= \delta_{jk} \; \frac{d}{d+1}
\end{split}
\end{equation}
in this case, so that Eq.~\eqref{eq:CFICRB} then leads to
\begin{equation}
\label{eq:SICCFIoptimal}
    \ctrace(J^{-1}) = \frac{(d^2-1)(d+1)}{d} = \CNH \, .
\end{equation}
From Eq.~\eqref{eq:CMIdef}, we then have
$\MI \leq \ctrace(J^{-1}) = \CNH \leq \MI$
with the last inequality from Eq.~\eqref{eq:orderingprecisionbounds}.
This proves~$\MI=\CNH$, meaning that the
ultimate individual precision
is attained for this model
by measuring SIC POVMs.
Notably, any rank-one SIC POVM
in~$d$ dimensions,
irrespective of its orientation,
constitutes an optimal individual measurement
in this scenario.

An alternative proof of this attainability
can be furnished using the GM inequality
for individual measurements~\cite{Gill2000},
\begin{equation}
\label{eq:GMinequality}
    \ctrace[ {J^{(\mathrm{SLD})}}^{-1} J] \leq d-1 \, .
\end{equation}
For the SLD QFI~$J^\mathrm{(SLD)}$
in Eq.~\eqref{eq:JSLDRLDQFI},
the inequality in Eq.~\eqref{eq:GMinequality}
implies that
\begin{equation}
    \ctrace(J^{-1})\geq (d^2-1)(d+1)/d \, ,
\end{equation}
which is saturated by the SIC POVM
CFI~$J$ from Eq.~\eqref{eq:CFINagaoka},
as seen in Eq.~\eqref{eq:SICCFIoptimal}.

\subsection{Ratio of~\texorpdfstring{$d+2$}{d+2}: Extension to Arbitrary States}
\label{subsec:ArbitStates}

In this section, we extend Lemmas~\ref{lemma:Holevo}
\&~\ref{theorem:NHCRB} and Theorem~\ref{th:ratdmms}
for~$\rho_\mathrm{m}$ to arbitrary qudit states~$\rho_\theta\neq\rho_\mathrm{m}$.
Such a qudit state can still be written
as in Eq.~\eqref{eq:quditblochvector},
but now the true parameter values~$\theta^*$ are non-zero
and~$\theta\approx\theta^*$.
In this case, we show that
$\CH[\rho_\theta] \geq \CH[\rho_\mathrm{m}]  - \sum_{j\in [n]} {\theta_j^*}^2$
and that~$\CNH[\rho_\theta] \leq \CNH[\rho_\mathrm{m}] - \sum_{j\in [n]} {\theta_j^*}^2$,
which, we then show, imply
\begin{equation*}
    \frac{\CNH[\rho_\theta]}{\CH[\rho_\theta] } \leq d+2 \, .
\end{equation*}
This establishes the maximum collective
quantum enhancement~$\rat[\{\rho_\theta\}]$
for the linear GMM model
to be~$d+2$.
We also argue that
the optimal individual measurements
are now IC POVMs, supported
by numerical results in Appendix~\ref{sec:appOptICPOVMs}.

The HCRB and the NHCRB
involve an additional~$-\ctrace(\theta^* {\theta^*}^\top) = -\sum_{j\in[n]} {\theta^*_j}^2$
term for non-zero~$\theta^*$
(Eqs.~\eqref{eq:HNCRBdefn},~\eqref{eq:Holevodefn1}~\&~\eqref{eq:Holevodefn2}).
For the HCRB, it is simple to see
from Eq.~\ref{eq:lowerbound}
in Appendix~\ref{subsec:ProofHCRB}
that
$\bigtrace(\mathbb{S}_\theta \mathbb{X}\mathbb{X}^\top)$
still lower-bounds~$\min_{\mathbb{L},\mathbb{X}} \mathbb{T}\mathrm{r}(\mathbb{S}_\theta \mathbb{L})$
(see also Remark~\ref{re:HCRBlowerbound} in Appendix~\ref{subsec:ProofHCRB})
so that
\begin{equation}
\label{eq:CHarbitstateineq}
    \CH[\rho_\theta] \geq \mathbb{T}\mathrm{r} ( \mathbb{S}_\theta \mathbb{X}\mathbb{X}^\top) \; -\sum_{j\in[n]} {\theta^*_j}^2 \,
\end{equation}
despite~$\mathbb{L} =\mathbb{X}\mathbb{X}^\top$ not
being the optimal solution anymore.
Note also that the purity of the true state~$\rho_{\theta}^*$ is
\begin{equation*}
    \mathrm{P}(\rho_\theta^*) = \qtrace((\rho_\theta^*)^2) = 1/d + \sum_{j\in [n]} {\theta_j^*}^2.
\end{equation*}
By explicit calculation, we find $\bigtrace ( \mathbb{S}_\theta \mathbb{X}\mathbb{X}^\top) = \nicefrac{d^2-1}{d} = \CH[\rho_\mathrm{m}]$,
and hence,
\begin{equation}
\label{eq:lowerboundHCRB}
    \CH[\rho_\theta] \geq \frac{d^2-1}{d} -\sum_{j\in[n]} {\theta^*_j}^2 = d - \mathrm{P}(\rho_\theta^*) \, .
\end{equation}
From Lemma~\ref{lemma:Holevo}, we know that
this inequality is saturated by the maximally-mixed
state~$\rho_\mathrm{m}$, which has purity~$1/d$.
Figure~\ref{fig:HolveoBoundLowerBound} (a) depicts how $d - \mathrm{P}(\rho_\theta^*)$
compares with the actual HCRB for qutrit states.

\begin{figure*}
    \centering
    \includegraphics[width=\textwidth]{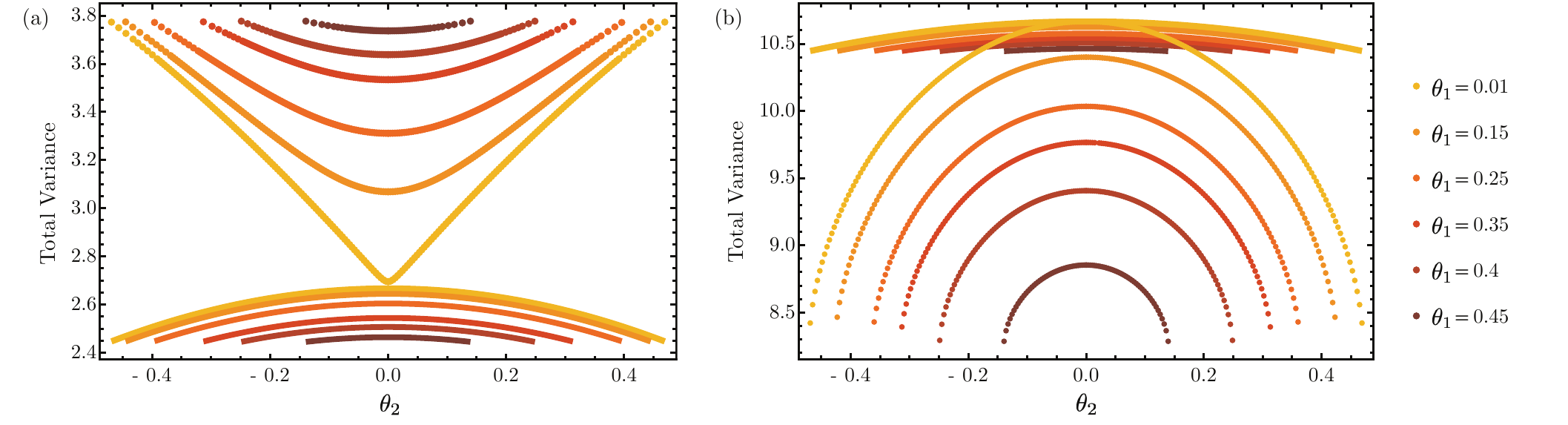}
    \caption{Comparison of the HCRB and the NHCRB
    to their lower and upper bounds, respectively.
    (a) HCRB and its lower bound~$d-\mathrm{P}(\rho_{\theta})$
    (from Eq.~\eqref{eq:lowerboundHCRB}).
    The lower solid parabolic curves show
    the lower bound and
    the upper triangular curves
    (beginning and ending with dots)
    show the numerically-computed HCRB.
    (b) NHCRB and its upper bound~$d^2+d-1-\mathrm{P}(\rho_{\theta})$
    (from Eq.~\eqref{eq:upperboundNHCRB}).
    The lower dotted curves
    show the numerically-computed
    NHCRB and the upper solid curves
    show the upper bound.
    The state chosen in both
    (a) and (b) is a mixed
    qutrit~$\rho_{\theta}=\mathds{1}_d/d + \theta_1 \lambda_2 + \theta_2 \lambda_4$.}
    \label{fig:HolveoBoundLowerBound}
\end{figure*}

For the NHCRB,
our key insight is that the optimal argument
$\mathbb{L}^*_{jk} = \nicefrac{d+1}{d+2} \left ( \{ \lambda_j, \lambda_k\} + \delta_{jk} \mathds{1}_d \right )$
from Lemma~\ref{lemma:primalsolutionNHCRB} in Sec.~\ref{subsec:OdRatioMaxMixState}
is still feasible:
$\mathbb{L}^*$
satisfies the constraints $\mathbb{L}_{jk} = \mathbb{L}_{kj}$ Hermitian and $\mathbb{L}  \succcurlyeq \mathbb{X} \mathbb{X}^\top$,
which are all independent of~$\rho_\theta$.
However,~$\mathbb{L}^*$ is not optimal
so~$\bigtrace[\mathbb{S}_\theta \mathbb{L}^*]$
only upper-bounds~$\min_{\mathbb{L}}\bigtrace [\mathbb{S}_\theta \mathbb{L}]$ in Eq.~\eqref{eq:HNCRBdefn}.
Again, we explicitly calculate~$\bigtrace[\mathbb{S}_\theta \mathbb{L}^*]$
to find
\begin{equation}
\label{eq:calSthetatrace}
    \bigtrace[\mathbb{S}_\theta \mathbb{L}^*] = \frac{(d^2-1)(d+1)}{d} = \CNH[\rho_\mathrm{m}]
\end{equation}
so that we can upper-bound the NHCRB as
\begin{equation}
\label{eq:upperboundNHCRB}
\begin{gathered}
    \CNH[\rho_\theta] \leq \frac{(d^2-1)(d+1)}{d} - \sum_j {\theta^*_j}^2\\
    = d^2+d-1-\mathrm{P}(\rho_{\theta^*}) \, .
\end{gathered}
\end{equation}
From Lemma~\ref{theorem:NHCRB},
we see that the inequality in Eq.~\eqref{eq:upperboundNHCRB}
is saturated
by the maximally-mixed state~$\rho_\mathrm{m}$.
Figure~\ref{fig:HolveoBoundLowerBound} (b)
depicts how $d^2+d-1-\mathrm{P}(\rho_\theta^*)$ compares
with the actual NHCRB for qutrit states.

\begin{theorem}
\label{th:ratdarbstate}
    For ONB tomography of arbitrary~$d$-dimensional qudit state~$\rho_\theta$, the maximum collective enhancement~$\rat[\{\rho_\theta\}]  \leq d+2$. 
\end{theorem}
\begin{proof}
Combining the lower bound for the HCRB in Eq.~\eqref{eq:lowerboundHCRB}
with the upper bound for the NHCRB in Eq.~\eqref{eq:upperboundNHCRB},
we get
\begin{equation}
\label{eq:ratioeq2}
    \frac{\CNH[\rho_\theta]}{\CH[\rho_\theta]} \leq \frac{d^2+d-1-\mathrm{P}(\rho_\theta)}{d-\mathrm{P}(\rho_\theta)} .
\end{equation}
Then, using $1/d\leq \mathrm{P}(\rho_\theta)\leq 1$,
we find the maximum of the right hand side of Eq.~\eqref{eq:ratioeq2}
to be~$d+2$, attained when~$\mathrm{P}(\rho_\theta)=1$,
i.e., when~$\rho_\theta$ is pure.
\end{proof}

Theorem~\ref{th:ratdarbstate} establishes a
loose upper bound that we expect
to never be attained.
The ratio between the two bounds
is actually maximised
by maximally-mixed states
and minimised by pure states,
as seen in Fig.~\ref{fig:estimateallarbitstate}.
This discrepancy is due to
the HCRB lower bound
in Eq.~\eqref{eq:CHarbitstateineq}
being a decreasing function of purity,
whereas numerical results
(Fig.~\ref{fig:estimateallarbitstate} and
weighted version in Appendix~\ref{sec:arbitweight})
and analytical results for the qubit case (Eq.~\eqref{eq:PatricksEqs})
show the HCRB to be an increasing function of purity
for this model.
Figure~\ref{fig:HolveoBoundLowerBound}
(a) and (b)
numerically compare
the lower and upper bounds
from Eqs.~\eqref{eq:lowerboundHCRB}
and~\eqref{eq:upperboundNHCRB}
with the true HCRB and NHCRB,
respectively,
for~$d=3$ and~$\rho_\theta$ close to~$\rho_\mathrm{m}$.
It is evident
that the lower and upper bounds
are valid for all~$\rho_\theta$
but saturated only for~$\rho_\mathrm{m}$.
It is also clear from
Fig.~\ref{fig:HolveoBoundLowerBound} (a)
that the HCRB lower bound in Eq.~\eqref{eq:lowerboundHCRB}
is a decreasing function of purity
whereas
the true HCRB is an increasing function of purity.
And, in fact, using~$\CH[\rho_\theta]\geq \CH[\rho_\mathrm{m}]$
instead of Eq.~\eqref{eq:CHarbitstateineq}
in the proof of Theorem~\ref{th:ratdarbstate}
leads to an upper bound of~$d+1$ on the ratio.
This is a tight bound and is saturated
by the model considered in Sec.~\ref{subsec:OdRatioMaxMixState}.
Nonetheless, Theorems~\ref{th:ratdmms}
and~\ref{th:ratdarbstate} together establish
that for the
linear GMM model,
and any other ONB model,~$\rat[\{\rho_\theta\}] \leq d+2$.

We extend the upper bounds
on the ratio
proved in Secs.~\ref{subsec:OdRatioMaxMixState}
and~\ref{subsec:ArbitStates}
to arbitrary, full-rank,
parameter-independent weight matrices~$W$
in Appendix~\ref{sec:arbitweight}.
Arbitrary weight matrices
correspond to reparameterisations
of the model~\cite{Lorcan21,Fujiwara1999,ABGG20},
i.e., estimating parameters
that are not coefficients of
any particular ONB.
Our results in Appendix~\ref{sec:arbitweight}
prove that for the maximally-mixed state~$\rho_\mathrm{m}$,
and for estimating any~$\nmax$
independent parameters locally,
the maximum collective enhancement is at most~$d+1$.
We also numerically show that
when estimating from any other state~$\rho_\theta$,
the ratio is smaller
than when estimating from~$\rho_\mathrm{m}$
with the same weight~$W$.
This suggests the maximum
enhancement
from collective measurements
over individual measurements
in any local tomography problem is~$d+1$,
i.e.,
\begin{equation}
      \rat_{\nmax} = \max_{\nmax\text{-parameter models}} \; \rat \left [\{\rho_\theta \} \right ] \leq d+1 \, ,
\end{equation}
and this upper bound is saturated by the model
studied in Sec.~\ref{subsec:OdRatioMaxMixState}.

Finally, the optimal separable measurements,
assuming all the~$\theta_j$ to be independent,
are IC-POVMs.
This is because to estimate~$d^2-1$ independent
parameters, one needs~$d^2-1$ independent probabilities
which can only arise from
measuring a POVM with
at least~$d^2$ linearly independent elements.
Having any more than~$d^2$ POVM elements
is also redundant, as the extra elements
cannot be linearly independent from the
first~$d^2$ elements.
In Appendix~\ref{sec:appOptICPOVMs},
we depict the transition
from SIC POVMs to IC POVMs
as the purity of~$\rho_\theta$
increases from~$1/d$ for the
maximally-mixed state to~$1$
for pure states.
This result is in line with previous findings
that IC POVMs are optimal for state estimation and
tomography~\cite{DAriano2000,DAriano2004}.

\subsection{Related Model: Estimating \texorpdfstring{$n < \nmax$}{n < nmax} GMMs}
\label{sec:EstimatingFewGGMMs}

In Secs.~\ref{subsec:OdRatioMaxMixState}
and~\ref{subsec:ArbitStates},
we studied the full-parameter linear GMM model
for the cases~$\theta^*=0$
and~$\theta^*\neq 0$.
In this section, we study the GMM model
with~$n<\nmax$ parameters,
with the remaining~$\nmax-n$ parameters
set to zero,
i.e., estimating~$\{\theta_j\}_{j\in [n]} \in\Theta$ from
\begin{equation}
    \rho_\theta = \mathds{1}_d/d 
    + \sum_{j\in [n]} \theta_j \lambda_j \, .
\end{equation}
The case with~$\{\theta_j\}_{j\in[\nmax]\setminus[n]} \neq 0$
is also interesting
but we do not study that here.
Moreover, we only provide analytic results
for the true state~$\rho_\theta^* = \rho_\mathrm{m}$.
This is because,
numerically, we see that
when the parameters not estimated
are set to zero, the ratio is maximised by~$\rho_\mathrm{m}$.
Although we have not specified
which~$n$ GMMs we choose to estimate,
and despite the NHCRB
(but not the HCRB)
depending on this choice,\footnote{See
Table~\ref{tab:Ratios}
in Appendix~\ref{sec:AppEstimatingFewGGMMs}}
the bounds we provide
on the ratio are independent of
this choice.

The differentiating factor for this model is
that~$X_j=\lambda_j$ are not the
sole possible LUB operators.
Nonetheless,
for the HCRB, this choice
is still optimal,
and the HCRB is~$n/d$,
the same form as Lemma~\ref{lemma:Holevo}.
For the NHCRB,
we find that linearly modifying~$\Lambda_d$
and bilinearly modifying~$\mathbb{L}_{jk}^*$
from Lemma~\ref{lemma:primalsolutionNHCRB}
gives us an upper bound of~$(d+1)n/d$,
also the same form as Lemma~\ref{theorem:NHCRB}.

For estimating parameters of~$\rho_\mathrm{m}$,
we find the SLD-optimal~$\mathbb{X}$,
i.e.,~$\arg \min_{\mathbb{X}} \bigtrace [ \mathbb{S}_\theta \mathbb{X} \mathbb{X}^\top]$,
to be~$X_j=\lambda_j$ for~$j\in[n]$
(see Appendix~\ref{sec:AppEstimatingFewGGMMs}).
This establishes that~$\SLD=n/d$.
Moreover, for this choice of~$\mathbb{X}$,
$\mathbb{L}=\mathbb{X}\mathbb{X}^\top$
satisfies all the HCRB constraints
(Eq.~\eqref{eq:Holevodefn2})
and yields
\begin{equation*}
\bigtrace[\mathbb{S}_\theta \mathbb{L} ]=
\bigtrace [ \mathbb{S}_\theta \mathbb{X} \mathbb{X}^\top] = \frac{n}{d} \, ,
\end{equation*}
so that~$\CH[\rho_\mathrm{m}]=\SLD[\rho_\mathrm{m}]=n/d$, as claimed.

For the NHCRB, the~$X_j$
are linear combinations of the~$d^2-1$ GMMs,
and can be written as

\begin{equation}
    \mathbb{X} = \mathbb{C}^{(2)} \Lambda_d \, ,
\end{equation}
where~$\mathbb{C}^{(2)}$ is a real matrix.
Then, $\mathbb{X}\mathbb{X}^\top = \mathbb{C}^{(2)} \Lambda \Lambda^\top {\mathbb{C}^{(2)}}^\top$.
We similarly modify~$\mathbb{L}^{*}$
from Eq.~\eqref{eq:optLmat}
to define~$\mathbb{L}^{**} \coloneqq  \mathbb{C}^{(2)} \mathbb{L}^* {\mathbb{C}^{(2)}}^\top$,
which ensures~$\mathbb{L}^{**} \succcurlyeq \mathbb{X}\mathbb{X}^\top$
because of~$\mathbb{L}^{*} \succcurlyeq \Lambda \Lambda^\top$
from Lemma~\ref{lemma:primalsolutionNHCRB}
in Appendix~\ref{sec:ProofLemmaNHCRBmaxmix}.
The NHCRB in Eq.~\eqref{eq:HNCRBdefn}
then becomes a minimisation
over~$\mathbb{L}$ and~$\mathbb{C}^{(2)}$.
However, if we choose our ansatz~$\mathbb{L}^{**}$
for~$\mathbb{L}$
and minimise only over~$\mathbb{C}^{(2)}$,
we should get a larger value,
i.e.,
\begin{equation}
\label{eq:inequality2}
    \begin{gathered}
        \CNH[\rho_\mathrm{m}] = \min_{\mathbb{L}, \mathbb{C}^{(2)}} \Big \{ \mathbb{T}\mathrm{r} [ \mathbb{S}_\theta \mathbb{L}] \, \vert \, \mathbb{L}_{jk} = \mathbb{L}_{kj}  \, , \\
         \quad \quad \mathbb{L}_{jk} \, \mathrm{Hermitian} \, , \\
        \qquad \quad \quad \; \; \mathbb{L} \succcurlyeq  \mathbb{C}^{(2)} \; \Lambda \Lambda^\top  {\mathbb{C}^{(2)}}^\top \Big \} \, \\
        \leq \min_{\mathbb{C}^{(2)}} \Big \{ \mathbb{T}\mathrm{r} [ \mathbb{S}_\theta \mathbb{L}^{**} ] \; \vert \;
        \mathbb{L}^{**} =  \mathbb{C}^{(2)} \; \Lambda \Lambda^\top  {\mathbb{C}^{(2)}}^\top \Big \}\\
        = \frac{(d+1)n}{d} \, .
    \end{gathered}
\end{equation}
The inequality in Eq.~\eqref{eq:inequality2} holds
because the second minimisation
is performed over a subset of the set
over which the first minimisation is performed
and the last equality follows after some algebra
(see Appendix~\ref{sec:AppEstimatingFewGGMMs}).
Combining this upper bound on the NHCRB
with~$\CH = n/d$,
we get the following theorem.
\begin{theorem}
\label{th:estfewGMMs}
    For estimating fewer-than-$\nmax$ coefficients
    of GMMs of the maximally-mixed qudit state~$\rho_\mathrm{m}$,
    the collective enhancement~$\CNH[\rho_\mathrm{m}]/\CH[\rho_\mathrm{m}]\leq d+1$.
\end{theorem}
Numerically,
we see this ratio actually depends on~$n$:
as $n$ increases from~2 to~$d^2-1$,
the ratio increases from~2 to~$d+1$.
Table~\ref{tab:Ratios} in Appendix~\ref{sec:AppEstimatingFewGGMMs}
depicts this increase,
though not monotonic,
by listing the two bounds and their ratios for~$d=3$.
As proven here, the HCRB only depends on~$n$ and~$d$.
Interestingly,
when the true values of the parameters
not being estimated are non-zero,
the maximally-mixed state is no longer
the ratio-maximising state.
However, the~$n$ bound in
Sec.~\ref{subsec:ration}
and numerical results in
Sec.~\ref{sec:numericalresults}
suggest that the maximum enhancement~$\rat_n$
cannot decrease with increasing~$n$ at fixed~$d$,
meaning~$\rat_n\leq\rat_{n+1}$.
And we have analysed the~$n=\nmax$
case in depth, so we expect
that for any~$n < \nmax$ model,
the same bound
of~$d+1$ should hold,
i.e.,~$\rat_n \leq \rat_{\nmax} \leq d+1$.

\subsection{True Ratio for GMM Model using MICRB}
\label{sec:trueratGMMmodel}

To establish an upper bound to
the true collective enhancement~$\ratMI[\{\rho_\theta\}]$,
we use the MICRB,
which quantifies the optimal precision
attainable with separable measurements~\cite{Nagaoka2005b,Hayashi1997,HayashiOuyang2023}.
The MICRB is formulated in Ref.~\cite{HayashiOuyang2023}
as a conic optimisation over an~$(n+1)d \times (n+1)d$-sized
operator~$X$ that lies
in a separable cone, $\mathcal{S}_{\text{SEP}}$,
within the space of positive operators of this size
(see definition in Eq.~\eqref{eq:CMIdef2}
in Appendix~\ref{sec:appMICRB}).
Additionally, the operator~$X$ satisfies linear trace constraints
arising from the POVM condition
and from the local unbiasedness condition
(Eqs.~\eqref{eq:const1},~\eqref{eq:const2}
in Appendix~\ref{sec:appMICRB}).
Subject to these constraints,
the MICRB minimises
the objective~$\bigtrace[(W \otimes \rho_\theta) X]$,
where the typically~$n\times n$ weight matrix~$W$
is modified to be~$(n+1) \times (n+1)$-sized
by padding its first row and column with zeros
%where~$W$ is now an~$(n+1)\times(n+1)$-sized
%weight matrix whose first row and column are zero
($W_{1j} = W_{j1} = 0$ for $j\in[n+1]$).

It is important to note that
one of the main contributions
of Ref.~\cite{HayashiOuyang2023}
was showing that various
precision bounds, including the NHCRB
and the HCRB, can all be recast as
the minimisation of this same objective
but subject to different inclusion constraints
on operator~$X$.
In particular, the NHCRB was shown to
require the positive semi-definiteness of~$X$,
whereas the MICRB was shown to require~$X$
to also be separable,
over an~$(n+1)$-dimensional real space
and a~$d$-dimensional complex space.
Our key observation is to
find that for the linear GMM model,
at the maximally-mixed state,
a solution to the MICRB is given by
\begin{equation}
    X_\mathrm{sol} \coloneqq \begin{pmatrix}
        \mathds{1}_d & \lambda_1 & \lambda_2 & \dots & \lambda_n \\
        \lambda_1 & \mathbb{L}^*_{11} & \mathbb{L}^*_{12}    & \dots  & \mathbb{L}^*_{1n}   \\
        \lambda_2 & \mathbb{L}^*_{21} & \mathbb{L}^*_{22}    & \dots  & \mathbb{L}^*_{2n}   \\
        \vdots & \vdots & \vdots & \ddots & \vdots \\
        \lambda_n & \mathbb{L}^*_{n1} & \mathbb{L}^*_{n2}  &  \dots  & \mathbb{L}^*_{nn}
    \end{pmatrix} \, ,
\end{equation}
where~$\mathbb{L}^*_{jk}$
is the NHCRB-optimal argument
presented in Eq.~\eqref{eq:optLmat} of
Sec.~\ref{subsec:OdRatioMaxMixState}.
This connection between the NHCRB solution
and the MICRB solution is not surprising,
given that two bounds
minimise the same objective subject to different constraints,
and are equal for the linear GMM model
at the maximally-mixed state.

In Appendix~\ref{sec:appMICRB},
we rigorously prove that~$X_\mathrm{sol}$
satisfies all the MICRB constraints,
including the POVM constraint,
the LUB constraints
and the separability constraint
(Lemma~\ref{lemma:proofSeparableCone}).
The MICRB objective value
for~$X_\mathrm{sol}$ equals
the NHCRB,~$\CNH[\rho_\mathrm{m}] = n(d+1)/d$,
as expected.
Further, we show that~$X_\mathrm{sol}$
is a feasible candidate
(but not necessarily optimal)
for the MICRB for arbitrary qudit states,
thereby leading to the upper-bound,
\begin{equation}
    \MI[\rho_\theta] \leq n \left(\frac{d+1}{d}\right) - \sum_j \theta_j^2 \, ,
\end{equation}
which was derived for
the NHCRB in Eq.~\eqref{eq:upperboundNHCRB}
of Sec.~\ref{subsec:ArbitStates}.
Combining this upper bound
for the MICRB
with the lower bound~$\CH[\rho_\theta] \geq n/d - \sum_j \theta_j^2$,
we prove in Theorem~\ref{th:ratdMIarbstate}
of Appendix~\ref{sec:appMICRB}
that the true maximum ratio~$\ratMI[\{\rho_\theta\}]$
in the linear GMM model
is upper-bounded by~$d+2$.
%\begingroup
\begin{theorem}
        For ONB tomography of arbitrary~$d$-dimensional qudit state~$\rho_\theta$, the maximum true collective enhancement~$\ratMI[\{\rho_\theta\}]
        %= \max_{\theta \in \Theta} \MI[\rho_\theta]/\CH[\rho_\theta] 
        \leq d+2$.
\end{theorem}
\addtocounter{theorem}{-1}
%\endgroup
This theorem
reproduces the upper bound
derived on the maximum ratio~$\rat[\{\rho_\theta\}]$
in Theorem~\ref{th:ratdarbstate}
and similarly extends to tomography in any ONB.
However,
as with Theorem~\ref{th:ratdarbstate},
we do not expect the upper-bound of~$d+2$
to be attainable.

\section*{Data availability}
The data that support
the findings of this study
are available from the
corresponding author upon reasonable request.

\section*{Code availability}
The codes that support
the findings of this study
are available from the
corresponding author upon reasonable request.

%\bibliographystyle{aps}

% Bibliography
\bibliography{mybibtex} %bibliography

\section*{Acknowledgements}
This research was funded by the Australian
Research Council Centre of Excellence CE170100012.
This research was also supported by A*STAR C230917010, Emerging Technology and A*STAR C230917004, Quantum Sensing.
JS is partially supported by
JSPS KAKENHI Grant Numbers JP21K11749, JP24K14816.
We are grateful to the National Computational Infrastructure (NCI)
for their super-computing resources that
were used for numerical investigations.

\section*{Author Contributions}
L.O.C. conceived the project.
A.D., L.O.C., J.S., S.K.Y. and S.M.A. developed the theory
and designed the numerical experiments.
A.D. and L.O.C. wrote the proofs
and A.D. performed
the numerical simulations.
A.D., L.O.C. and S.M.A. wrote the manuscript.
All authors contributed to discussions
regarding the results in this paper.
S.M.A. and P.K.L. supervised the project.

\section*{Competing Interests}
All authors declare no
financial or non-financial
competing interests.

\appendix

\section{Gell-Mann Matrices and Tomography via Parameter Estimation}
\label{sec:appB}

The main advantage of the Bloch representation
for qubits,
\begin{equation}
    \rho = \frac{1}{2} \left ( \mathds{1}_2 + \sum_{j\in\{x,y,z\}} \theta_j \sigma_j \right ) \, ,
    \label{eq:blochvectordecomp}
\end{equation}
where~$\mathcal{P}\coloneqq \{\sigma_x, \sigma_y, \sigma_z\}$
is the Pauli basis,
is the convenience of working with
the real-valued Bloch vector~$\theta\coloneqq \{\theta_x, \theta_y, \theta_z\}\in\mathbb{R}^{3}$
instead of the equivalent
complex operator~$\rho\in\mathbb{C}^{2\times 2}$.
The same convenience is availed in 
three dimensions by replacing~$\mathcal{P}$
with the GMMs,
$\Lambda_3\coloneqq\{\lambda_j\}_{j=1}^8$.
These constitute an orthonormal basis
over the reals
for the space of~$3\times 3$
traceless Hermitian matrices
and generalise the Pauli matrices to three dimensions.
So for a qutrit state~$\rho$, we can write
\begin{equation}
    \label{eq:qutritblochvector}
    \rho = \mathds{1}_3/3 + \sum_{j=1}^8 \theta_j \lambda_j
\end{equation}
with
\begin{equation}
\begin{gathered}
    \lambda_1 = \frac{1}{\sqrt{2}} \begin{bmatrix}0 & 1 & 0 \\ 1 & 0 & 0 \\ 0 & 0 & 0\end{bmatrix}, \quad
    \lambda_2 =\frac{1}{\sqrt{2}} \begin{bmatrix} 0 & -i & 0 \\ i & 0 & 0 \\ 0 & 0 & 0 \end{bmatrix}, \; \\
    \lambda_3 = \frac{1}{\sqrt{2}}\begin{bmatrix} 1 & 0 & 0 \\ 0 & -1 & 0 \\ 0 & 0 & 0 \end{bmatrix}, \quad
    \lambda_4 =\frac{1}{\sqrt{2}} \begin{bmatrix} 0 & 0 & 1 \\ 0 & 0 & 0 \\ 1 & 0 & 0 \end{bmatrix},\; \\
    \lambda_5 = \frac{1}{\sqrt{2}}\begin{bmatrix} 0 & 0 & -i \\ 0 & 0 & 0 \\ i & 0 & 0 \end{bmatrix}, \quad
    \lambda_6 = \frac{1}{\sqrt{2}} \begin{bmatrix} 0 & 0 & 0 \\ 0 & 0 & 1 \\ 0 & 1 & 0 \end{bmatrix},\\
    \lambda_7 =\frac{1}{\sqrt{2}} \begin{bmatrix} 0 & 0 & 0 \\ 0 & 0 & -i \\ 0 & i & 0 \end{bmatrix}, \;
    \lambda_8 =\frac{1}{\sqrt{6}} \begin{bmatrix} 1 & 0 & 0 \\ 0 & 1 & 0 \\ 0 & 0 & -2 \end{bmatrix}  .
\end{gathered}
\label{eq:explicitGMMs}
\end{equation}
Note that
we choose a different convention in Eq.~\eqref{eq:qutritblochvector}
from that in
Eq.~\eqref{eq:blochvectordecomp},
and we set~$\Tr(\lambda_j \lambda_k) = \delta_{jk}$
instead of the standard~$2 \delta_{jk}$ in Eq.~\eqref{eq:explicitGMMs} for convenience.

The eight GMMs in Eq.~\eqref{eq:explicitGMMs}
for~$d=3$
can be extended to~$d>3$
leading to the generalised GMMs~$\Lambda_d$
(that we shall also refer to as GMMs).
In fact,~$\Lambda_d$ consists of~$\binom{d}{2}$ real, symmetric matrices that generalise~$\sigma_x$,
$\binom{d}{2}$ imaginary, skew-symmetric matrices that generalise~$\sigma_y$,
and~$d-1$ real, diagonal matrices that generalise~$\sigma_z$.
In total, we have
$d^2-1$ matrices,~$\{\lambda_j\}_{j=1}^{d^2-1}$,
in~$\Lambda_d$, and,
for arbitrary qudit density matrix~$\rho$ in~$d$ dimensions,
we can write
\begin{equation}
    \label{eq:quditblochvectorapp}
    \rho = \mathds{1}_d/d + \sum_{j=1}^{d^2-1} \theta_j \lambda_j
\end{equation}
to get a one-to-one map between~$\rho\leftrightarrow\theta$.
Resultantly, a qudit state estimation or tomography problem
can be treated as a parameter estimation problem
with~$\theta$ as the unknown parameter.
Note that our convention
in Eq.~\eqref{eq:quditblochvectorapp} is different from that used in some existing literature~\cite{Bertlmann2008}
but
is equivalent up to a re-scaling of the parameters,
which leaves the ratio unchanged.

\section{Proof of Generalised Gell-Mann Matrix Identities}
\label{sec:appIdentities}

In this appendix, we prove the following identities
for~$\Lambda_d=\{\lambda_j\}_{j=1}^{d^2-1}$.
\begin{enumerate}
    \item $\sum_{j\in[n]} \lambda_j^2 = \frac{d^2-1}{d} \mathds{1}_d$
    \item $\sum_{m\in[n]} \lambda_m \lambda_j \lambda_m = -\frac1{d} \lambda_j$
    \item $\sum_{j,k\in[n]} \lambda_j \lambda_k \lambda_j \lambda_k = - \frac{d^2-1}{d^2} \mathds{1}_d$
\end{enumerate}

\textit{Proof of Identity 1.}
It is known that~$\sum_{j\in[n]}\lambda_j^2$ is a group invariant called the Casimir operator~\cite{GMMCasimir}. Thus,~$\sum_{j\in[n]}\lambda_j^2 = C \; \mathds{1}_d$ for some constant~$C$. We use the trace
condition~$\Tr(\lambda_j \lambda_k) = \delta_{jk}$
\begin{equation}
    \qtrace ( \sum_{j\in [n]} \lambda_j^2) = \sum_{j\in [n]} \qtrace (\lambda_j^2) = n = C d,
\end{equation}
which implies~$C=n/d$, proving
$$\sum_{j\in[n]}\lambda_j^2 = \frac{d^2-1}{d} \mathds{1}_d.$$
\qed

\textit{Proof of Identity 2.}
For this proof, we use some properties of GMMs from Ref.~\cite{GMprops}
(see page 17, Sec. \textbf{4.6} \textit{Gell-Mann Matrices in $n$-dimensions}).
Writing the product~$\lambda_m\lambda_j$ in terms of the commutator
and the anti-commutator,
we get
\begin{equation}
\label{eq:gmproduct}
\begin{split}
2 \lambda_m \lambda_j &= \{ \lambda_m, \lambda_j\} + [\lambda_m, \lambda_j]\\
&= \frac2{d} \delta_{mj} \mathds{1}_d + \sum_c d_{mjc} \lambda_c + \sum_c i f_{mjc} \lambda_c \, ,
\end{split}
\end{equation}
where~$d_{jkl} = \qtrace(\{\lambda_j, \lambda_k\}\lambda_l)$
and~$f_{jkl} = -i \qtrace([\lambda_j, \lambda_k]\lambda_l)$
are the fully-symmetric and fully-antisymmetric
structure constants\footnote{%
Owing to different normalisation conventions,
our $d_{abc}$ and $f_{abc}$ are scaled up
by a factor of~$\sqrt{2}$ compared to Ref.~\cite{GMprops}.}
of~$\mathfrak{su}(d)$~\cite{GMMCasimir,GMprops}.
Repeating the process after right-multiplying Eq.~\eqref{eq:gmproduct}
by~$\lambda_m$,
\begin{equation}
\begin{split}
2 &\lambda_m \lambda_j \lambda_m  = \frac{2{\delta_{mj}}}{d}  \lambda_m + \sum_c \left ( d_{mjc} + i f_{mjc} \right ) \lambda_c \lambda_m \\
&=  \frac{2 {\delta_{mj}}}{d}  \lambda_m + \frac12 \sum_c \left ( d_{mjc} + i f_{mjc} \right ) \Bigg ( \frac{2{\delta_{cm}}}{d}  \mathds{1}_d  \\
&\quad \quad \quad \quad \quad \quad \quad \quad  +\sum_p (d_{cmp} + i f_{cmp}) \lambda_p \Bigg )\\
&=  \frac2{d} \delta_{mj} \lambda_m + \frac1{d} (d_{mjm} + i f_{mjm}) \mathds{1}_d \\
& \quad \quad \quad + \frac12 \sum_{c,p} (d_{mjc} + i f_{mjc}) (d_{cmp} + i f_{cmp}) \lambda_p
\end{split}
\end{equation}

Due to anti-symmetry, $f_{mjm} = 0$,
and 
\begin{gather}
    \begin{aligned}
        (d_{mjc} &+ i f_{mjc}) (d_{cmp} + i f_{cmp})\\
        &=\big [ (d_{mjc} d_{cmp} - f_{mjc} f_{cmp})\\
        &+ i (d_{mjc} f_{cmp} + f_{mjc} d_{cmp} )\big ] .
    \end{aligned}
\end{gather}

Thus,
\begin{gather}
        \sum_m \lambda_m \lambda_j \lambda_m = \frac{1}{d} \lambda_j + \frac{1}{2 d} \underbrace{\sum_m d_{mjm}}_{\Circled{1}}  \mathds{1}_d\\
        + \frac14 \sum_p \Bigg [ \underbrace{\sum_{m,c} d_{mjc} d_{cmp}}_{\Circled{2}} - \underbrace{\sum_{m,c} f_{mjc} f_{cmp}}_{\Circled{3}}\\
        + i \Bigg ( \underbrace{\sum_{m,c} d_{mjc} f_{cmp}}_{\Circled{4}} + 
 \underbrace{\sum_{m,c} f_{mjc} d_{cmp}}_{\Circled{5}} \Bigg )\Bigg ] \lambda_p .
\end{gather}

Below we evaluate
terms~$\Circled{1},\Circled{2},\Circled{3},\Circled{4}$
and~$\Circled{5}$ one by one,
using properties of the GMMs
listed in Ref.~\cite{GMprops}.

\begin{gather}
\begin{aligned}
\Circled{1}: \; \sum_m d_{mjm} &= \sum_m d_{jmm} \\
= \frac{1}{4} \qtrace \left  [ \lambda_j \sum_m \{ \lambda_m, \lambda_m\}  \right ] 
&= \frac12 \qtrace \left [ \lambda_j \frac{d^2-1}{d} \mathds{1}_d \right ]\\
= \frac{d^2-1}{2d} \qtrace ( \lambda_j) &= 0
\end{aligned}
\end{gather}

\begin{equation}
\begin{split}
\Circled{2}: \; \sum_{m,c} d_{mjc} d_{cmp} = \sum_{m,c} d_{jmc} d_{pmc} = 2 \frac{d^2-4}{d} \delta_{jp}
\end{split}
\end{equation}

\begin{equation}
\begin{split}
\Circled{3}: \; \sum_{m,c} f_{mjc} f_{cmp} = \sum_{m,c} f_{jmc} f_{pmc} = 2 d \; \delta_{jp}
\end{split}
\end{equation}

The Jacobi identity~\cite{GMprops} reads
$$\sum_k d_{abk} f_{kcl} + d_{bck} f_{kal} + d_{cak} f_{kbl} = 0.$$
If we set $a=c$ and then sum over~$a$, we get
$$2 \sum_{a,k} d_{bak} f_{lak} = \sum_k \left ( \sum_a d_{aak} \right ) f_{blk}.$$
Using this to simplify~$\Circled{4}$,
we get
\begin{equation}
\begin{split}
\Circled{4}: \; \sum_{m,c} d_{mjc} f_{cmp} &= - \sum_{m,c} d_{jmc} f_{pmc} \\
&= -\frac12 \sum_{m,c} d_{mmc} f_{jpc}
\end{split}
\end{equation}
and, similarly,
for~$\Circled{5}$ we get
\begin{equation}
\begin{split}
\Circled{5}: \; &\sum_{m,c} f_{mjc} d_{cmp}
= - \sum_{m,c} d_{pmc} f_{jmc} \\
= &-\frac12 \sum_{m,c} d_{mmc} f_{pjc}
= \frac12 \sum_{m,c} d_{mmc} f_{jpc} \; ,
\end{split}
\end{equation}
so that
\begin{equation}
  \Circled{4}+\Circled{5}:\;  \sum_{m,c} (d_{mjc} f_{cmp}
  + f_{mjc} d_{cmp} ) = 0 \, .
\end{equation}

Combining the expressions for~$\Circled{1},\Circled{2}$ and~$\Circled{3}$,
we get
\begin{equation}
\begin{split}
\sum_m \lambda_m \lambda_j \lambda_m &= \frac1{d} \lambda_j + \frac12 \sum_p \left ( \frac{d^2-4}{d} -  d \right ) \delta_{jp} \lambda_p\\
&= \frac1{d} \lambda_j - \frac{2}{d} \lambda_j = -\frac{1}{d} \lambda_j \, ,
\end{split}
\end{equation}
thus proving Identity~2.
\qed

\begin{corollary}
\label{cor:corrolary}
    By linearity, any~$d \times d$ traceless Hermitian matrix $A$ satisfies
    \begin{equation}
        \sum_m \lambda_m A \lambda_m = -\frac{1}{d} A.
    \end{equation}
\end{corollary}

\begin{corollary}
\label{cor:corollary3}
For any~$j,k\in[n]$,
\begin{equation}
    \sum_m \lambda_m \lambda_j \lambda_k \lambda_m = \delta_{jk} \mathds{1}_d - \frac{1}{d} \lambda_j \lambda_k .
\end{equation}
\end{corollary}
\begin{proof}
    To see this, start with assuming~$j\neq k$.
    From Eq.~\eqref{eq:gmproduct},
    this implies $\lambda_j \lambda_k$ is traceless Hermitian,
    and thus from Corollary~\ref{cor:corrolary},
    $$ \sum_m \lambda_m \lambda_j \lambda_k \lambda_m =  - \frac{1}{d} \lambda_j \lambda_k .$$
    Similarly, for~$j=k$,
    $\lambda_j \lambda_k - \frac1{d} \mathds{1}_d$ is a traceless, Hermitian matrix (see Eq.~\eqref{eq:gmproduct}).
    Thus, from Corollary~\ref{cor:corrolary},
    \begin{align}
            \sum_m \lambda_m \left (\lambda_j^2 - \frac{1}{d} \mathds{1}_d \right ) \lambda_m &= -\frac{1}{d} \left (\lambda_j^2 - \frac{1}{d} \mathds{1}_d \right ), \nonumber \\
            \intertext{so that}
            \sum_m \lambda_m \lambda_j^2 \lambda_m 
            &= \mathds{1}_d - \frac{1}{d} \lambda_j^2 \, .
    \label{eq:identity5partb}
    \end{align}
This concludes the proof of Corollary~\ref{cor:corollary3}.
\end{proof}

\textit{Proof of Identity 3.}
Using Identity 2, we have 
\begin{equation}
\sum_j \lambda_j \lambda_k \lambda_j \lambda_k = \left ( \sum_j \lambda_j \lambda_k \lambda_j \right ) \lambda_k= -\frac{1}{d} \lambda_k^2 \, .
\end{equation}
Summing over~$k$ and using Identity~1,
we find
\begin{equation}
    \sum_{j,k} \lambda_j \lambda_k \lambda_j \lambda_k  = -\frac{1}{d} \sum_k \lambda_k^2 = - \frac{d^2-1}{d^2} \mathds{1}_d \, ,
\end{equation}
which proves Identity~3.
\qed

\section{Deferred Proofs: Proof of Lemma~\texorpdfstring{\ref{lemma:Holevo}}{1}}
\label{subsec:ProofHCRB}

\begin{proof}[Proof of Lemma~\ref{lemma:Holevo}]
    The proof is segmented into three parts.
    (i) First we establish that the~$X_j$ are completely
    and uniquely determined by
    the local unbiasedness conditions
    to be~$X_j=\lambda_j$.
    This can be traced back to the trace orthonormality
    $\qtrace [ \lambda_j \lambda_k ] = \delta_{jk}$ of GMMs.
    (ii) We establish a lower bound on~$\CH$.
    (iii) We show this lower bound is achieved
    by valid choices of arguments~$\mathbb{S}_\theta$ and~$\mathbb{L}$,
    implying~$\CH$ is equal to the lower bound.

    \textit{Part (i):}
    The local unbiased conditions
    \begin{equation}
        \qtrace [ \rho_\theta X_j] = \theta_j \quad  \text{and} \quad  \qtrace [ \partial_j \rho_\theta X_k] = \delta_{jk}
    \end{equation}
    for~$j,k \in [d^2-1]$ at~$\theta=0$ become
    \begin{equation}
        \qtrace[X_j] = 0 \quad \text{and} \quad \qtrace [ \lambda_j X_k] = \delta_{jk} \, .
    \end{equation}
    It is simple to verify that
    the only solution to this
    is~$X_j=\lambda_j$.
    By virtue of being traceless and Hermitian,
    each~$X_j$ must be a linear combination of~$\lambda_k$s,
    i.e.,
    \begin{equation}
        X_j = \sum_k c_{jk} \lambda_k \quad c_{jk} \in \mathbb{R} \, ,
    \end{equation}
    which then means~$c_{jk}$ must
    satisfy~$\qtrace[ \lambda_j X_k] =  c_{kj} = \delta_{jk}$,
    implying~$X_j = \lambda_j$.

    \textit{Part (ii):}
    Tracing over the parameter indices~$(\ctrace)$ in
    \begin{equation}
        \Tr[ \mathbb{S}_\theta \mathbb{L} ] \succcurlyeq \Tr [ \mathbb{S}_\theta \mathbb{X} \mathbb{X}^\top ] 
    \end{equation}
    yields
    \begin{equation}
                \mathbb{T}\mathrm{r}[ \mathbb{S}_\theta \mathbb{L} ] \geq \mathbb{T}\mathrm{r} [ \mathbb{S}_\theta \mathbb{X} \mathbb{X}^\top ] .
                \label{eq:lowerbound}
    \end{equation}

    \textit{Part (iii):}
    The choice~$\mathbb{L} = \mathbb{X} \mathbb{X}^\top$
    leads to
    \begin{equation}
        \left ( \qtrace[ \mathbb{S}_\theta \mathbb{L} ]\right )_{j,k} = \qtrace[ \mathds{1}/d \;  \lambda_j \lambda_k] = \delta_{jk}/d, 
    \end{equation}
    which is real, symmetric and has trace (over parameter indices)
    \begin{equation}
    \bigtrace [ \mathbb{S}_\theta \mathbb{L} ] = \bigtrace [ \mathbb{S}_\theta \mathbb{X} \mathbb{X}^\top ] = \frac{d^2-1}{d}.  
    \end{equation}
    Finally, from \textit{part (ii)} we know that a lower~$\bigtrace [ \mathbb{S}_\theta \mathbb{L} ]$ is not possible, thus proving Eq.~\eqref{eq:Holevo}.
\end{proof}

    \begin{remark}
    \label{re:HCRBlowerbound}
        When there exists a unique
        set of LUB operators~$\mathbb{X}$,
        the inequality in Eq.~\eqref{eq:lowerbound} holds regardless
        of whether~$\mathbb{L}=\mathbb{X}\mathbb{X}^\top$ is a valid choice according to the HCRB constraints in Eq.~\eqref{eq:Holevodefn2}.
        We utilise this fact to solve the full-parameter
        linear GMM model for arbitrary states
        in Sec.~\ref{subsec:ArbitStates}
        and for arbitrary weight matrices in Appendix~\ref{sec:arbitweight}.
    \end{remark}

\section{SLD \& RLD CRBs}
\label{sec:appSLDRLD}

The two simplest quantum Cram\'er-Rao bounds,
the SLD and the RLD,
generalise the logarithmic derivative
of a parameterised probability distribution,
\begin{equation}
\partial_\theta p_\theta(x) = p_\theta (x) \partial_\theta \left [ \log(p_\theta(x)) \right ] \, ,
\end{equation}
to linear operators
acting on the density matrix~$\rho_\theta$.
The SLD version produces
Hermitian operators~$\{L_j^\mathrm{(SLD)}\}$
and the RLD version produces operators~$\{L_j^\mathrm{(RLD)}\}$
defined implicitly via
\begin{align}
    \label{eq:SLDops}
   2 \partial_j \rho_\theta &\eqqcolon  L_j^\mathrm{(SLD)} \rho_\theta + \rho_\theta L_j^\mathrm{(SLD)} \, , \\
     \partial_j \rho_\theta &\eqqcolon  \rho_\theta L_j^\mathrm{(RLD)} \, .
    \label{eq:RLDops}
\end{align}
Once Eqs.~\eqref{eq:SLDops} and~\eqref{eq:RLDops}
are solved for~$\{L_j^\mathrm{(SLD)}\}$
and~$\{L_j^\mathrm{(RLD)}\}$,
the corresponding QFIs can
be computed using
\begin{align}
\label{eq:SLDQFIdefns}
    \left [ J^\mathrm{(SLD)} \right]_{jk} &\coloneqq \Re \left [  \qtrace \left [ \rho_\theta L_j^\mathrm{(SLD)} L_k^\mathrm{(SLD)} \right ] \right ] \, , \\
    \left [ J^\mathrm{(RLD)} \right]_{jk} &\coloneqq \qtrace \left [ \rho_\theta L_k^\mathrm{(RLD)} {L_j^\mathrm{(RLD)}}^\dagger \right ]  \, . \label{eq:RLDQFIdefns}
\end{align}
Notably,~$ J^\mathrm{(SLD)}$ is real and symmetric
and~$J^\mathrm{(RLD)}$ is complex and Hermitian.
Finally,
the traced versions of the SLD and RLD QFI
matrix inequalities~$V_\theta \succcurlyeq {J^\mathrm{(SLD)}}^{-1}$
and~$V_\theta \succcurlyeq {J^\mathrm{(RLD)}}^{-1}$
yield the scalar
SLD and RLD CRBs
\begin{equation}
\ctrace(V_\theta) \geq \SLD \quad \text{\&} \quad \ctrace(V_\theta) \geq \RLD \, ,
\end{equation}
with
\begin{align}
    \label{eq:SLDbounds}
    \mathrm{C}_\mathrm{SLD} &= \ctrace \left [ {J^\mathrm{(SLD)}}^{-1} \right ] \, , \\
    \mathrm{C}_\mathrm{RLD} &= \ctrace \left [ \Re \left [ J^\mathrm{(RLD)}\right ]^{-1} \right ] \label{eq:RLDbounds}\\
    &\qquad \qquad \qquad  + \;  \left \Vert \Im \left [ J^\mathrm{(RLD)}\right ]^{-1} \right \Vert_1 \, , \nonumber
\end{align}
where~$\Vert X \Vert_1 \coloneqq \Tr (\sqrt{X^\dagger X})$
denotes the trace norm.
The SLD and RLD CRBs are not attainable in general,
especially in multi-parameter contexts.
For more details on the SLD and RLD CRB,
see Ref.~\cite{ABGG20}.

For the model
in Sec.~\ref{subsec:OdRatioMaxMixState},
$\rho_\theta^*=\mathds{1}_d/d$ and~$\partial_j\rho_\theta = \lambda_j$,
so Eqs.~\eqref{eq:SLDops} \& \eqref{eq:RLDops}
become
\begin{equation}
    \begin{split}
        2 \lambda_j = 2/d \;  L_j^\mathrm{SLD} \, , \\
        \lambda_j = 1/d \;  L_j^\mathrm{RLD} \, 
    \end{split}
\end{equation}
implying~$L_j^\mathrm{SLD} = L_j^\mathrm{RLD} = d \lambda_j$.
A direct computation of
Eqs.~\eqref{eq:SLDQFIdefns},~\eqref{eq:RLDQFIdefns},~\eqref{eq:SLDbounds}
\&~\eqref{eq:RLDbounds}
then yields the QFIs
\begin{equation}
    J^{(\mathrm{SLD})} = J^{(\mathrm{RLD})} = \begin{bmatrix}
        d & 0 & \dots & 0 \\
        0 & d & \dots & \vdots \\
        \vdots & \vdots & \ddots & 0\\
        0 & \hdots & 0 & d
    \end{bmatrix}_{n\times n} \, ,
\end{equation}
and the scalar CRBs
\begin{equation}
    \begin{split}
        \SLD = \RLD = \frac{d^2-1}{d} \, ,
    \end{split}
\end{equation}
as in main text Eqs.~\eqref{eq:JSLDRLDQFI} and~\eqref{eq:SLDRLDmodelmaxmix}.

\section{Deferred Proofs: Proof of Lemma~\texorpdfstring{\ref{theorem:NHCRB}}{2}}
\label{sec:ProofLemmaNHCRBmaxmix}

\subsection{SDP Formulation and Dual Problem}
\label{subsec:ProofNHCRB}

The SDP formulation
of the NHCRB~\cite{Lorcan21} is
\begin{gather}
\label{eq:HNCRBdef2SDP}
\begin{aligned}
    \CNH := \min_{\mathbb{Y}} \Big [ \bigtrace [ \mathbb{F}_0 \mathbb{Y} ]
 & \vert \bigtrace [ \mathbb{F}_k \mathbb{Y} ] = c_k \, , \\
      \mathbb{F}_0 = \begin{pmatrix} \mathbb{S}_\theta & 0 \\ 0 & 0 \end{pmatrix} \, , 
      \, &\mathbb{Y}  \succcurlyeq 0 \Big ] \; ,
\end{aligned}
\end{gather}
where $\mathbb{F}_k$ are constant matrices and
$c_k$ are constants,
as defined in Ref.~\cite{Lorcan21} (supplementary note 4).
The dual problem to the SDP in Eq.~\eqref{eq:HNCRBdef2SDP} reads
\begin{gather}
\label{eq:HNCRBdualdef}
\begin{aligned}
    \tilde{\mathrm{C}}_{\mathrm{NHCRB}} := \max_{y} 
 \Big [ \sum_k y_k c_k & \vert \sum_k y_k \mathbb{F}_k \preccurlyeq \mathbb{F}_0   \Big ] \; .
\end{aligned}
\end{gather}

In the following two lemmas,
we present solutions
to the primal and dual problems.

\begin{lemma}
\label{lemma:primalsolutionNHCRB}
    The optimal~$\mathbb{L}$ for the primal problem
    in Eq.~\eqref{eq:HNCRBdefn}
    is
    \begin{equation}
        \mathbb{L}^*_{jk} = \nicefrac{d+1}{d+2} \left ( \{ \lambda_j, \lambda_k\} + \delta_{jk} \mathds{1}_d \right ) 
    \end{equation}
where~$j,k \in [n]$ and~$\{\lambda_j, \lambda_k\}$ is the anti-commutator.
\end{lemma}

\begin{lemma}
\label{lemma:dualsolutionNHCRB}
    The optimal~$y$ for the dual problem in Eq.~\eqref{eq:HNCRBdualdef},
    $y^*$, is
    such that
    \begin{equation}
    \setlength\arraycolsep{2.5pt}
        \sum_{k} y^*_k \mathbb{F}_k = {\begin{bmatrix}
            {0} & {\mathbb{G}^{(1)}_{12}} & {\dots} & {\mathbb{G}^{(1)}_{1n}} & {\mathbb{G}_1^{(2)}} \\
            {\mathbb{G}^{(1)}_{21}} & {0} & {\dots} & {\mathbb{G}^{(1)}_{2n}} & {\vdots} \\
            {\vdots} & {\vdots} & {\ddots} & {\mathbb{G}^{(1)}_{(n-1)n}} & \; \\
            {\mathbb{G}^{(1)}_{n1}} & {\dots} & {\mathbb{G}^{(1)}_{n(n-1)}} & {0} & {\mathbb{G}_n^{(2)}} \\
            {\mathbb{G}_1^{(2)}} & {\dots} & \; & {\mathbb{G}_n^{(2)}} & {\mathbb{G}^{(3)}}
        \end{bmatrix}}
    \end{equation}
    with
    \begin{gather}
        \begin{aligned}
            \mathbb{G}_{jk}^{(1)} &= -\nicefrac{1}{d} \; \left [ \lambda_j, \lambda_k \right ] \\
            \mathbb{G}_j^{(2)} &= \nicefrac{d+1}{d} \; \; \lambda_j \\
            \mathbb{G}^{(3)} &= - \nicefrac{(d^2-1)(d+1)}{d^2} \; \; \mathds{1}_d ,
        \end{aligned}
    \end{gather}
    where~$j,k \in [n]$ and~$ \left [ \lambda_j, \lambda_k \right ]$ is the commutator.
\end{lemma}

The proof of Lemmas~\ref{lemma:primalsolutionNHCRB}
and~\ref{lemma:dualsolutionNHCRB}
is broken up into the following three subsections.
In Subsec.~\ref{subsec:appProofLemmaPrimal},
we prove the feasibility of~$\mathbb{L}^*$
from Lemma~\ref{lemma:primalsolutionNHCRB}.
In Subsec.~\ref{subsec:appProofLemmaDual}
we prove the feasibility of~$y^*$
from Lemma~\ref{lemma:dualsolutionNHCRB}.
Finally in Subsec.~\ref{subsec:ProofOptimality}
we prove that the primal-objective value
from~$\mathbb{L}^*$ equals the
dual objective value from~$y^*$,
thus establishing their optimality
and proving Lemma~\ref{theorem:NHCRB}
from the main text.

\subsection{Feasibility of \texorpdfstring{$\mathbb{L}^*$}{L*} for the Primal Problem}
\label{subsec:appProofLemmaPrimal}
The $\mathbb{L}^*_{jk}$ from Lemma~\ref{lemma:primalsolutionNHCRB} is easily seen to be symmetric in~$j$ and $k$, meaning~$\mathbb{L}^*_{jk} = \mathbb{L}^*_{kj}$.
$\mathbb{L}^*_{jk}$ is also seen to be Hermitian.
To show $\mathbb{L}^*$ is feasible,
it only remains to show $\mathbb{L}^*-\mathbb{X}{\mathbb{X}}^\top\succcurlyeq 0$.
Writing $\mathbb{L}^*-\mathbb{X}{\mathbb{X}}^\top$
as a block matrix,
\begin{gather}
\begin{aligned}
    (\mathbb{L}^*-\mathbb{X}{\mathbb{X}}^\top)_{jk} &= \frac{d+1}{d+2} \left ( \{ \lambda_j, \lambda_k \} + \delta_{jk} \mathds{1}_d \right ) - \lambda_j \lambda_k\\
    &= \frac{d+1}{d+2} \left ( \delta_{jk} \mathds{1}_d + {\mathbb{N}_1}_{jk} - {\mathbb{N}_2}_{jk} \right )
\end{aligned}
\end{gather}
we see that we need to prove~$\mathbb{N}\coloneqq \mathds{1}_{nd}+\mathbb{N}_1 - \mathbb{N}_2\succcurlyeq 0$,
where we have defined block matrices
\begin{equation*}
    {\left (\mathbb{N}_1\right )}_{jk} \coloneqq \lambda_k \lambda_j \quad \& \quad {\left (\mathbb{N}_2\right )}_{jk} \coloneqq \frac{\lambda_j \lambda_k}{d+1}.
\end{equation*}
We first evaluate~$(\mathbb{N}_1 - \mathbb{N}_2)^2=\mathbb{N}_1^2 + \mathbb{N}_2^2 - \mathbb{N}_1\mathbb{N}_2 - \mathbb{N}_2 \mathbb{N}_1$ to find
\begin{gather}
    \begin{aligned}
        (\mathbb{N}_1^2)_{jk} &= \delta_{jk} \mathds{1}_d - \nicefrac{1}{d} \lambda_j \lambda_k\\
        (\mathbb{N}_2^2)_{jk} &= \frac{d-1}{d(d+1)} \lambda_j \lambda_k \\
        (\mathbb{N}_1 \mathbb{N}_2)_{jk} &= -\frac{1}{d(d+1)} \lambda_j \lambda_k\\
        (\mathbb{N}_2 \mathbb{N}_1)_{jk} &= -\frac{1}{d(d+1)} \lambda_j \lambda_k,
    \end{aligned}
\end{gather}
where we have used the identities proven in Appendix~\ref{sec:appIdentities}.
Combining these results we arrive at $\big((\mathbb{N}_1 - \mathbb{N}_2)^2\big)_{jk} = \delta_{jk} \mathds{1}_d$ which means~$(\mathbb{N}_1 - \mathbb{N}_2)^2 = \mathds{1}_{nd}$. From this, and using that $\mathbb{N}_1-\mathbb{N}_2$ is Hermitian,
we can conclude that the eigenvalues of~$\mathbb{N}_1-\mathbb{N}_2$ are~$\pm 1$.
Hence the eigenvalues of~$\mathbb{N}=\mathds{1}_{nd} + \mathbb{N}_1 - \mathbb{N}_2$ are~0 and~2.
This proves that~$\mathbb{N}$ is a positive semi-definite operator,
and that~$\mathbb{L}^* - \mathbb{X} \mathbb{X}^\top\geq 0$.
\qed

\subsection{Feasibility of \texorpdfstring{$y^*$}{y*} for the Dual Problem}
\label{subsec:appProofLemmaDual}
We need to show that the matrix~$\sum_{k} y^*_k \mathbb{F}_k$,
which explicitly is
\begin{widetext}
\begin{equation}
    \renewcommand{\arraystretch}{0.5}
        \begin{bmatrix}
            0 & \frac{1}{d} [\lambda_2,\lambda_1] & \dots & \frac{1}{d} {[\lambda_n,\lambda_1]} & \frac{d+1}{d} \lambda_1 \\
            \frac{1}{d} {[\lambda_1,\lambda_2]} & 0 & \dots & \frac{1}{d} {[\lambda_n,\lambda_2]} & \vdots \\
            \vdots & \vdots & \ddots & \frac{1}{d} {[\lambda_n,\lambda_{n-1}]} & \; \\
            \frac{1}{d} {[\lambda_1,\lambda_n]} & \dots & \frac{1}{d} {[\lambda_{n-1},\lambda_n]} & 0 & \frac{d+1}{d} \lambda_n \\
            \frac{d+1}{d} \lambda_1 & \dots & \; & \frac{d+1}{d} \lambda_n & -\frac{n(d+1)}{d^2} \mathds{1}_d
        \end{bmatrix}
    \renewcommand{\arraystretch}{1},
\end{equation}
\end{widetext}
satisfies~$\mathbb{F}_0-\sum_{k} y^*_k \mathbb{F}_k\succcurlyeq 0$.
Note that this~$\sum_{k} y^*_k \mathbb{F}_k$
corresponds to~$y^*_j$ values
\begin{gather}
    \label{eq:optyvaluesdual}
    y^{(1)}_j=0, \quad y^{(2)}_{jk} = \frac{d+1}{d} \delta_{jk}, \quad y^{(3)}_{jk}=0\\
    y^{(4)}_{jkl} = -f_{jkl}/d, \quad y^{(5)}_{j} = - \frac{n(d+1)}{d \sqrt{d}} \delta_{j1}
\end{gather}
whereas the corresponding~$c_j$ values are
\begin{gather}
    \label{eq:cvaluesdual}
    c^{(1)}_j=0, \quad c^{(2)}_{jk} = 2 \delta_{jk}, \quad c^{(3)}_{jk}=0\\
    c^{(4)}_{jkl} = 0, \quad c^{(5)}_{j} = \sqrt{d} \delta_{j1}    
\end{gather}
so that the dual objective value is
\begin{equation}
\begin{split}
    \sum_j 2 y^{(2)}_{jj} + \sqrt{d} y^{(5)}_j &= \frac{2n (d+1)}{d}-\frac{n(d+1)}{d} \\
    &= \frac{n(d+1)}{d} .
\end{split}
\end{equation}

To show this~$y^*$ is feasible,
note that
proving~$\mathbb{F}_0 - \sum_{k} y^*_k \mathbb{F}_k\succcurlyeq 0$
is equivalent to showing
\begin{widetext}
\begin{equation*}
    \renewcommand{\arraystretch}{0.5}
        \begin{bmatrix}
            \mathds{1}_d &  [\lambda_1,\lambda_2] & \dots &  {[\lambda_1,\lambda_n]} & -(d+1) \lambda_1 \\
             {[\lambda_2,\lambda_1]} & \mathds{1}_d & \dots & {[\lambda_2,\lambda_n]} & \vdots \\
            \vdots & \vdots & \ddots &  {[\lambda_{n-1},\lambda_{n}]} & \; \\
            {[\lambda_n,\lambda_1]} & \dots &  {[\lambda_n,\lambda_{n-1}]} & \mathds{1}_d & -(d+1)\lambda_n \\
            -(d+1)\lambda_1 & \dots & \; & -(d+1) \lambda_n & \frac{n(d+1)}{d} \mathds{1}_d
        \end{bmatrix}
    \renewcommand{\arraystretch}{1}
\end{equation*}
\end{widetext}
is positive semi-definite.
Using Schur's complement lemma,
this can be simplified to showing
\begin{widetext}
\begin{equation}
%\begin{gather}
    \renewcommand{\arraystretch}{0.5}
        \begin{bmatrix}
            \mathds{1}_d &  [\lambda_1,\lambda_2] & \dots &  {[\lambda_1,\lambda_n]}\\
             {[\lambda_2,\lambda_1]} & \mathds{1}_d & \dots & {[\lambda_2,\lambda_n]}\\
            \vdots & \vdots & \ddots &  {[\lambda_{n-1},\lambda_{n}]} \\
            {[\lambda_n,\lambda_1]} & \dots &  {[\lambda_n,\lambda_{n-1}]} & \mathds{1}_d
        \end{bmatrix}
    \renewcommand{\arraystretch}{1} 
    - \frac{d}{d-1} \begin{bmatrix}
        \lambda_1^2 & \lambda_1 \lambda_2 & \dots & \lambda_1 \lambda_n\\
        \lambda_2 \lambda_1 & \ddots & \dots & \vdots \\
        \vdots & \vdots & \; & \vdots \\
        \lambda_n \lambda_1 & \lambda_n \lambda_2 & \dots & \lambda_n^2
    \end{bmatrix} \succcurlyeq 0
    \label{eq:simplformdualfeas}
%\end{gather}
\end{equation}
\end{widetext}
We rewrite the left hand side
of Eq.~\eqref{eq:simplformdualfeas}
in the block-matrix representation as
\begin{gather*}
    \delta_{jk} \mathds{1}_d + [ \lambda_j, \lambda_k] - \frac{d}{d-1} \lambda_j \lambda_k\\
    = \delta_{jk} \mathds{1}_d - \frac{1}{d-1} \lambda_j \lambda_k - \lambda_k \lambda_j\\
    = (\mathds{1}_{nd})_{jk} - \left ( \left (\mathbb{M}_1\right)_{jk}+\left (\mathbb{M}_2\right)_{jk} \right ) \, ,
\end{gather*}
where we have defined
\begin{equation}
    (\mathbb{M}_1)_{jk} \coloneqq \frac{1}{d-1} \lambda_j \lambda_k \quad \& \quad (\mathbb{M}_2)_{jk} \coloneqq \lambda_k \lambda_j .
\end{equation}
Thus,
we finally need to prove the following theorem to establish the feasibility of~$y^*$.

\begin{lemma}
\label{lemma:dualfeasibility}
The operator~$\mathbb{M}\coloneqq \mathds{1}_{nd} - (\mathbb{M}_1+\mathbb{M}_2)$ is positive semi-definite.
\end{lemma}

Before we can prove Lemma~\ref{lemma:dualfeasibility},
we first need to prove the following two lemmas.

\begin{lemma}
    \label{lemma:dualfeasibilitylemma1}
    $\mathbb{M}_1$ and~$\mathbb{M}_2$ commute, i.e.,
    $\mathbb{M}_1 \mathbb{M}_2 = \mathbb{M}_2 \mathbb{M}_1$.
\end{lemma}
\begin{proof}
Using Identity~2 from Appendix~\ref{sec:appIdentities}, we have 
\begin{equation}
\begin{split}
(\mathbb{M}_1 \mathbb{M}_2)_{jk} &= \sum_l (\mathbb{M}_1)_{jl}  (\mathbb{M}_2)_{lk}\\
&= \frac{1}{d-1} \sum_l \lambda_j \lambda_l \lambda_k \lambda_l \\
& = \frac{1}{d-1} \lambda_j \left (\frac{-1}{d} \lambda_k\right ) \\
&= -\frac{1}{d(d-1)} \lambda_j \lambda_k \, ,
\end{split}
\end{equation}
whereas
\begin{equation}
\begin{split}
(\mathbb{M}_2 \mathbb{M}_1)_{jk} &= \sum_l (\mathbb{M}_2)_{jl}  (\mathbb{M}_1)_{lk}\\
&=\frac{1}{d-1} \sum_l \lambda_l \lambda_j \lambda_l \lambda_k \\
&= \frac{1}{d-1} \left ( \frac{-1}{d} \lambda_j \right ) \lambda_k \\
&= -\frac{1}{d(d-1)} \lambda_j \lambda_k .
\end{split}
\end{equation}
Hence $\mathbb{M}_1 \mathbb{M}_2 = \mathbb{M}_2 \mathbb{M}_1$,
which also implies that $\mathbb{M}_1$ and $\mathbb{M}_2$ share some eigenvectors.
\end{proof}

\begin{lemma}
    \label{lemma:dualfeasibilitylemma2}
    $\mathbb{M}_1+\mathbb{M}_2$ satisfies
    $(\mathbb{M}_1 + \mathbb{M}_2)^2 = \mathds{1}_{nd}$
    or, equivalently, 
    $$ \left [(\mathbb{M}_1 + \mathbb{M}_2)^2 \right ]_{jk} = \delta_{jk} \mathds{1}_d$$
where~$j,k\in[n]$.
\end{lemma}
\begin{proof}
\begin{equation*}
\begin{split}
&\left [(\mathbb{M}_1 + \mathbb{M}_2)^2 \right ]_{jk}\\
&= \sum_l \left (\mathbb{M}_1 + \mathbb{M}_2 \right)_{jl} \left (\mathbb{M}_1 + \mathbb{M}_2 \right)_{lk} \\
&= \sum_l \left ( \frac{1}{d-1} \lambda_j \lambda_l + \lambda_l \lambda_j\right ) \left ( \frac{1}{d-1} \lambda_l \lambda_k + \lambda_k \lambda_l\right ) \\
&= \frac{1}{(d-1)^2} \lambda_j \left ( \sum_l \lambda_l^2 \right ) \lambda_k +  \sum_l \lambda_l \lambda_j \lambda_k \lambda_l\\
&+ \frac{1}{d-1} \left [ \lambda_j \left (\sum_l \lambda_l \lambda_k \lambda_l \right ) + \left ( \sum_l \lambda_l \lambda_j \lambda_l \right ) \lambda_k \right ] \\
\intertext{so using Corollary~\ref{cor:corollary3} and Identities~1 and~2,}
&=  \delta_{jk} \mathds{1}_d \\
& \; + {\left ( {\frac{d^2-1}{d(d-1)^2}} - {\frac{1}{d}} - {\frac{1}{d(d-1)}} - {\frac{1}{d(d-1)}} \right )} {\lambda_j} {\lambda_k}\\
&=  \delta_{jk} \mathds{1}_d.
\end{split}
\end{equation*}
\end{proof}

Now we can prove Lemma~\ref{lemma:dualfeasibility} as follows.
\begin{proof}[Proof of Lemma~\ref{lemma:dualfeasibility}]
From Lemma~\ref{lemma:dualfeasibilitylemma2},
the eigenvalues of~$(\mathbb{M}_1 + \mathbb{M}_2)^2$ must all be~$1$. As $(\mathbb{M}_1 + \mathbb{M}_2)$ is Hermitian, its eigenvalues must be~$\pm 1$.
It follows that the eigenvalues of~$\mathbb{M}=\mathds{1}_{nd} - (\mathbb{M}_1 + \mathbb{M}_2)$ must be either~$2$ or~$0$.
Hence~$\mathbb{M}$, being a Hermitian matrix with non-negative eigenvalues, must be positive semi-definite.
\end{proof}

\subsection{Optimality of Solutions \&
Proof of Lemma~\texorpdfstring{\ref{theorem:NHCRB}}{2}}
\label{subsec:ProofOptimality}

\begin{proof}[Proof of Lemma~\ref{theorem:NHCRB}]
Note that, by direct calculation,
\begin{equation}
\bigtrace(\mathbb{S}_\theta \mathbb{L}^*) = \sum_k y^*_k c_k = \frac{(d^2-1)(d+1)}{d} \; . 
\end{equation}
In other words,~$\mathbb{L}^*$ is primal-feasible
and~$y^*$ is dual-feasible
and the primal value equals the dual value.
This lets us conclude that~$\nicefrac{(d^2-1)(d+1)}{d}$
is the true optimal value of the primal and dual problems,
and that~$\mathbb{L}^*$ and~$y^*$ are optimal solutions
to the primal and dual problems, respectively.
As a result, we have
    \begin{equation}
        \CNH = \frac{(d^2-1)(d+1)}{d}.
    \end{equation}
\end{proof}

\subsection{Attainability of NHCRB via SIC POVMs}
\label{subsec:ProofCFI}

\begin{lemma}
\label{th:CFI}
The CFI matrix
for estimating all GMMs from the maximally-mixed state~$\rho_\mathrm{m}$
by measuring the SIC POVM in~$d$ dimensions is
\begin{equation}
    J = \begin{bmatrix}
        \frac{d}{d+1} & 0 &\cdots & 0 \\
        0 & \frac{d}{d+1} & \cdots & 0 \\
        \vdots & \vdots & \ddots & \vdots \\
        0 & 0  & \dots & \frac{d}{d+1} 
    \end{bmatrix}_{n\times n} .
    \label{eq:theorem1}
\end{equation}
\end{lemma}

\begin{proof}[Proof of Lemma~\ref{th:CFI}]
    In the multi-parameter case,
    the CFI matrix~$J_{jk}$ ($j,k\in [n]$)
    is given by
    \begin{equation}
    \label{eq:QFImatrix112}
    J_{jk} \left [\{\Pi_l\}\right] = \sum_{l=1}^{d^2} \frac{ \qtrace  \left [ \partial_j \rho_\theta \Pi_l \right ] \qtrace  \left [ \partial_k \rho_\theta \Pi_l \right ]}{\qtrace \left [ \rho_\theta \Pi_l \right ] } \, ,
    \end{equation}

    From \cite{R13} we have that for any (rank-one)
    SIC POVM $\{\Pi_l\}_{l=1}^{d^2}$,
    \begin{equation}
        \sum_{l=1}^{d^2} \qtrace [\rho \Pi_l  ]^2 = \frac{\qtrace[\rho^2] +1}{d(d+1)}
        \label{eq:indexofcoincidence}
    \end{equation}
    for arbitrary density matrix~$\rho$.
    For the diagonal elements in Eq.~\eqref{eq:QFImatrix112},
    substituting~$\rho = \mathds{1}/d +\theta_j \lambda_j$ into Eq.~\eqref{eq:indexofcoincidence}
    and using~$\qtrace(\rho^2) = 1/d+\theta_j^2$ gives
    \begin{equation}
     J_{jj}=d^2 \sum_{l=1}^{d^2} \qtrace \left [ \lambda_j \Pi_l \right ]^2 = \frac{d}{d+1}
     \label{eq:diagonal} ,
    \end{equation}
    whereas for the off-diagonal elements,
    substituting~$\rho = \mathds{1}/d + \theta_j \lambda_j + \theta_k \lambda_k$ into Eq.~\eqref{eq:indexofcoincidence} and using~$\qtrace(\rho^2) = 1/d+\theta_j^2+\theta_k^2$ gives
    \begin{equation}
        J_{jk} = d^2 \sum_{l=1}^{d^2} \qtrace[ \lambda_j \Pi_l] \qtrace[ \lambda_k \Pi_l] = 0 \quad (j\neq k),
    \end{equation}
    thus proving Eq.~\eqref{eq:theorem1}.
\end{proof}

\section{Estimating a Subset of GMMs}
\label{sec:AppEstimatingFewGGMMs}

Consider estimating a subset~$\{\lambda_j\}_{j\in K}$
of GMMs from the maximally-mixed state~$\rho_\mathrm{m}$.
Here~$K$ denotes a subset of~$n$ indices from~$1$ to~$\nmax$
($K\subseteq [\nmax], \vert K \vert = n$).
We denote the maximum NHCRB-to-HCRB
ratio here as~$\rat_n[\{\rho_\mathrm{m}\}]$,
referring to the model~$\{\rho_\mathrm{m}\}$ of
estimating~$n$ GMM coefficients from~$\rho_\mathrm{m}$.
Now, the corresponding unbiased operators
can be written as
\begin{equation}
    X_j = \lambda_j + \sum_{m\in [\nmax]\setminus K} c_{jm} \lambda_m, \quad j \in K, \, c_{jm}\in \mathbb{R},
\end{equation}
which follows from
the unbiasedness conditions in Eq.~\eqref{eq:unbiasednessnew}.
Specifically,
$\Tr(\partial_j \rho_\theta X_k) = \Tr(\lambda_j X_k) =\delta_{jk}$
forces each~$X_j$ to contain a unit
contribution from~$\lambda_j$
due to the orthonormality of the GMMs
and~$\Tr(\rho_\theta X_j) = \theta_j$
implies the only other GMMs contributing
to~$X_j$ must be the ones not being estimated.

Notice that at the block-matrix level,
this can be rewritten as
\begin{equation}
    \mathbb{X} = \bigg [\begin{array}{c|c}
    \mathds{1}_n & \mathbb{C}^{(1)}_{n  \times (\nmax-n)}
    \end{array} \bigg ]_{n \times \nmax} \Lambda_d \, ,
\end{equation}
where~$\mathbb{X} \coloneqq [ X_1, \dots, X_n]^\top, \Lambda_d\coloneqq [ \lambda_1, \dots, \lambda_{\nmax} ]^\top$
and~$\mathbb{C}^{(1)}_{ab} = c_{ab}$.
For convenience, we also define
\begin{equation}
\mathbb{C}^{(2)} \coloneqq \bigg [\begin{array}{c|c}
    \mathds{1}_n & \mathbb{C}^{(1)}_{n  \times (\nmax-n)}
    \end{array} \bigg ]_{n \times \nmax}    
\end{equation}
so that~$\mathbb{X} \mathbb{X}^\top = \mathbb{C}^{(2)} \Lambda_d \Lambda_d^\top {\mathbb{C}^{(2)}}^\top$.

We can now use~$\CH[\rho_\mathrm{m}] \geq \SLD[\rho_\mathrm{m}] = \min_{\mathbb{X}} \bigtrace \left [  \mathbb{S}_\theta \mathbb{X} \mathbb{X}^\top \right ]$
to get
\begin{equation}
\begin{split}
    \CH[\rho_\mathrm{m}] &\geq \frac{1}{d} \min_{\mathbb{C}^{(1)}} \left ( n + \sum_{a,b} (\mathbb{C}^{(1)}_{ab})^2\right ) \\
    &= \frac{n}{d} = \SLD[\rho_\mathrm{m}] \, .
\end{split}
\end{equation}
Moreover,
as in Appendix~\ref{subsec:ProofHCRB},
$\mathbb{L}=\mathbb{X}\mathbb{X}^\top$
is a valid choice leading to
\begin{equation}
    \left ( \qtrace[\mathbb{S}_\theta \mathbb{L} ]\right )_{j,k} = \frac1{d} \left ( \delta_{jk} + \sum_{l\in [\nmax]\setminus K} \mathbb{C}^{(1)}_{jl} \mathbb{C}^{(1)}_{kl} \right ) \, ,
\end{equation}
which is real, symmetric and gives~$\bigtrace[\mathbb{S}_\theta \mathbb{L} ] = \bigtrace[\mathbb{S}_\theta \mathbb{X} \mathbb{X}^\top ]$. This proves
\begin{equation}
    \CH[\rho_\mathrm{m}] = \frac{n}{d}
\end{equation}
following the same arguments as in Appendix~\ref{subsec:ProofHCRB}.
For estimating all~$\nmax$ parameters
this reduces to Lemma~\ref{lemma:Holevo}.
Numerical checks also verify this result,
as shown, e.g., in Table~\ref{tab:Ratios}.

Notice that~$\Lambda_d \Lambda_d^\top$ is the same as~$\mathbb{X} \mathbb{X}^\top$
from Lemma~\ref{lemma:primalsolutionNHCRB},
so that,
using the fact that~${\mathbb{C}^{(2)}}^\top \mathbb{C}^{(2)} \succcurlyeq 0$,
we can modify~$\mathbb{L}^*$ from Lemma~\ref{lemma:primalsolutionNHCRB}
as
\begin{equation}
    \mathbb{L}^{**}(\mathbb{C}^{(1)}) \coloneqq \mathbb{C}^{(2)} \mathbb{L}^* {\mathbb{C}^{(2)}}^\top \, .
\end{equation}
It then follows from Lemma~\ref{lemma:primalsolutionNHCRB}
that
\begin{equation}
\begin{gathered}
    \mathbb{L}^{*} - \Lambda \Lambda^\top \succcurlyeq 0 \\
    \implies \mathbb{C}^{(2)} \left ( \mathbb{L}^{*} - \Lambda \Lambda^\top \right ) {\mathbb{C}^{(2)}}^\top \succcurlyeq 0 \\
    \implies  \mathbb{L}^{**} - \mathbb{X} \mathbb{X}^\top \succcurlyeq 0 \, .
\end{gathered}
\end{equation}
That this~$\mathbb{L}^{**}$ satisfies
the other NHCRB constraints
($\mathbb{L}^{**}_{jk}  = \mathbb{L}^{**}_{kj}$ Hermitian
from Eq.~\eqref{eq:HNCRBdefn})
for all~$\mathbb{C}^{(1)}$
is also easy to check.

Note that~$\CNH[\rho_\mathrm{m}]$ is now defined by the following minimisation:
\begin{equation}
    \begin{gathered}
        \CNH \coloneqq \min_{\mathbb{L}, \mathbb{C}^{(1)}} \Big \{ \mathbb{T}\mathrm{r} [ \mathbb{S}_\theta \mathbb{L} \; \vert \; \mathbb{L}_{jk} = \mathbb{L}_{kj} \; \mathrm{Hermitian} \, , \\
        \mathbb{L} \succcurlyeq  [\begin{array}{c|c} \mathds{1}_n & \mathbb{C}^{(1)}\end{array}  ] \; \Lambda \Lambda^\top \; { [\begin{array}{c|c} \mathds{1}_n & \mathbb{C}^{(1)}\end{array}  ]}^\top \Big \} \, ,
    \end{gathered}
\end{equation}
whereas if we restrict the minimisation over~$\mathbb{L}$
to a minimisation over our ansatz~$\mathbb{L}^{**}(\mathbb{C}^{(1)})$,
we should get a larger value than~$\CNH$,
i.e.,
\begin{equation}
\label{eq:appinequality2}
    \begin{gathered}
        \min_{\mathbb{L}, \mathbb{C}^{(1)}} \Big \{ \mathbb{T}\mathrm{r} [ \mathbb{S}_\theta \mathbb{L} \; \vert \; \mathbb{L}_{jk} = \mathbb{L}_{kj} \; \mathrm{Hermitian} \, , \\
        \mathbb{L} \succcurlyeq  [\begin{array}{c|c} \mathds{1}_n & \mathbb{C}^{(1)}\end{array}  ] \; \Lambda \Lambda^\top \; { [\begin{array}{c|c} \mathds{1}_n & \mathbb{C}^{(1)}\end{array}  ]}^\top \Big \} \, \\
        \leq \min_{\mathbb{C}^{(1)}} \Big \{ \mathbb{T}\mathrm{r} [ \mathbb{S}_\theta \mathbb{L}^{**} \; \vert \; \mathbb{L}^{**}_{jk} = \mathbb{L}^{**}_{kj} \; \mathrm{Hermitian} \, , \\
        \mathbb{L}^{**} \succcurlyeq  [\begin{array}{c|c} \mathds{1}_n & \mathbb{C}^{(1)}\end{array}  ] \; \Lambda \Lambda^\top \; { [\begin{array}{c|c} \mathds{1}_n & \mathbb{C}^{(1)}\end{array}  ]}^\top \Big \} \, .
    \end{gathered}
\end{equation}
This is because the minimisation on the RHS of Eq.~\eqref{eq:appinequality2} is over a subset of the set over which the minimisation on the LHS is performed.
The quantity on the RHS of Eq.~\eqref{eq:appinequality2}
can then be simplified to
$$\min_{\mathbb{C}^{(1)}} \left \{ \frac{d+1}{d} \left ( n + \sum_{a,b} \left (\mathbb{C}^{(1)}_{ab}\right )^2\right ) \right \} = \frac{(d+1)n}{d} \, .$$
This lets us upper-bound~$\CNH[\rho_\mathrm{m}]$ as
\begin{equation}
    \CNH[\rho_\mathrm{m}] \leq \frac{(d+1) n}{d} \, ,
\end{equation}
which for estimating all~$\nmax$ parameters reduces to
Eq.~\eqref{eq:upperboundNHCRB} from Subsec.~\ref{subsec:ArbitStates}.
Combining with~$\CH[\rho_\mathrm{m}] = n/d$,
we find
\begin{equation}
    \rat_n[\{\rho_\mathrm{m}\}] = \max_{K\subseteq [\nmax], \vert K \vert = n}\frac{\CNH[\rho_\mathrm{m}]}{\CH[\rho_\mathrm{m}]} \leq d+1 \, ,
\end{equation}
as claimed in Theorem~\ref{th:estfewGMMs}.
Numerically,
we see the ratio~$\rat_n[\{\rho_\mathrm{m}\}]$ actually depends on~$n$:
as~$n$ increases up to~$\nmax$,
the ratio increases up to~$d+1$.
\renewcommand{\arraystretch}{1.5}
\begin{table}[htbp]
    \centering
    \begin{tabular}{|c|c|c|c|}
        \hline
         n & $\CH$ & Range:$\CNH$ & \makecell{Max Ratio\\ $\rat_n[\{\rho_\mathrm{m}\}]$} \\ \hline
         2 & $2/3$ & (2/3, 4/3) & 2 \\ \hline         
         3 & $1$ & (3/2, 3) & 3\\ \hline
         4 & $4/3$ & (2.8270, 4.3154) & 3.2365\\ \hline
         5 & $5/3$ & ($25/6$, 6.6427) & 3.9856 \\ \hline
         6 & $2$ & (6, 7.0921) & 3.5461 \\ \hline
         7 & $7/3$  & (8.4369, 8.4951) & 3.6408\\ \hline
         8 & $8/3$  & $32/3$ &  4 \\ \hline
    \end{tabular}
    \caption{HCRB and NHCRB for estimating
    a subset~$\{\lambda_j\}_{j\in K}\subseteq \Lambda_3$
    of GMMs from the maximally-mixed qutrit state~$\rho_\mathrm{m}$.
    The HCRB depends only on the number of parameters,
    $\vert K\vert=n$, but
    the NHCRB depends on the subset~$K$ chosen,
    so we tabulate its range in the third column,
    as (Min NHCRB, Max NHCRB).
    The fourth column lists the maximum
    ratio,~$\rat_n[\{\rho_\mathrm{m}\}]$,
    between the NHCRB and the HCRB,
    taking into account all possible subsets~$\{\lambda_j\}_{j\in K}$.}
    \label{tab:Ratios}
\end{table}
\renewcommand{\arraystretch}{1}
Table~\ref{tab:Ratios}
lists out~$\CH$, the minimum and maximum values
of~$\CNH$ and the maximum ratio~$\rat_n[\{\rho_\mathrm{m}\}]$
for estimating a given number,~$n$, of GMMs from qutrits.
The HCRB only depends on~$n$ but
not on which GMMs are chosen and is equal to~$n/d$.

\section{Summary of Relevant Results from Ref.~\texorpdfstring{\cite{Suzuki2023}}{[46]}}
\label{sec:summaryJunTR}

In this appendix, we summarize the approach
and results from Ref.~\cite{Suzuki2023}
that are relevant for proving
the upper bound to the NHCRB
in main-text Eq.~\eqref{eq:upNHBJunTR},
\begin{equation}
\begin{split}
    \CNH[\rho_\theta] \leq \min_\mathbb{X} \bigg \{ &\ctrace ( \mathbb{Z}_\theta[\mathbb{X}]) \\
    &+   \sum_{j,k\in[n]} \left \Vert \rho_\theta [X_j, X_k] \right \Vert_1 \bigg \} \, .
\end{split}
\end{equation}
Ref.~\cite{Suzuki2023} analyses the concept of gap persistence
between the NHCRB
and the HCRB with increasing
number of copies in multi-parameter quantum estimation.
In doing so, the authors
upper-bound and lower-bound the NHCRB
(Secs.~3.2 \&~3.3, pg. 58 in~\cite{Suzuki2023}).

Say the~$d$-dimensional Hilbert space
of the qudit,~$\mathcal{H}_d$
is combined with the~$n$-dimensional
(complex) parameter space
to define an extended
Hilbert space~$\mathbb{H} = \mathbb{C}^n \otimes \mathcal{H}_d$.
Ref.~\cite{Suzuki2023} then
defines the NHCRB as
\begin{equation}
    \CNH \coloneqq \min_{\mathbb{X}} F_{\mathrm{NH}}(\mathbb{X}) 
\end{equation}
subject to~$\mathbb{X}$ being LUB operators,
and with the Nagaoka-Hayashi (NH)
function~$F_{\mathrm{NH}}(\mathbb{X})$ given by
\begin{equation}
\begin{gathered}
    F_{\mathrm{NH}}(\mathbb{X}) \coloneqq \min_{\mathbb{L}}\big\{\bigtrace[\mathbb{S}\mathbb{L}] \, \vert \,  \mathbb{L} \in \mathcal{L}_{+, \operatorname{sym}}(\mathbb{H}), \\
    \quad \quad \quad \quad \quad \quad \quad \quad \quad \, \mathbb{L} \geq \mathbb{X}\mathbb{X}^{\top}\big\} \, ,
\end{gathered}
\end{equation}
where~$\mathcal{L}_{+, \operatorname{sym}}(\mathbb{H})$
denotes the set of all positive semidefinite
operators~$\mathbb{L}\succcurlyeq 0$ on~$\mathbb{H}$
that are also symmetric under
the partial transpose with respect to the
first Hilbert space, i.e.,~$\mathbb{L}_{jk} = \mathbb{L}_{kj}$ for all~$j,k\in[n]$.

Ref.~\cite{Suzuki2023} then proves that
the NH function can be rewritten as
\begin{equation}
    \begin{split}
   F_{\mathrm{NH}}(\mathbb{X}) &= \ctrace\{\operatorname{Re} \mathbb{Z}_\theta[\mathbb{X}]\} + F_{\mathrm{NH}, 2}(\mathbb{X}) \, , \\
   F_{\mathrm{NH}, 2}(\mathbb{X}) &\coloneqq \min_{\mathbb{V}}\big\{\mathbb{T}\mathrm{r} [\mathbb{V}] \, \vert \,  \mathbb{V} \in \mathcal{L}_{+, \operatorname{sym}}(\mathbb{H}), \\
   & \qquad \quad   \mathbb{V} \geq \operatorname{sym}_{-}(\sqrt{\mathbb{S}_\theta} \mathbb{X} \mathbb{X}^{\top}\sqrt{\mathbb{S}_\theta}) \big\} \, ,
    \end{split}
\end{equation}
where~$\mathbb{Z}_\theta[\mathbb{X}]_{jk} = \qtrace ( \rho_\theta X_j X_k )$ as in Eq.~\eqref{eq:Holevodefn1},
$\mathbb{S}_\theta = \mathds{1}_n\otimes \rho_\theta$,
and~$\mathrm{sym}_{-}(\mathbb{A})=\frac12 (\mathbb{A}-\mathbb{A}^\top)$, with~$^\top$ denoting
partial transpose with respect to parameter space.
This rearrangement makes
\begin{equation}
    \CNH = \min_{\mathbb{X}} \left \{   \ctrace\{\operatorname{Re} \mathbb{Z}_\theta[\mathbb{X}]\} + F_{\mathrm{NH}, 2}(\mathbb{X}) \right \} \, ,
\end{equation}
in which the first term
is equal to the Holevo
objective function in Eq.~\eqref{eq:Holevodefn1}.

Finally, the authors prove
as one of their results
(Theorem 1, Sec.~3.2 in Ref.~\cite{Suzuki2023})
that
the second term of the
NH function is bounded from above as
\begin{equation}
   F_{\mathrm{NH}, 2}(\mathbb{X}) \leq %\frac{1}{2} 
   \sum_{j, k} \left \Vert \sqrt{\rho_\theta}[X_{j}, X_{k}] \sqrt{\rho_\theta}\right \Vert_1 \, ,
\end{equation}
so that the NHCRB can be upper-bounded
as
\begin{equation}
\begin{split}
   \CNH \leq \min_{\mathbb{X}}\bigg\{&\ctrace\{\Re \mathbb{Z}_\theta[\mathbb{X}]\}\\  
   %\frac{1}{2} 
   & \quad \,  + \sum_{j, k}  \left \Vert \sqrt{\rho_\theta} [X_{j}, X_{k}] \sqrt{\rho_\theta} \right \Vert_1 \bigg \}\, ,
\end{split}
\end{equation}
from which
main-text Eq.~\eqref{eq:upNHBJunTR}
follows.
Physically, this argument reveals
the difference between the HCRB and the NHCRB
to originate from the second term,~$F_{\mathrm{NH}, 2}(\mathbb{X})$;
this term captures the non-commutativity inherent in quantum measurements for multi-parameter estimation.
This result from Ref.~\cite{Suzuki2023} forms
a basis of our proof of Theorem~\ref{th:ratngenstate}
showing that the ratio $\CNH[\rho_\theta]/\CH[\rho_\theta]$ is upper-bounded by~$n$.
For further details, readers are encouraged to consult Ref.~\cite{Suzuki2023}, which
presents an in-depth exploration
of the relationship between these two bounds.

\section{Extension to Arbitrary Weight Matrices}
\label{sec:arbitweight}

In this section, we extend the
ratio bound of~$d+1$ for
the linear GMM model
to arbitrary, parameter-independent, positive weight matrices~$W$.
For fair comparison with the
unweighted case,
corresponding to~$W=\mathds{1}_n$,
we trace-normalise~$\ctrace(W)=n$.
Additionally,
$W$ must be real, symmetric
and positive~$(W\succ 0)$.
Below, we shall refer to
estimating from~$\rho_\theta$
under weight matrix~$W$
as estimating from~$(\rho_\theta, W)$
and denote the corresponding
precision bounds by~$\CH^W[\rho_\theta]$
and~$\CNH^W[\rho_\theta]$.

This weighted model corresponds
to reparameterisations
of the linear GMM model~\cite{Lorcan21,Fujiwara1999,ABGG20},
i.e., estimating any~$\nmax$ parameters
that are not necessarily coefficients of the GMMs.
Similar to the other cases
where all~$\nmax$
parameters are estimated,
the unbiased operators
are uniquely fixed to be~$X_j = \lambda_j$.
We first bound the weighted
HCRB and the weighted NHCRB
to prove that the collective enhancement
is at most~$d+1$ for
estimating from~$(\rho_\mathrm{m}, W)$
for any~$W$.
Then, to extend to arbitrary states~$\rho_\theta \neq \rho_\mathrm{m}$,
we numerically
demonstrate that
the collective enhancement for estimating from~$(\rho_\theta, W)$
is always smaller
than the collective enhancement
for estimating from~$(\rho_\mathrm{m}, W)$.
However, we do not prove this.

The weighted HCRB is defined via~\cite{Albarelli2019}
\begin{equation}
    \CH^W[\rho_\theta] \coloneqq \min_{\substack{V\in \mathbb{R}^{n\times n},\\ V=V^\top}} \left \{ \ctrace[W V] \; \vert \; V \succcurlyeq \mathbb{Z}_\theta[\mathbb{X}] \right \} \, ,
\end{equation}
where, by explicit computation
for the maximally-mixed case,~$\mathbb{Z}_\theta[\mathbb{X}]_{jk}=\Tr [ \rho_\mathrm{m} X_j X_k] = \delta_{jk}/d$
or~$\mathbb{Z}_\theta[\mathbb{X}] = 1/d \, \mathds{1}_n$.
Then, it follows
from the positivity of~$W$
that~$V\succcurlyeq 1/d \; \mathds{1}_n$ implies
\begin{equation}
    W V \succcurlyeq \frac{1}{d} W \implies \ctrace[W V] \geq \ctrace[W]/d = \frac{n}{d} \, .
\end{equation}
This proves~$\CH^W[\rho_\mathrm{m}] \geq n/d$.

The weighted NHCRB is defined via~\cite{Lorcan21}
\begin{equation}
\label{eq:weightedNagaoka}
    \begin{split}
    \CNH^W[\rho_\theta] \coloneqq \min_{\mathbb{L}} \Bigg \{ \ctrace[W V] \; \vert \; V = \qtrace[\mathbb{S}_\theta \mathbb{L}] \, \\
    \mathbb{S}_\theta = \mathds{1}_n\otimes \rho_\theta, \, \mathbb{L}_{jk} = \mathbb{L}_{kj} \mathrm{Hermitian} \, ,\\
    \mathbb{L} \succcurlyeq \mathbb{X} \mathbb{X}^\top \Bigg \} \, .
    \end{split}
\end{equation}
Notably, the feasibility constraints on~$\mathbb{L}$
are unchanged from the unweighted case,
i.e.,
the optimal~$\mathbb{L}^*$ from
Lemma~\ref{lemma:primalsolutionNHCRB}
still satisfies~$\mathbb{L}^*_{jk} = \mathbb{L}^*_{kj}$
Hermitian
and~$\mathbb{L}^* \succcurlyeq \mathbb{X} \mathbb{X}^\top$,
despite not being optimal for
the minimisation in Eq.~\eqref{eq:weightedNagaoka}.
This sub-optimal~$\mathbb{L}^*$
thus yields an upper bound to the minimum
in Eq.~\eqref{eq:weightedNagaoka},
\begin{equation}
\label{eq:mmmsnhw}
\begin{split}
    \CNH^W[\rho_\mathrm{m}] \leq \ctrace\left [W \qtrace \left [\frac{1}{d} \mathds{1}_{nd} \mathbb{L}^* \right ] \right ] \\
    = \frac{d+1}{d} \ctrace[W] = \frac{n(d+1)}{d}  \, ,
\end{split}    
\end{equation}
which proves~$\CNH^W[\rho_\mathrm{m}] \leq n (d+1)/d$.
Combining with~$\CH^W[\rho_\mathrm{m}] \geq n/d$
then proves the claim,
\begin{equation}
    \frac{\CNH^W[\rho_\mathrm{m}]}{\CH^W[\rho_\mathrm{m}]} \leq d+1 \, .
\end{equation}

\begin{figure}[htb]
    \centering
    \includegraphics[width=0.8\columnwidth]{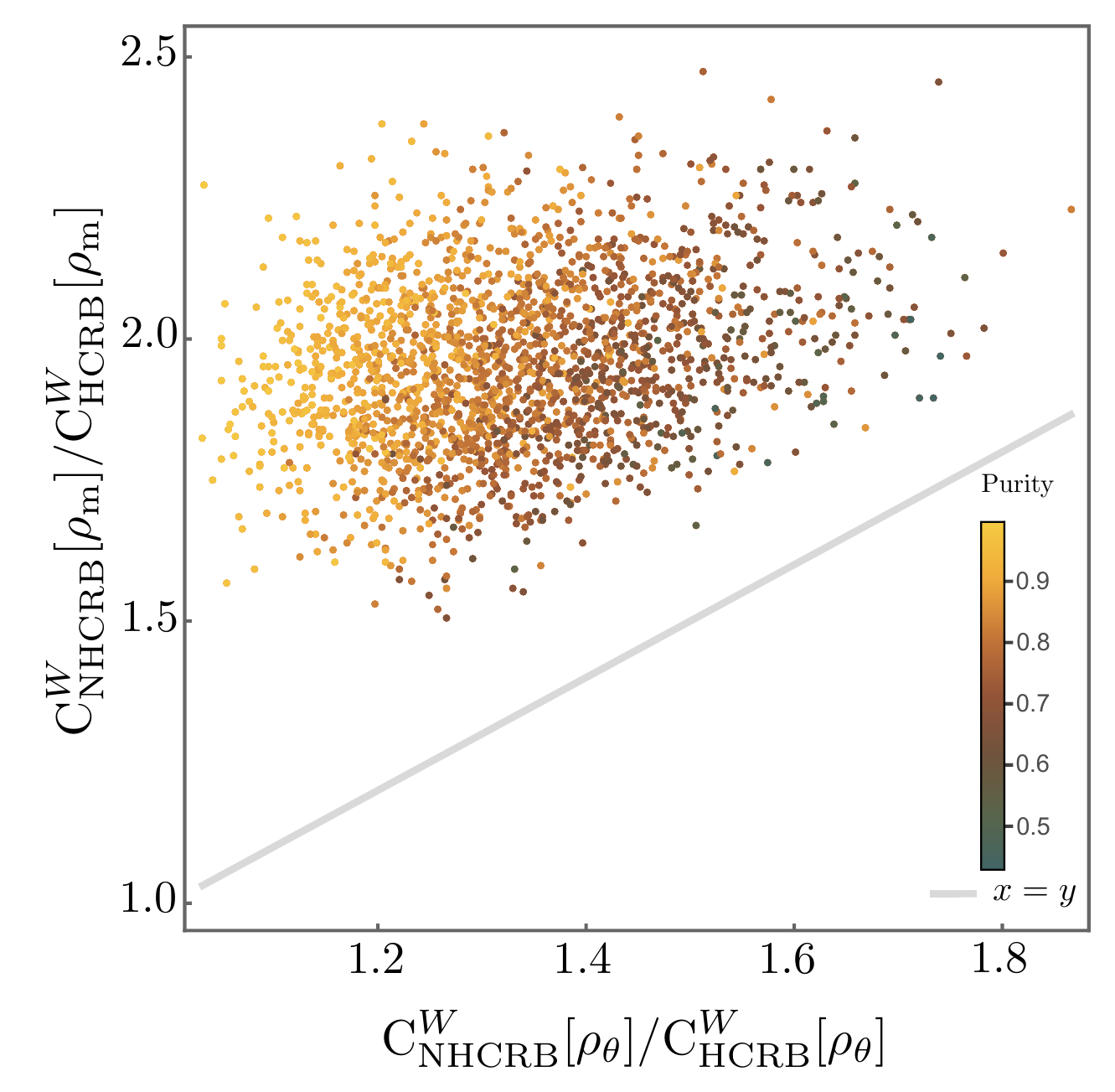}
    \caption{Comparison of the weighted ratio~$\CNH^W/\CH^W$
    for estimating from~$(\rho_\theta, W)$
    to that from~$(\rho_\mathrm{m}, W)$,
    for the full-parameter linear GMM model (5000 samples).
    The points are colour-coded
    by the purity of~$\rho_\theta$
    and the gray line corresponds
    to~$y=x$.}
    \label{fig:weightedtomo}
\end{figure}

So far, in this weighted tomography setting,
which is equivalent to full tomography in arbitrary basis,
we have established the
ratio to be at most~$d+1$
only for the maximally-mixed state.
We do not prove the bound
for arbitrary states
but numerically demonstrate
its validity in Fig.~\ref{fig:weightedtomo}.
By generating random
full-rank, real, symmetric and positive~$W$,
such that~$\ctrace(W)=n$,
and random full-rank states~$\rho_\theta$,
we compare the ratio
for~$(\rho_\theta, W)$,~$\CNH^W[\rho_\theta]/\CH^W[\rho_\theta]$,
to the ratio
for~$(\rho_\mathrm{m},W)$,~$\CNH^W[\rho_\mathrm{m}]/\CH^W[\rho_\mathrm{m}]$.
Repeating this over~5000
random samples of~$W$ and~$\rho_\theta$,
we find the ratio for~$(\rho_\theta,W)$
to always be smaller than the
ratio for~$(\rho_\mathrm{m}, W)$.
This means
\begin{equation}
    \frac{\CNH^W[\rho_\theta]}{\CH^W[\rho_\theta]}
    \leq \frac{\CNH^W[\rho_\mathrm{m}]}{\CH^W[\rho_\mathrm{m}]} \leq d+1 \, ,
\end{equation}
thus establishing the upper bound
of~$d+1$ for arbitrary full-parameter estimation
(or tomography in arbitrary basis)
from any state,
\begin{equation}
    \rat_{\nmax} \leq d+1 \, .
\end{equation}

\section{True Maximum Collective Enhancement in GMM Model}
\label{sec:appMICRB}

The tight bound for individual precision
called the MICRB~$\MI$
was reformulated in Ref.~\cite{HayashiOuyang2023}
using the following conic
optimisation problem (Eq.~(39) in Ref.~\cite{HayashiOuyang2023}),
\begin{equation}
\label{eq:CMIdef2}
    S(P_1) \coloneqq \min_{X\in\mathcal{S}_{SEP}} \left \{ \bigtrace  \left [ \left ( W \otimes \rho \right ) X \right ]  {\big \vert} \mathsf{C}_1, \mathsf{C}_2 \mathrm{hold}. \right \}, 
\end{equation}
where~$W$ is the~$(n+1)\times (n+1)$ weight matrix
defined with~$W_{11} = W_{1j} = W_{j1} = 0$
for all~$j \in \{1, \dots, n+1\}$,
and~$\mathsf{C}_1,\mathsf{C}_2$
refer to two equality constraints
on the~$(n+1)d\times(n+1)d$-sized
operator~$X$,
%\begin{equation}
\begin{equation}
\label{eq:const1}
    \mathsf{C}_1: \, \, \,    \Tr_\mathcal{X}\left [ ( \ketbra{0} \otimes \mathds{1}_d) X \right ] = \mathds{1}_d \, ,
\end{equation}
and
\begin{equation}
\begin{split}    
    \label{eq:const2}
    \mathsf{C}_2: \, \, \bigtrace  \left [ \left ( \frac{\ket{0}\bra{j} + \ket{j}\bra{0}}{2} \otimes \partial_k \rho \right ) X\right ] &=  \delta_{jk} \\
        (\text{for all} \; 1 < j, k \leq n&+1)  \, ,
\end{split}
\end{equation}
which ensure the POVM constraint
and the local unbiasedness constraint, respectively.
The constraint~$\mathsf{C}_1$
involves partial tracing
over space~$\mathcal{X}$, which
denotes an~$(n+1)$-dimensional
real vector space spanned by~$\{\ket{0}, \dots, \ket{n}\}$.
Lastly, the
optimisation domain in Eq.~\eqref{eq:CMIdef2}
is the separable cone~$\mathcal{S}_{SEP}$,
denoting
the convex hull of~$(n+1)d$-dimensional operators
that are tensor products of
real, symmetric, positive-semidefinite operators on the space~$\mathcal{X}$
(denoted~$\mathcal{M}_{rs,+}(\mathcal{X})$ in Ref.~\cite{HayashiOuyang2023}),
and complex, Hermitian, positive-semi-definite operators
on the~$d$-dimensional Hilbert space~$\mathcal{H}_d$
(denoted~$\mathcal{B}_{sa,+}(\mathcal{H}_d)$ in Ref.~\cite{HayashiOuyang2023}).
In the notation of Ref.~\cite{HayashiOuyang2023},
this cone is denoted~$\mathcal{S}_{SEP} \coloneqq \mathrm{conv}(\mathcal{M}_{rs,+}(\mathcal{X}) \otimes \mathcal{B}_{sa,+}(\mathcal{H}_d))$.

For the problem of GMM tomography
of arbitrary qudit states,
the derivatives appearing in constraint~$\mathsf{C}_2$,
$\partial_k \rho = \lambda_k$,
are parameter independent.
Thus, feasibility for the MICRB minimisation
(comprising inclusion in the separable cone~$X\in\mathcal{S}_{SEP}$,
and constraints~$\mathsf{C}_1$ \&~$\mathsf{C}_2$)
does not depend on the state~$\rho$ being estimated,
i.e., the true values of the parameters~$\theta$.
Let us define a candidate solution~$X_\mathrm{sol}$
to the MICRB in a block-wise manner,
\begin{equation}
\label{eq:candSoldef}
    (X_\mathrm{sol})_{jk} \coloneqq \begin{cases}
        \mathds{1}_d & j=k=1 \\
        \lambda_j & j > k = 1 \\
        \lambda_k & k > j = 1 \\
        \frac{d+1}{d+2} \left ( \{\lambda_j, \lambda_k\} + \delta_{jk} \mathds{1}_d \right) & j,k > 1
    \end{cases} \, ,
\end{equation}
where block indices~$j,k$ run from~1
to~$n+1$.
In fact, this candidate solution
can be rewritten as
\begin{equation}
    X_\mathrm{sol} = \begin{bmatrix}
        \mathds{1}_d & \mathbb{X}^\top \\
        \mathbb{X} & \mathbb{L}^*
    \end{bmatrix}
\end{equation}
where~$\mathbb{X}=\{\lambda_1, \dots, \lambda_n\}^\top$
(from Eq.~\eqref{eq:optXsol}) and~$\mathbb{L}^*$
(from Eq.~\eqref{eq:optLmat})
are the NHCRB-optimal solutions
(see Lemma~\ref{lemma:primalsolutionNHCRB}).
The NHCRB optimisation constraint
was the positivity of~$X_\mathrm{sol}$ above~\cite{HayashiOuyang2023},
which was proven
through~$\mathbb{L}^* - \mathbb{X}\mathbb{X}^\top \succcurlyeq 0$
in Lemma~\ref{lemma:primalsolutionNHCRB} (Appendix~\ref{sec:ProofLemmaNHCRBmaxmix}).
Thus, Lemma~\ref{lemma:primalsolutionNHCRB} solved
the NHCRB for the maximally-mixed state.

It can be easily checked
that~$X_\mathrm{sol}$
satisfies both constraints~$\mathsf{C}_1$,
through the $j=k=1$ term,
and~$\mathsf{C}_2$,
through the~$j>k=1$ and the~$k>j=1$ terms.
The proof that~$X_\mathrm{sol}$
belongs to the separable cone~$\mathcal{S}_{SEP}$
(deferred to Lemma~\ref{lemma:proofSeparableCone} below)
follows by expressing~$X_\mathrm{sol}$
as the sum over~$d^2$ operators,~$\sum_{l\in[d^2]} \Xi_l \otimes \Pi_l$,
where each~$\Xi_l$ is an estimator
matrix in~$\mathcal{M}_{rs,+}(\mathcal{X})$
and~$\Pi_l$ is the~$l^\text{th}$ element
of a SIC POVM, therefore belonging to~$\mathcal{B}_{sa,+}(\mathcal{H}_d)$.
Thus,~$X_\mathrm{sol}$ is feasible for the MICRB
minimisation for GMM tomography of arbitrary qudits.

In fact, while~$X_\mathrm{sol}$
is feasible (not necessarily optimal) for all true states,
it is the MICRB-optimal solution for the maximally-mixed state.
Tomography in the GMM basis corresponds to
the identity-weighted full GMM model,
so the appropriate matrix~$W$ is
\begin{equation}
    W_\mathrm{id} = \begin{bmatrix}
        0 & 0 & \dots 0 \\
        0 & 1 & \dots 0 \\
        0 & 0 & \dots 0 \\
        0 & 0 & \dots 1 
    \end{bmatrix}_{(n+1)\times (n+1)} \, .
\end{equation}
In this case,
the objective in Eq.~\eqref{eq:CMIdef2}
equals~$\bigtrace \left [ (\mathds{1}_n \otimes \rho ) X_{22} \right ]$
where~$X_{22}$ denotes the block of~$X$
starting from row and column indices~$d+1$
up to indices~$(n+1)d$.
For the candidate solution,~$X_{22}$ is~$\mathbb{L}^*$,
so the objective equals~$\bigtrace (\mathbb{S}_\theta \mathbb{L}^*)$
where~$\mathbb{S}_\theta = \mathds{1}_n \otimes \rho$
as previously defined.
This objective value is identical to
the NHCRB objective (see Eq.~\eqref{eq:HNCRBdefn})
that was computed in Lemma~\ref{theorem:NHCRB}
(and proved in Appendix~\ref{subsec:ProofOptimality}).
For the maximally-mixed state~$\rho_\mathrm{m} = \mathds{1}_d/d$,
the objective value attained by~$X_\mathrm{sol}$
therefore equals
\begin{equation}
\label{eq:Xsolobjective}
\begin{split}
    \bigtrace \left [ (W_\mathrm{id} \otimes \rho_\mathrm{m}) X_\mathrm{sol} \right ] = \frac{\bigtrace[\mathbb{L}^*]}{d}  \\
    = \frac{n(d+1)}{d}  = \CNH[\rho_\mathrm{m}] \, .
\end{split}
\end{equation}
As~$X_\mathrm{sol}$ is feasible,
this upper-bounds the minimum in Eq.~\eqref{eq:CMIdef2}
as~$\MI[\rho_\mathrm{m}] \leq \CNH[\rho_\mathrm{m}]$,
whereas, by definition,~$\CNH[\rho_\mathrm{m}] \leq \MI[\rho_\mathrm{m}]$,
thus proving~$\MI[\rho_\mathrm{m}]=\CNH[\rho_\mathrm{m}]$
and the optimality of $X_\mathrm{sol}$.
Further, since~$X_\mathrm{sol}$ is feasible for any~$\rho$,
the objective~$\bigtrace [ (W_\mathrm{id} \otimes \rho) X_\mathrm{sol}]$
also upper-bounds the minimisation
in Eq.~\eqref{eq:CMIdef2}
for any other qudit state~$\rho$.
By direct computation (see Eq.~\eqref{eq:calSthetatrace})
we have
\begin{equation*}
    \bigtrace \left [ (W_\mathrm{id} \otimes \rho) X_\mathrm{sol} \right ] = \frac{n(d+1)}{d} \, ,
\end{equation*}
so that,
taking into account
the correction for
non-zero true parameter values,
this proves~$\MI[\rho_\theta] \leq \frac{n(d+1)}{d} - \sum_j \theta_j^2$,
similar to Eq.~\eqref{eq:upperboundNHCRB} for the NHCRB.

Then, the same argument used
to prove Theorem~\ref{th:ratdarbstate}
for the NHCRB can be used to prove
the analogue theorem below for the MICRB.
\begin{theorem}
\label{th:ratdMIarbstate}
    For tomography on arbitrary~$d$-dimensional qudit states~$\rho_\theta$, the maximum true collective enhancement~$\ratMI[\{\rho_\theta\}] = \max_{\theta \in \Theta} \MI[\rho_\theta]/\CH[\rho_\theta] \leq d+2$. 
\end{theorem}
\begin{proof}
Combining the lower bound for the HCRB in Eq.~\eqref{eq:lowerboundHCRB}
with the upper bound for the MICRB
given by~$\MI[\rho_\theta] \leq \frac{n(d+1)}{d} - \sum_j \theta_j^2$,
we get
\begin{equation}
\label{eq:ratioMIeq2}
    \frac{\MI[\rho_\theta]}{\CH[\rho_\theta]} \leq \frac{d^2+d-1-\mathrm{P}(\rho_\theta)}{d-\mathrm{P}(\rho_\theta)} .
\end{equation}
Then, using $1/d\leq \mathrm{P}(\rho_\theta)\leq 1$,
we find the maximum of the right hand side of Eq.~\eqref{eq:ratioMIeq2}
to be~$d+2$, attained when~$\mathrm{P}(\rho_\theta)=1$,
i.e., when~$\rho_\theta$ is pure.
\end{proof}

We now prove that~$X_\mathrm{sol}$
belongs to the separable cone~$\mathcal{S}_{SEP}$.
The proof proceeds utilising the fact
that~$\mathbb{X}$ and~$\mathbb{L}^*$ appearing in~$X_\mathrm{sol}$
originate from the optimal solution
to the NHCRB for the maximally-mixed state,
where SIC POVMs (if they exist) are the optimal measurements
(proved in Lemma~\ref{th:CFI}, Appendix~\ref{subsec:ProofCFI}).
This fact connects the SIC POVM elements,
denoted~$\Pi^*_l$ for~$1\leq l \leq d^2$,
to the optimal unbiased operators,~$X_j = \lambda_j$,
and the optimal operator~$\mathbb{L}^*$
through a real, linear transform.
This transformation can then be used
to construct classical matrices~$\Xi_l \in \mathcal{M}_{rs,+}(\mathcal{X})$
that prove~$X_\mathrm{sol} = \sum_{l\in[d^2]} \Xi_l \otimes \Pi^*_l$.

\begin{lemma}
\label{lemma:proofSeparableCone}
    The candidate solution~$X_\mathrm{sol}$
    defined in Eq.~\eqref{eq:candSoldef}
    belongs to the
    separable cone~$\mathcal{S}_{SEP}$.
\end{lemma}
\begin{proof}
Let~$\{\Pi_1^*, \dots, \Pi_{d^2}^*\}$ denote an optimal SIC POVM
(if one exists) attaining the NHCRB
for the identity-weighed
full GMM model on maximally-mixed state.
Then, from the NHCRB definition~\cite{Lorcan21},
there exists an~$n\times d^2$ real matrix~$\xi$
(given by~$\hat{\theta}_{jl}^*-\theta_j$,
where~$\hat{\theta}_{jl}$ is the classical
estimator function that assigns a value
to~$\theta_j$ based on outcome~$l$)
that simultaneously satisfies
the relations
(Eqs. (12) \& (13) in Ref.~\cite{Lorcan21})
\begin{equation}
    \mathbb{L}^* = \sum_{l\in[d^2]} \begin{pmatrix} \xi_{1l} \\ \xi_{2l} \\ \vdots \\ \xi_{nl} \end{pmatrix}\begin{pmatrix} \xi_{1l} & \xi_{2l} & \dots & \xi_{nl} \end{pmatrix} \otimes \Pi_l^*
\end{equation}
and
(Eq.~(17) in Ref.~\cite{Lorcan21})
\begin{equation}
    X_j = \lambda_j = \sum_{l\in[d^2]} \xi_{jl} \Pi^*_l  \quad (\text{for all} \; j \in [n]) \, ,
\end{equation}
thus connecting the optimal solution~$\mathbb{L}^*$,
the optimal unbiased operators~$\mathbb{X}$ and
the optimal SIC POVM~$\{\Pi_l^*\}$.
Finally, we can decompose~$\mathds{1}_d = \sum_{l\in[d^2]} \Pi_l^*$
using the POVM.
This lets us rewrite the candidate solution as
\begin{equation}
\begin{gathered}
    X_\mathrm{sol}  = \begin{bmatrix}
        \mathds{1}_d & \mathbb{X} \\
        \mathbb{X}^\top & \mathbb{L}^*
    \end{bmatrix} \\
    = \sum_{l\in[d^2]} \begin{bmatrix}
        \Pi_l^* &  \xi_{1l} \Pi_l^*    & \dots & \xi_{nl} \Pi_l^* \\
        \xi_{1l} \Pi_l^* & \xi_{1l}^2 \Pi_l^* & \dots & \xi_{1l} \xi_{nl} \Pi_l^*  \\
        \vdots & \vdots &\ddots & \vdots \\
        \xi_{nl} \Pi_l^* & \xi_{1l} \xi_{nl} \Pi_l^* & \dots & \xi_{nl}^2 \Pi_l^*
    \end{bmatrix}    \\
    {=} {\sum_{l\in[d^2]}} \underbrace{\begin{bmatrix}
        1 &  \xi_{1l}   & \dots & \xi_{nl} \\
        \xi_{1l}  & \xi_{1l}^2   & \dots & \xi_{1l} \xi_{nl}  \\
        \vdots & \vdots &\ddots & \vdots \\
        \xi_{nl}  & \xi_{1l} \xi_{nl} & \dots & \xi_{nl}^2 
    \end{bmatrix}}_{\Xi_l \in \mathcal{M}_{rs,+}(\mathcal{X})}   {\otimes} \underbrace{\Pi_l^*}_{\mathcal{B}_{sa,+}(\mathcal{H}_d)}
\end{gathered}
\end{equation}
thus proving~$X_\mathrm{sol} \in \mathcal{S}_{SEP}$.
\end{proof}

\section{Gill-Massar Cram\'{e}r-Rao Bound}
\label{sec:appNHvsGM}

In fact,
the Gill-Massar CRB (GMCRB)~\cite{Gill2000},
\begin{equation}
 \CGM[\rho_\theta] \coloneqq \frac{\left ( \ctrace[J_\mathrm{SLD}^{-1/2}]\right)^2}{d-1} \, ,
\end{equation}
which is obtained
by inserting the classical CRB,~$\ctrace(V_\theta) \geq \ctrace(J^{-1})$,
into Eq.~\eqref{eq:GMinequality},
is identical to the NHCRB
for the full-parameter linear GMM model.
This follows from
the inequality in Eq.~\eqref{eq:GMinequality}
being saturated in this case
(Sec.~VC below Eq.~(54) in \cite{Gill2000}).
We further numerically verify this equivalence
in Fig.~\ref{fig:GMNHcomp},
where we plot the two bounds
for estimating all~8 GMM coefficients
from 2000 random qutrit states.
Both bounds agree for this model,
as evidenced by the points all
lying on the~$y=x$ line.
However, this equivalence
raises the question of why
we choose the NHCRB over the GMCRB
as our main tool
to quantify finite-copy precision,
which we now answer.

\begin{figure}[htb]
    \centering
    \includegraphics[width=0.8\columnwidth]{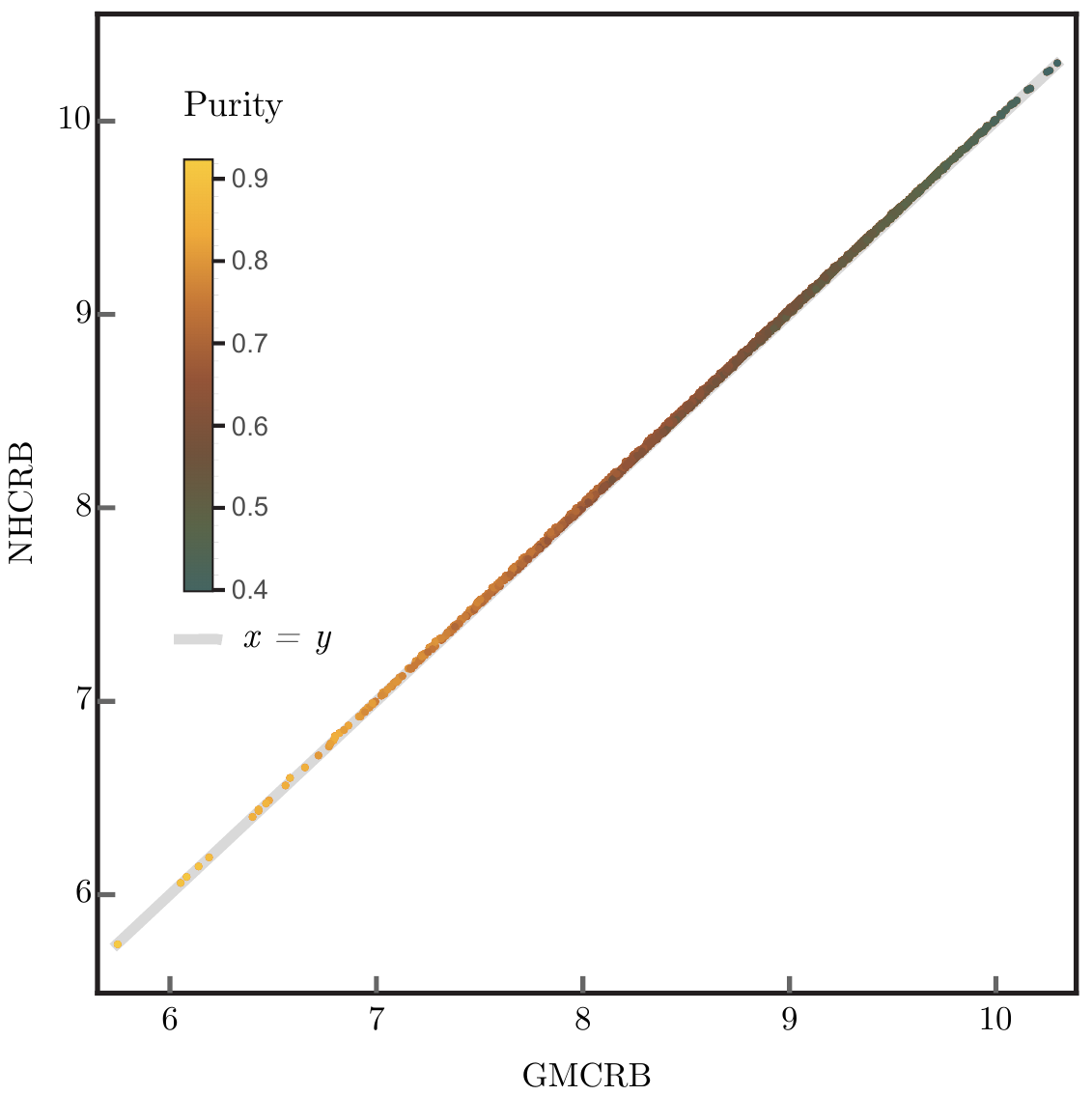}
    \caption{Comparison of the NHCRB and the GMCRB for tomography
    in the GMM basis~($n=8$, $d=3$). The two bounds are equal for the 2000 randomly-generated states and are color-coded by purity of the state.}
    \label{fig:GMNHcomp}
\end{figure}

\begin{figure}
    \centering
    \includegraphics[width=0.9\columnwidth]{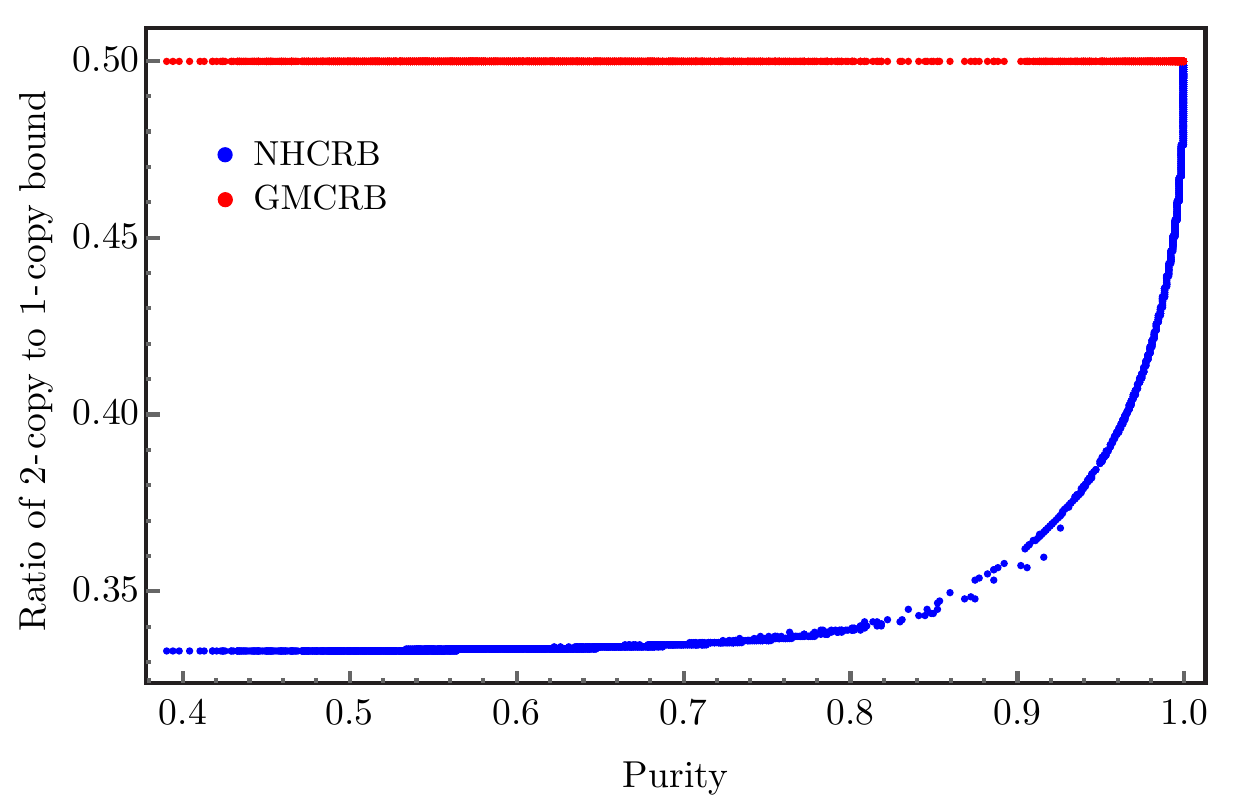}
    \caption{Comparison of the two-copy to one-copy ratio
    for the NHCRB and the GMCRB.
    Bounds correspond to tomography in the GMM basis
    ($n=8$, $d=3$)
    for 5000 random states.
    The GMCRB is additive and underestimates the two-copy enhancement
    except for pure states, where the two bounds agree
    and there is no two-copy enhancement.}
    \label{fig:GMNHcomp2}
\end{figure}

\begin{figure}
    \centering
    \includegraphics[width=0.8\columnwidth]{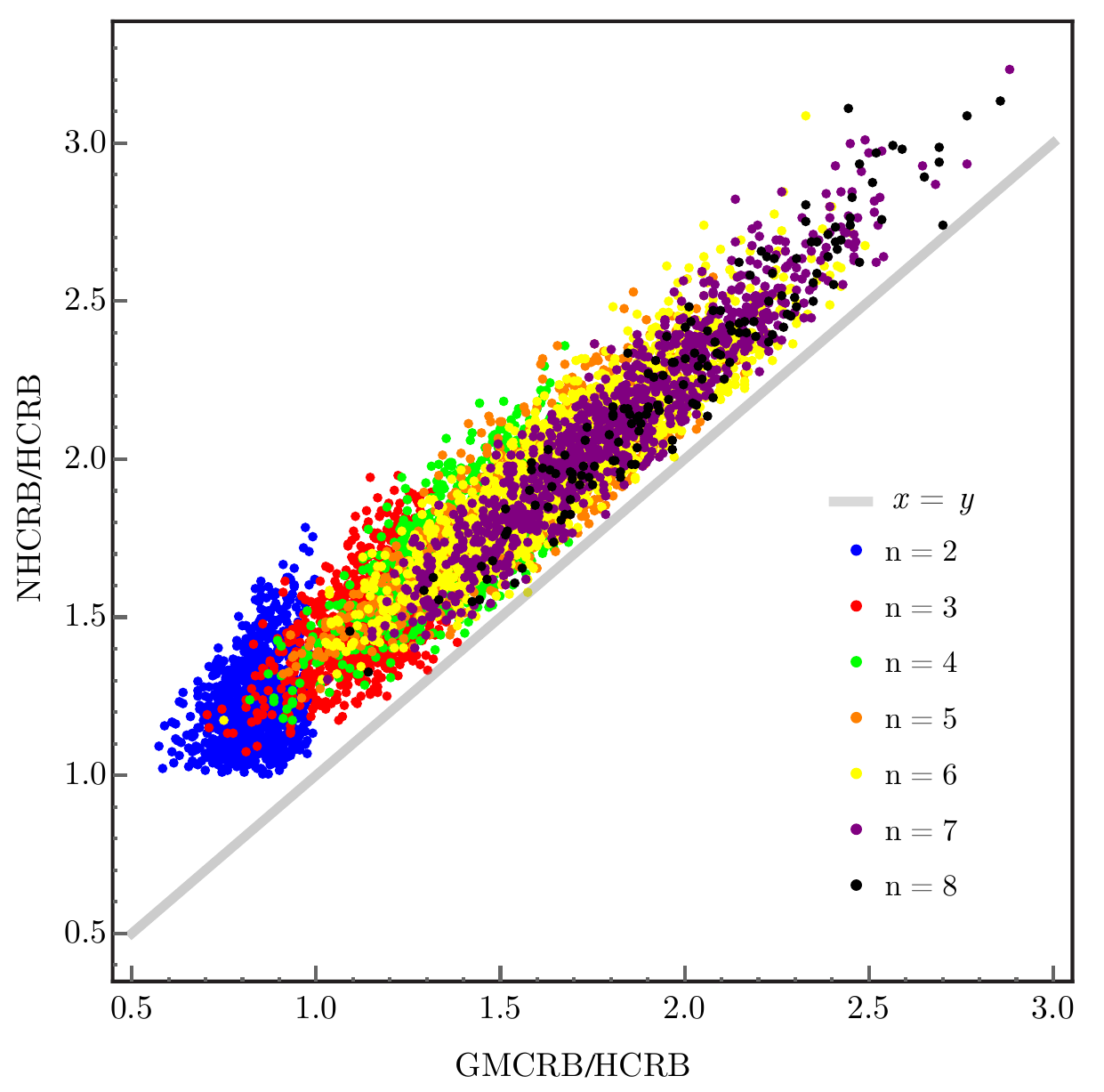}
    \caption{Comparison of the NHCRB and the GMCRB,
    normalised by the HCRB, for estimating fewer than~$\nmax$
    independent (but non-orthogonal)
    parameters via individual measurements. The bounds are
    calculated for 2000 random qutrit states with number of
    parameters~$n$ ranging from two to eight.
    The NHCRB is tighter than the GMCRB in this case,
    even for full-parameter models (black points).}
    \label{fig:GMNHfewpara}
\end{figure}

In short, the NHCRB is
generally tighter
than the GMCRB across a range
of qudit estimation models
(e.g., models comprising a few parameters,~$n<\nmax$,
see Figs.~\ref{fig:GMNHratioComp},~\ref{fig:GMNHfewpara})
and thus is better-suited
for the exploration of model-free quantities
like~$\ratMI_n$ and~$\ratMI$.
In Fig.~\ref{fig:GMNHratioComp},
we compare the GMCRB-to-HCRB ratio~$\ratGM_n$
(dark gray bars, blue line)
and the NHCRB-to-HCRB ratio~$\rat_n$
(light gray bars, red line)
by combining ratios from known analytic models
and from~1300
randomly-sampled numerical models
for each~$n$ from 1 to~$\nmax=8$ and~$d=3$.
The NHCRB ratio satisfies~$\rat_n\leq n$,
whereas the GMCRB ratio satisfies~$\ratGM_n \leq n/(d-1)$,
proved in Sec.~\ref{subsec:analresultsothrbnds}.

Similarly, the NHCRB
shows a sub-additive scaling with number of copies
similar to~$\MI$, whereas
the standard GMCRB is additive with number of copies.
This means that NHCRB ratios
can be directly used to compare multi-copy collective precision
to separably-attainable precision,
whereas the GMCRB requires a modification
for the two-copy setting~\cite{Zhu1}.
The multi-copy GMCRB~\cite{Gill2000} is
defined via
\begin{equation}
\begin{split}
    \CGM[\rho_\theta^{\otimes k}] &\coloneqq \min_{V_\theta^{(k)} \succcurlyeq 0} \Big \{ \ctrace( V_\theta^{(k)}) {\Big \vert}  \\
    &\ctrace (J_\mathrm{SLD}^{-1} (k V^{(k)}_\theta )^{-1} ) \leq d-1 \Big \} \, ,
\end{split}
\label{eq:GMCRBdefn}
\end{equation}
where~$^{(k)}$ represents $k$-copy quantities.
Rephrasing the minimisation
in Eq.~\eqref{eq:GMCRBdefn}
in terms of~$k V_\theta^{(k)}$
directly leads to
\begin{equation}
    \CGM[\rho_\theta^{\otimes k}] = \frac{1}{k} \CGM[\rho_\theta] \, ,
    \label{eq:GMadditive}
\end{equation}
meaning the GMCRB is additive
for measuring~$k$ copies of~$\rho_\theta$
simultaneously.
This complements
the well-known additivity of the SLD QFI
(Eqs.~(72)~\&~(73) in~\cite{TA14}),
on which the GMCRB is based.
We note that
the minimisation
in Eq.~\eqref{eq:GMCRBdefn}
has the closed-form solution:
\begin{equation}
 \CGM[\rho_\theta^{\otimes k}] = \frac{\left ( \ctrace[J_\mathrm{SLD}^{-1/2}]\right)^2}{k(d-1)} \, .
\end{equation}

For~$k=2$, Eq.~\eqref{eq:GMadditive}
implies that the ratio of the two-copy bound
to the one-copy bound is exactly half
for the GMCRB,
as can be seen in Fig.~\ref{fig:GMNHcomp2}.
In Fig.~\ref{fig:GMNHcomp2},
we compare the ratio of two-copy to one-copy
bounds for the NHCRB and the GMCRB
over 5000 randomly generated qutrit states.
It is clear that the NHCRB
is not additive with respect to number of copies;
instead, the two-copy NHCRB is always
smaller than the two-copy GMCRB,
except for pure states where the two bounds agree.
This subadditivity of the NHCRB
and additivity of the GMCRB
can be attributed to the fact
that
the $k$-copy GMCRB considers
individually measuring each
of the~$k$ copies,
whereas the~$k$-copy NHCRB
considers measuring the $k$-copies
simultaneously or collectively.
As a result, the gap between
two-copy NHCRB and GMCRB represents
the increase in precision from
two-copy measurements compared to
one-copy measurements.
Notably the optimal Fisher information
is also not additive under tensoring.

Moreover, for estimating fewer than~$\nmax$ parameters,
the NHCRB is strictly higher than the GMCRB
even in the one-copy case,
i.e., the former is a tighter bound.
Figure~\ref{fig:GMNHfewpara} depicts this
by considering the estimation of~2 to~8
arbitrary parameters from 2000 randomly generated
qutrit states
(following the same methodology
as used in Fig.~\ref{fig:estimatearbitstatearbitparas}
to generate the states and parameters).
The GMCRB and NHCRB are computed for this model
and are both normalised by the HCRB.
It is clear that all the plotted points
lie above the~$y=x$ line,
numerically demonstrating that the NHCRB
is tighter than the GMCRB in this case.
Nonetheless, Fig.~\ref{fig:GMNHfewpara} also
reveals an increasing trend of the
ratio between the individual-optimal
and collective-optimal precisions
with number of parameters,
irrespective of the particular choice
of the individual-precision bound.

\begin{figure}[h]
    \centering
    \includegraphics[width=\columnwidth]{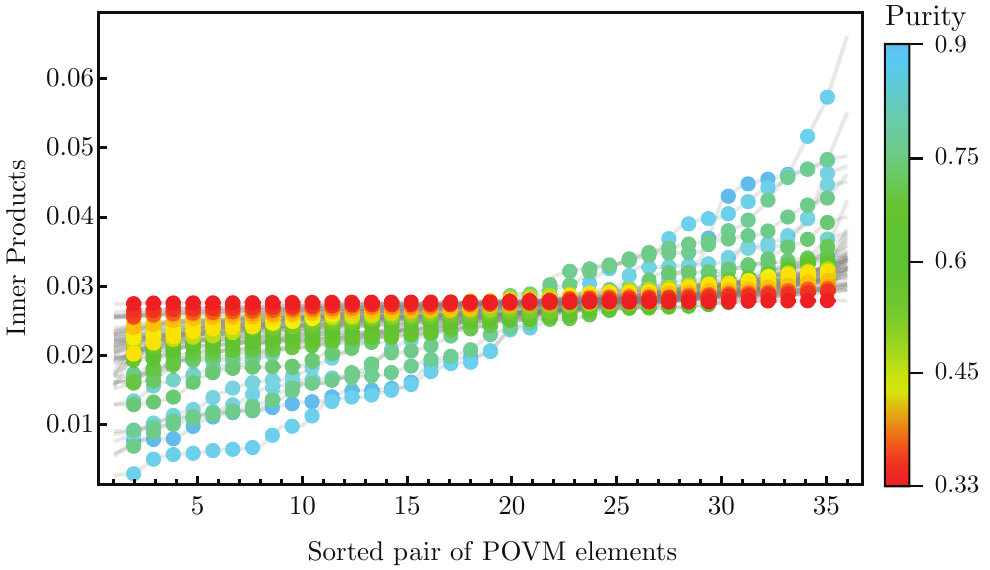}
    \caption{Evolution of the qutrit individual-optimal measurements
    from SIC POVM (red) to IC POVMs (all other colours)
    as purity of~$\rho_\theta$
    increases from 1/3 for the maximally-mixed state to
    1 for pure states.}
    \label{fig:InnerProds}
\end{figure}

\section{Optimal IC POVMs for Arbitrary States}
\label{sec:appOptICPOVMs}

In this section we numerically investigate
the optimal POVMs saturating the NHCRB
for the full-parameter linear GMM model
and for arbitrary
states~$\rho_\theta^*$.
As the purity of~$\rho_\theta^*$
increases from~$1/d$ to~$1$,
the optimal individual measurements
evolve from SIC POVMs to distorted IC POVMs.
This transition is depicted in Fig.~\ref{fig:InnerProds},
where the
inner products between the POVM vectors
are equal at minimum purity but spread out
with increasing purity.
For Fig.~\ref{fig:InnerProds},
we first generate 500 random mixed qutrit states
by uniformly-randomly choosing
the parameters~$\{\theta_j\}$
and rejection-sampling
to ensure the positivity of~$\rho_\theta$.
For each state,
we numerically solve for the optimal
one-copy,~$d^2$-element, rank-one POVM
and ensure that it saturates the NHCRB.
Then we compute the inner-product between
every pair of elements of this optimal POVM.
We then bin the states into 57 purity
intervals and average the
sorted list of inner products
over each interval.
Finally, we plot these sorted inner-products
for each purity interval, colour-coded by
the average purity of that interval.

\begin{figure*}[htbp]
    \centering
    \includegraphics[width=0.8\textwidth]{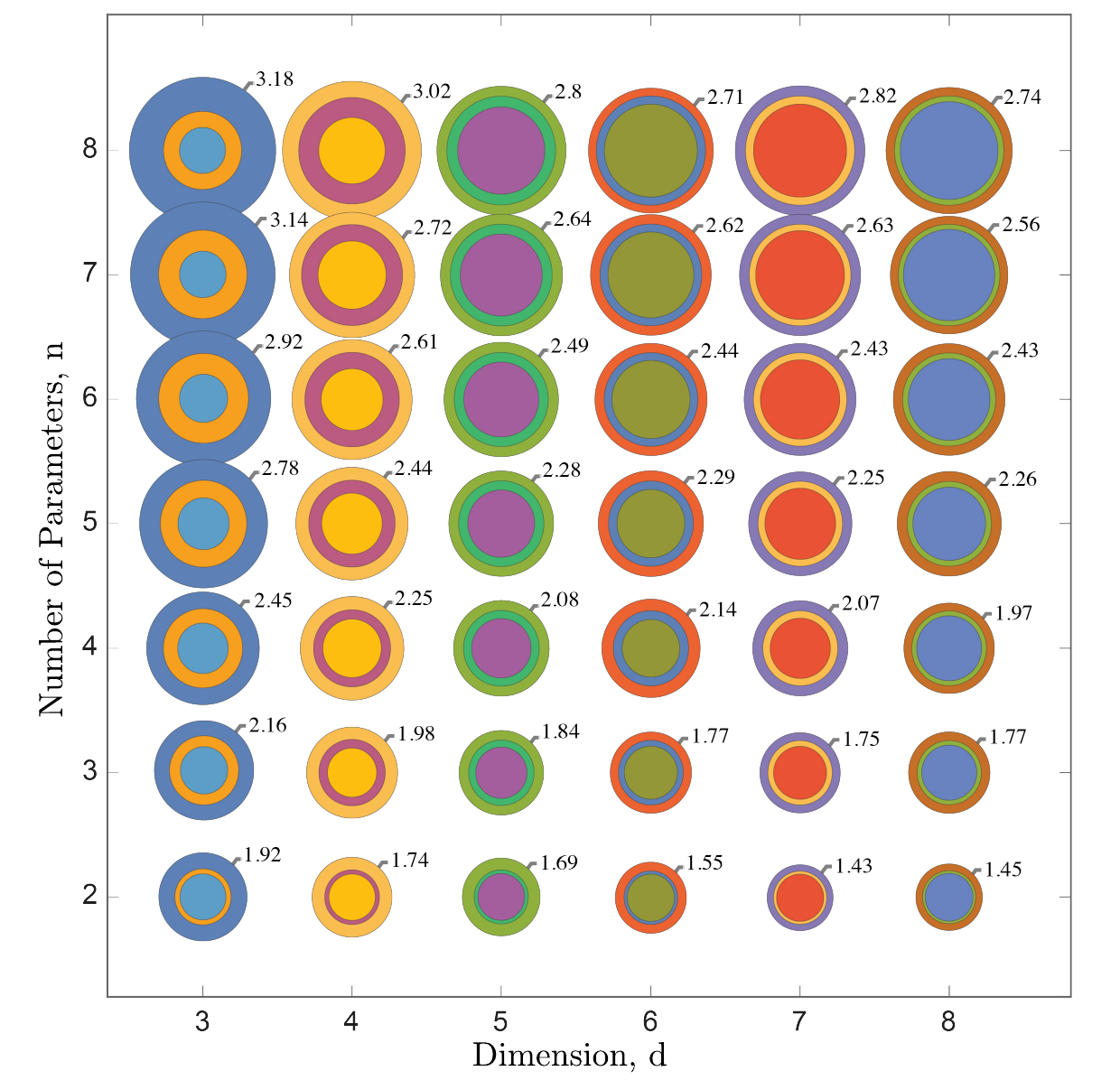}
    \caption{%
    Bubble plot of random sampling
    data for the ratio between the NHCRB and the HCRB
    for estimating arbitrary parameters
    from arbitrary qudit states.
    Bubbles are plotted on a grid
    over qudit dimension~$d$
    and number of parameters~$n$.
    The size (diameter) of the bubbles
    indicate the minimum,
    the average and the maximum ratios sampled
    for each~$d$ and~$n$, and
    the numerical labels are the
      maximum ratio up to three significant figures.}
    \label{fig:GridPlotDN}
\end{figure*}

\section{Random-Sampling of States and Parameters}
\label{sec:appGridPlot}

For the random-sampling experiments
in Fig.~\ref{fig:estimateallarbitstate},
we generate random mixed qudit states
by first generating
an entry-wise random~$d\times d$ complex matrix~$S$,
and then assigning~$\rho_\theta =
S S^\dagger/\qtrace(S S^\dagger)$.
This procedure ensures~$\rho_\theta =
\rho_\theta^\dagger$,
$\rho_\theta \succcurlyeq 0$
and~$\qtrace(\rho_\theta) = 1$.
The true GMM coefficients
($\theta^*$ for the GMM model)
can be found via~$\Tr(\rho_\theta \lambda_j)$.
Unfortunately, this procedure
generates low-purity states
with a much higher probability
than high-purity states, which becomes
a problem for~$d=3$ and~$4$.
We circumvent this issue
by generating additional samples
of the form~$(1-p) \rho_\theta + p \, \mathds{1}_d/d$
and~$(1-p)\rho_\theta + p \ketbra{+}_d$,
where~$p \in [0,1]$.
This sampling method is
non-uniform but our aim here is
not to sample uniformly
according to some measure,
but rather to find models
with extremal properties.
We compute
the ratio for the full-parameter
linear GMM model for all these states,
the random samples and their convex combinations,
to produce the yellow points in Fig.~\ref{fig:estimateallarbitstate}.
The ratio-maximising (red) and ratio-minimising (blue) states
at fixed purity are found by numerically
maximising and minimising the ratio
over the state space.

For the random-sampling experiments
in Figs.~\ref{fig:estimatearbitstatearbitparas}
and~\ref{fig:GridPlotDN},
we generate random mixed qudit states
by the following technique.
For each~$d$ and~$n$,
we uniformly-randomly choose~$\nmax$
coefficients~$\{\phi_j\}_{j\in [\nmax]}$
from the interval~$\left [-\sqrt{(d-1)/d}, \sqrt{(d-1)/d} \right ]$.
These define a
random state~$\rho_\theta = \mathds{1}_d + \sum_{j\in [\nmax]} \phi_j \lambda_j$
which is guaranteed to be trace-one and Hermitian,
but not positive.
We ensure the positive semi-definiteness of~$\rho_\theta$
by rejection sampling (discarding if it is not positive).
This process generates a valid random qudit state.
Next we generate the~$n$ arbitrary
parameters~$\{\theta_j\}_{j\in [n]}$
by generating at random the parameter derivatives
$\partial_j \rho_\theta$,
which must be Hermitian and traceless.
We do this by writing each~$\partial_j \rho_\theta$
in the GMM basis and randomly generating
the coefficients in this basis.
Then we rejection-sample to ensure
the~$n$ parameter derivatives are linearly-independent,
and lead to a valid model.

Figure~\ref{fig:GridPlotDN} indicates that
the~$10^4$ number of samples is relatively small
for higher~$d$ and~$n$---the minimum ratio observed,
which should be close to one, is much larger
for large~$d$ and~$n$.
This is because our sampling method
generates states with low purity with higher probability
and states with high purity with lower probability.
As a result, the increasing or decreasing trends
of the maximum observed ratio
with~$n$ or~$d$ are not perfect
for large~$d$ and~$n$ in Fig.~\ref{fig:estimatearbitstatearbitparas}.

\end{document}